%% file: main.tex
\documentclass[aap]{imsart}

\RequirePackage{amsthm,amsmath,amsfonts,amssymb}
\RequirePackage[numbers]{natbib}
\RequirePackage{graphicx}
\RequirePackage{textcomp}
\RequirePackage{xcolor}
\RequirePackage{booktabs}
\input{PackageLoad2021.tex}

\startlocaldefs
\input{Notation_header}
\input{acronyms}
\RequirePackage[colorlinks,citecolor=blue,urlcolor=blue]{hyperref}

\settoggle{Track}{true}

\newcommand{\addva}[1]{\addv{a}{off}{#1}}

\newcommand{\addvb}[1]{\addv{b}{off}{#1}}

\newenvironment{IEEEproof}[1][\proofname]{%
  \proof[\bfseries #1]%
}{\endproof}

\endlocaldefs

\begin{document}

\begin{frontmatter}
\title{On The Reliability Function of Discrete Memoryless Multiple-Access Channel with Feedback}
\runtitle{Reliability Function MAC-FB}

\begin{aug}
\author[A]{\fnms{Mohsen}~\snm{Heidari}\ead[label=e1]{mheidar@iu.edu}},
\author[B]{\fnms{Achilleas}~\snm{Anastasopoulos}\ead[label=e2]{anastas@umich.edu}}\\
\and
\author[B]{\fnms{S. Sandeep}~\snm{Pradhan}\ead[label=e3]{pradhanv@umich.edu}}
\address[A]{CS Department, Indiana University, Bloomington\printead[presep={,\ }]{e1}}
\address[B]{Department of EECS, University of Michigan, Ann Arbor\printead[presep={,\ }]{e2}\printead[presep={,\ }]{e3}}

\end{aug}

\begin{abstract}
 The reliability function of a channel is the maximum achievable exponential rate of decay  of the error probability as a function of the transmission rate.   In this work, we derive bounds on  the reliability function of discrete memoryless multiple-access channels (MAC) with noiseless feedback.    We show that our bounds are tight for a variety of MACs, such as $m$-ary additive and two independent point-to-point channels. The bounds are expressed in terms of a new information measure called ``variable-length directed information". The upper bound is proved by analyzing  stochastic processes defined based on the entropy of
the message, given the past channel's outputs. Our method relies on tools from the theory of martingales, variable-length information measures, and a new technique called time pruning. 
We further propose a  variable-length achievable scheme consisting of three phases: (i) data transmission, (ii) hybrid data-confirmation, and (iii) full confirmation. We show that  two-phase-type schemes are strictly suboptimal in achieving the  MAC's reliability function. Moreover, we study the shape of the lower-bound and show that it increases linearly with respect to a specific Euclidean distance measure defined between the transmission rate pair and the capacity boundary. As side results, we derive an upper bound on the capacity of MAC with noiseless feedback and study a new problem involving a hybrid of hypothesis testing and data transmission. 
\end{abstract}



\end{frontmatter}


\input{intro.tex}

\input{prelim.tex}

\input{CapOuter.tex}
\input{ErrExpResults.tex}

\input{shape_bounds.tex}

\section{Analysis and Proof Techniques}\label{sec:analysis}

The typical behavior of the entropy is shown in Fig.~\ref{fig:entropy12} (generated for a \ac{ptp} channel). The figure shows $H(W|y^t), t>0$, where $W$ is the message taking values from $[1:2^{10}]$, and $y^t$ is the channel output realization. 
For MAC, we study the drift of five different entropies involving the messages' individual, conditional, and joint entropy.  We derive the upper bound by analyzing the rate of the drift of these entropies. Particularly, we derive bounds on the slope of the linear and logarithmic drifts in terms of variable-length mutual information and relative entropy, respectively. For that, we use tools from the theory of martingales and variable-length information measures. 

\begin{figure}[hbtp]
\centering
\vspace{-10pt}
\includegraphics[scale=0.35]{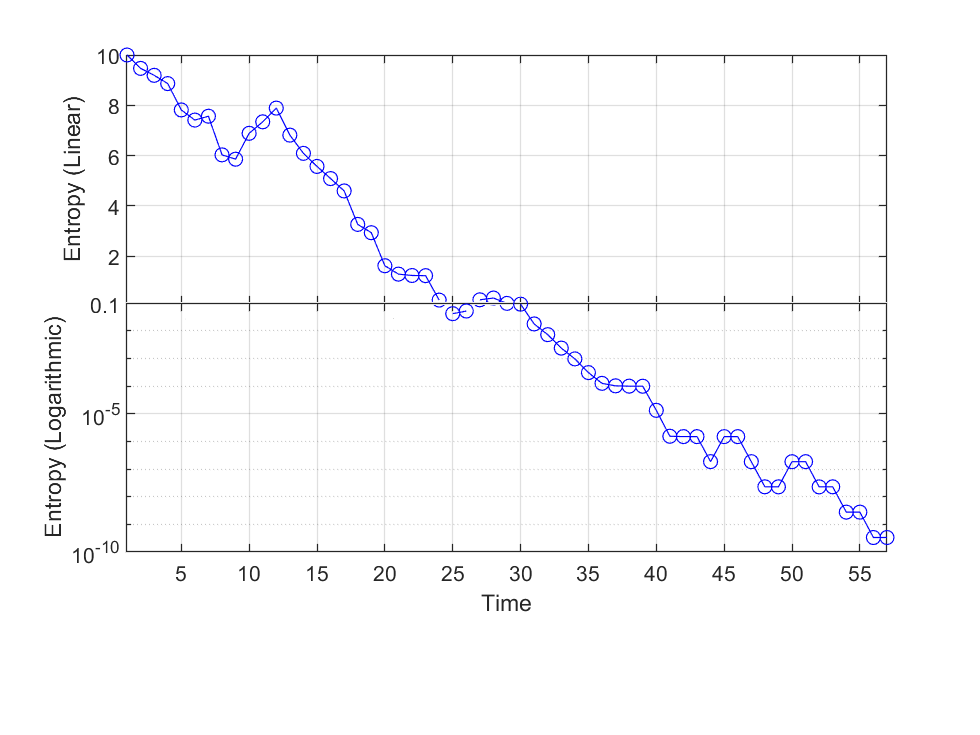}
\vspace{-35pt}
\caption{Realization of the entropy process decrease. The plot uses linear scale for values above a threshold (e.g., 0.1) and log scale for values below this threshold.}
\label{fig:entropy12}
\end{figure}

\input{DriftAnalysis.tex}

\section{Achievability Schemes}\label{sec:achievability}
\input{TwoPhaseLB}

\input{ThreePhaseLB}

\section{Conclusion}\label{sec: conclution}
This paper studies the bounds on the achievable rate and error exponent of variable length codes for communications over multiple access channels. The bounds are tight for a variety of channels. 
The upper bound on the reliability function is derived via analysis of the drift of the message entropy conditioned on the channel output. For the lower bound,  a three-phase achievable scheme is proposed.

\begin{appendix}
\input{VLExample}

\input{proofs/proof_thm_capacity.tex}

\input{SupportingHP}

\section{Proof of Lemma \ref{lem:pruned sub martingale 2}}\label{proof:lem:pruned sub martingale 2}
We first proceed with a simpler version of the lemma stated below. 
\input{proofs/proof_lem_pruned_submartingale.tex}
\subsection*{Proof of Lemma \ref{lem:pruned sub martingale 2}}
Now we are ready to prove Lemma \ref{lem:pruned sub martingale 2}.
\input{proofs/proof_lem_pruned_submartingale2.tex}

\input{proofs/proof_lem_logdrift.tex}
\section{Proof of Lemma \ref{lem:Err Exp conenction}}\label{proof:lem:Err Exp conenction}
\input{proofs/proof_lem_ErrExp_connection}
\input{proofs/Proof_lem_maxmartingale}
\input{proofs/proof_three_phase_lb}
\input{proofs/proof_lem_lb_shape}
\end{appendix}

\bibliographystyle{imsart-number}
\bibliography{main}


\end{document}

%% file: PackageLoad2021.tex
\usepackage{comment}

\usepackage{amsmath}
\usepackage{amsfonts}
\usepackage{amsthm}   
\usepackage{mathrsfs}
\usepackage{tikz}
\usepackage[utf8]{inputenc}

\usepackage{epstopdf}
\usepackage{mathtools}
\usepackage{graphicx}
\usepackage{bbm}
\usepackage{enumerate}
\usepackage{epsf,verbatim,amssymb,array,multicol,multirow}  
\usepackage{psfrag,bm,xspace}
\usepackage{hhline}
\usepackage{xstring}
\usepackage{ifthen}

\usepackage{physics}

\ifdefined \theorem 
\else
  \newtheorem{theorem}{Theorem}
\fi
\newtheorem{lem}{Lemma}
\ifdefined \corollary 
\else
\newtheorem{corollary}{Corollary}
\fi

\ifdefined \proposition 
\else
\newtheorem{proposition}{Proposition}
\fi

\ifdefined \definition 
\else
\newtheorem{definition}{Definition}
\fi

\ifdefined \example 
\else
\newtheorem{example}{Example}
\fi
\def\etal{\textit{et al.}}

\ifdefined \remark 
\else
\newtheorem{remark}{Remark}
\fi

\newtheorem{innercustomgeneric}{\customgenericname}
\providecommand{\customgenericname}{}
\newcommand{\newcustomtheorem}[2]{%
  \newenvironment{#1}[1]
  {%
   \renewcommand\customgenericname{#2}%
   \renewcommand\theinnercustomgeneric{##1}%
   \innercustomgeneric
  }
  {\endinnercustomgeneric}
}

\newcustomtheorem{customthm}{Theorem}
\newcustomtheorem{customlem}{Lemma}
\newcustomtheorem{customcor}{Corollary}

\makeatletter
\def\old@comma{,}
\catcode`\,=13
\def,{%
  \ifmmode%
    \old@comma\discretionary{}{}{}%
  \else%
    \old@comma%
  \fi%
}
\makeatother

\newcommand\numberthis{\addtocounter{equation}{1}\tag{\theequation}}

\usepackage{xcolor}
\definecolor{darkblue}{rgb}{0.1,0.1,0.8}
\definecolor{DarkGreen}{rgb}{0,0.6,0}
\definecolor{brickred}{rgb}{0.8, 0.25, 0.33}
\definecolor{britishracinggreen}{rgb}{0.0, 0.26, 0.15}
\definecolor{calpolypomonagreen}{rgb}{0.12, 0.3, 0.17}
\definecolor{ao(english)}{rgb}{0.0, 0.5, 0.0}
	\definecolor{cadmiumgreen}{rgb}{0.0, 0.42, 0.24}
\definecolor{burgundy}{rgb}{0.5, 0.0, 0.13}
\usepackage{etoolbox}

\providetoggle{Blue_revision}
\settoggle{Blue_revision}{true}

\providetoggle{Track}
\settoggle{Track}{true}
\newcommand{\addv}[3]{%
	\iftoggle{Track}{%
    	\IfEqCase{#1}{%
       	 	{a}{\ifthenelse{\equal{#2}{ON}}{{\color{cadmiumgreen}#3}}{#3}}%
        	{b}{\ifthenelse{\equal{#2}{ON}}{{\color{brickred}#3}}{#3}}%
       		{c}{\ifthenelse{\equal{#2}{ON}}{{\color{burgundy}#3}}{#3}}%
    	}[\PackageError{tree}{Undefined option to tree: #1}{}]%
	}{#3}%
}

 \usepackage{hyperref}

\hypersetup{
colorlinks=true, %
  pdfstartview={FitH},
    linkcolor=red,
    citecolor=blue,
    urlcolor={blue!80!black}
}

\usepackage{graphicx}

\newcounter{relctr} 
\everydisplay\expandafter{\the\everydisplay\setcounter{relctr}{0}} 

\AtBeginDocument{} 

%% file: Notation_header.tex
\global\long\def\RR{\mathbb{R}}

\global\long\def\NN{\mathbb{N}}

\global\long\def\ZZ{\mathbb{Z}}
\global\long\def\EE{\mathbb{E}}
\global\long\def\PP{\mathbb{P}}

\global\long\def\11{\mathbbm{1}}

\newcommand{\bfa}{\mathbf{a}}

\newcommand{\bfy}{\mathbf{y}}
\newcommand{\bfz}{\mathbf{z}}

\newcommand{\bfX}{\mathbf{X}}

\newcommand{\bfZ}{\mathbf{Z}}

\newcommand{\CAL}[1]{\mathcal{#1}} 

\newcommand{\CF}{\CAL{F}}

\newcommand{\CY}{\CAL{Y}}

\global\long\def\+{\oplus}

\newcommand{\prob}[1]{\PP\qty(#1) }
\def\<{\langle}
\def\>{\rangle}

\ifdefined \var 
  \renewcommand{\var}{\mathsf{var}}
\else
  \newcommand{\var}{\mathsf{var}}
\fi

\ifdefined \abs \else
 \newcommand{\abs}[1]{\lvert#1\rvert}
\fi

\ifdefined \norm \else
 \newcommand{\norm}[1]{\lVert#1\rVert}
\fi

\ifdefined \set 
  \renewcommand{\set}[1]{\left\{#1\right\}}
\else
  \newcommand{\set}[1]{\left\{#1\right\}}
\fi

\newcommand*{\medcap}{\mathbin{\scalebox{1}{\ensuremath{\bigcap}}}}%
\newcommand*{\medcup}{\mathbin{\scalebox{1}{\ensuremath{\bigcup}}}}%

\DeclareMathOperator*{\argmax}{arg\,max}

\usepackage{stackengine}
\def\deq{\mathrel{\ensurestackMath{\stackon[1pt]{=}{\scriptstyle\Delta}}}}

\def\PMAC{\mathcal{P}_{MAC}}
\def\CMAC{\mathcal{C}_{FB}^{VLC}}
\def\Tleps{{\tau_{\epsilon}}}
\def\Tueps{{\tau^{\epsilon}}}

\def\Tlepstld{{\tilde{\tau}_{\epsilon}}}
\def\Tuepstld{{\tilde{\tau}^{\epsilon}}}
\def\Tuepstldi{{\tilde{\tau}^{\epsilon,i}}}
\def\Tlepstldi{{\tilde{\tau}_{\epsilon}^i}}
\def\Zprune{L_n}
\def\Hbar{{\bar{H}}}
\def\Htld{{\tilde{H}}}
\def\J{{J}}
\def\D{{D}}

\def\Pdata{P}

\def\Eo{E_{o,\Pdata}}
\def\Elb{E_{l, \text{3-phase}}}

\newcommand{\xHT}[1]{\underline{x}_{#1, 3}}
\newcommand{\xDHT}[1]{x_{#1,2}}

\newcommand{\Qbar}[1]{\bar{Q}^{#1}}
\newcommand{\Qbarp}[1]{\bar{Q'}^{#1}}
\newcommand{\DQbar}[1]{D_{\Qbar{#1}}}
\newcommand{\DHT}[1]{\bar{D}_{\mathsf{P}_3}(00 \| #1)}
\newcommand{\ILP}[1]{I_L^{#1}(\Pdata)}
\def\xstar{x^*_{1,r}}
\def\Xstar{\xstar} 

\def\wstar{w^*_{1,r}}
\def\Wstar{\wstar}

%% file: acronyms.tex
\usepackage[nolist,nohyperlinks]{acronym}
\newacro{ptp}[PtP]{Point-to-Point}

\newacro{iid}[i.i.d.]{independent and identically distributed} 
\newacro{IID}[i.i.d.]{independent and identically distributed} 
\newacro{DMC}[DMC]{discrete memoryless channel}
\acrodefplural{DMC}{discrete memoryless channels}
\newacro{MAC}[MAC]{multiple access channel}
\acrodefplural{MAC}{multiple access channels}

\newacro{VLC}[VLC]{variable length code}
\acrodefplural{VLC}{variable length codes}

\newacro{wrt}[w.r.t]{with respect to}
\newacro{AVC}[AVC]{arbitrary varying channel}
\acrodefplural{AVC}{arbitrary varying channels}

%% file: intro.tex
\section{Introduction}
Channel capacity determines the maximum transmission rate for communications with vanishing error probability. The reliability function determines the maximum exponential decay rate of error probability as a function of the transmission rate. It measures the trade-off between  reliability and delay in communication. The reliability function has been studied extensively since the dawn of information theory \cite{Gallager,berlekamp1964block,Shannon1967,Burnashev,Csiszar2011}. However, its single-letter characterization for elementary setups such as \acp{DMC} is still an open problem. At this point, several bounds exist, such as the random-coding lower and sphere-packing upper bound. The results are promising for other settings where the channel output is available as feedback to the encoder. 
In a seminal work, Burnashev \cite{Burnashev} completely characterized the error exponent of \acp{DMC}  with noiseless and casual feedback. This characterization has a simple yet intuitive form:
\begin{align}\label{eq: E(R) error exponent }
E(R)=C_1\Big(1-\frac{R}{C}\Big) ,
\end{align}
where $R$ is the (average) transmission rate, $C$ is the channel's capacity, and $C_1$ is  the maximal relative entropy between conditional output distributions. Intuitively, $C_1$ is the maximum exponent for binary hypothesis testing over the channel. The Burnashev's exponent can significantly exceed the sphere-packing exponent for no-feedback communications as it approaches capacity with a nonzero slope. These improvements are obtained using \acp{VLC}, where the communication length depends on the channel's realizations.  

Yamamoto and Itoh \cite{Yamamoto} introduced a \ac{VLC} scheme achieving the error exponent in \eqref{eq: E(R) error exponent }. Their scheme consists of two distinct transmission phases, the data and the confirmation phase. In the data phase, the message is encoded using a capacity-achieving fixed block length code.  During the confirmation phase, the transmitter (who knows if the codeword was correctly decoded at the receiver through noiseless feedback)
sends one bit of information to the receiver indicating if the decoded message is correct. The decoder performs a binary hypothesis test to decide if $0$ or $1$ was transmitted.

The reliability function has been studied for various communication setups such as channels with states and/or memory \cite{Como2009,Tchamkerten2005,Tchamkerten2006,Tatikonda2009,Heidari2022}, multi-user channels with or without feedback \cite{Gallager,Nazari2014,Nazari2015,Heidari2018,Farkas2014}, and other setups \cite{Tamir2021}.  For a comprehensive survey of the reliability function, see \cite{Haroutunian2007} and references therein.

Among such setups, this paper studies the reliability function of \acp{MAC} with noiseless feedback. There are several bounds on the reliability function of \ac{MAC} without feedback \cite{Nazari2015,Csiszar2021}. Given the promising derivation of Burnashev for \ac{DMC}, the reliability function of \ac{MAC} with feedback is an important problem \cite{Heidari2018}.

In the context of communications over \ac{MAC}, the benefits of feedback are more prominent than \ac{DMC}. For instance, Gaarder and Wolf \cite{Gaarder-Wolf} showed that feedback  expands the capacity region even using fixed-length codes. Willems \cite{Willems-FB} derived the feedback-capacity region for a class of MACs, and Kramer \cite{Kramer-thesis} derived a multi-letter characterization of the feedback-capacity region of two-user MAC with feedback using fixed-length codes. The expressions are given in terms of  \textit{directed information}, introduced by Massy \cite{Massey1990} and Marko \cite{Marko1973}. Single-letter characterizations of the capacity region and the error exponent for general MACs are still open problems. 

\subsection{Summary of the Contributions}
While the focus in the literature  has been on fixed-length coding schemes for communications over MAC, in this paper, we study capacity and reliability function with \acp{VLC} for communications over MAC. We derive lower and upper bounds on the reliability function of \ac{MAC} with feedback and provide an upper bound on its feedback capacity when \acp{VLC} are permitted. Below we highlight some of the main contributions of this work.  

\paragraph*{MAC Feedback Capacity} We introduce an outer bound on the feedback capacity of \acp{VLC} for communications over MAC. We formulate a new information measure for variable-length sequences of random variables called ``variable-length directed information." Using this information measure, we derive our outer bound on the capacity region  and show that it subsumes existing bounds of Kramer \cite{Kramer-thesis}, Tatikonda and Mitter \cite{Tatikonda2009}, and Permuter \etal \cite{Permuter2009}  where fixed-length codes are studied.

\paragraph*{Upper bound on the reliability function}
We derive an upper bound on the reliability function of two-user \ac{MAC} with noiseless feedback. The bound is expressed in terms of our ``variable-length" mutual information and relative entropy. 
As a corollary, we show that this bound reduces to Burnashev's expression when the channel consists of two parallel DMCs.  Our approach relies on analyzing the entropy of the stochastic process defined based on the entropy of the message given the past channel's output. This method has been studied in our earlier works  \cite{Heidari2018,Heidari2022}. In this work, we take a more comprehensive approach by considering five different entropies, involving individual, conditional, and joint entropy of the messages for the two-user MAC problem. 

The entropies at time zero are equal to the logarithm of the message sizes. As time goes on with more channel outputs, the entropies decrease linearly (in expectation) to a sufficiently small threshold $\epsilon>0$. After a transient phase, the entropies drift  logarithmically. We derive the upper bound by analyzing the rate of the drift of these entropies. We bound the slope of the linear and logarithmic drifts in terms of variable-length mutual information and relative entropy, respectively. For that, we use tools from the theory of martingales and variable-length information measures. 

One of the major technical challenges in our proofs is the transient period from linear to logarithmic drifts. We address the transient period by proposing a new technique called ``time pruning". In the conventional stochastic process $\{X_t\}_{t>0}$, time $t$ increases linearly by one unit at each step. With the time-pruning technique, we allow the time to increase randomly as a stopping time with respect to the process itself.  As a result, we obtain a random process where multiple samples $X_t$ can appear at each step. This novel technique is critical in our analysis as it allows for jumping to the future. 

\paragraph*{A three-phase achievable scheme}
Lastly, we propose an achievable scheme, thus establishing a lower bound on the error exponent. Our scheme is a  variable-length achievable scheme consisting of three phases: (i) data transmission, where both encoders  send the messages; (ii) hybrid data confirmation, where one encoder starts the confirmation while the other is still in the data transmission
stage; and (iii) full confirmation, where both encoders send one bit of confirmation message to the receiver. As a side result, we  study a new problem involving a hybrid of hypothesis testing and data transmission. 

We show that while the two-phase scheme of  Yamamoto and Itoh \cite{Yamamoto} is optimal for \ac{DMC}, its natural extension to MAC is sub-optimal in achieving the  MAC's reliability function. 
By studying a concrete example, we show the necessity of a three-phase scheme allowing the encoders to stop transmission at different times.

\paragraph*{Geometric study of the lower bound} Finally, we study the geometry of the lower bound. We show that it increases linearly with respect to a specific Euclidean distance measure defined between the transmission rate pair and the capacity boundary. For a more intuitive visualization of the reliability function, we derive a simplified version of our lower bound that takes the form 
\begin{align*}
   E(R_1, R_2)\geq  D_{lb} \Big(1-\frac{\|R\|}{\mathcal{C}(\theta_{R})} \Big),  
\end{align*}
where $(\|R\|, \theta_R)$ represents the polar coordinates of the rate pair  $(R_1, R_2)$ in $\RR^2$
plane and $\mathcal{C}(\theta_{R})$ is the point at the intersection of the capacity boundary and the line crossing the origin with angle $\theta_{R}$.

\subsection{Organization of the paper}
 The paper is organized as follows: Section \ref{sec: problem formulation} provides basic definitions and the problem formulation. In Section \ref{sec:main results} we present the main results on the feedback capacity of MAC using VLCs and our bounds on the reliability function. In Section \ref{sec: the shape of the bounds}, we study the tightness of the bounds and present a geometric study of the lower bound. Section \ref{sec:analysis} presents the proofs and analysis of the upper bound. Section \ref{sec:achievability} presents our coding scheme and its analysis for the lower bound. Lastly, Section \ref{sec: conclution} concludes the paper.

%% file: prelim.tex
\section{Definitions and Problem Formulation}\label{sec: problem formulation}



\subsection{Casual conditioning}
Consider a sequence of random variables $(X_t, Y_t, Z_t), t \in [n]$. The casually conditional entropy of $X^n$ given $Y^n$ is defined as \cite{Kramer-thesis} 
\begin{align*}
    H(X^n || Y^n) \deq \sum_{t = 1}^n H(X_t | X^{t-1},  Y^t),
\end{align*}
where $H$ denotes Shannon entropy. 
We use \textit{directed information} and \textit{conditional directed information} as defined in  \cite{Kramer-thesis}. The directed information from a sequence ${X}^n$ to a sequence ${Y}^n$ when causally conditioned on ${Z}^n$ is defined as 
\begin{align}\label{eq:directed MI}
I({X}^n \rightarrow {Y}^n|| {Z}^n) &\deq H(X^n||Z^n) - H(X^n || Y^n, Z^n) =  \sum_{i=1}^n I(X^{i}; Y_i| Y^{i-1}, Z^i).
\end{align}

\subsection{Variable-length information measures}
In this paper, we consider a set of information measures suitable for variable-length sequences of random variables. Consider the stochastic process $\set{Y_t}_{t>0}$ and let $T$ be a stopping time with respect to the filtration $\mathcal{F}_t, t>0$ generated by $Y^t$. 
We study the random-variable $Y^T$ that takes values from the set $\medcup_{t>0}\CY^t$.  To picture $Y^T$, let us extend $\CY$ by adding a dummy symbol as $\CY\medcup\{\zeta\}.$ Then, $Y^T$ is equivalent to a sequence that equals to $Y_t$ for all $t\leq T$ and $\zeta$ for $t>T$. 
The entropy of the (variable-length) random vector $Y^T$ equals
\begin{align}\label{eq:variable entropy}
    H(Y^T) =  \EE \Big[\sum_{t=1}^T H(Y_t | \CF_{Y, t-1})\Big],
\end{align}
where the expectation is taken over the filtration $\CF_{Y, t}, t\in \NN$, the $\sigma$-algebra generated from $Y^{t}$. Moreover,
$H(Y^T)=H(Y^T,T)=H(T)+H(Y^T|T)$. See an illustrative example in Appendix \ref{app:VL example}.

Consider a sequence of random variables $\set{(X_t, Y_t, Z_t)}_{t>0}$ and let $T$ be a stopping time with respect to the filtration $\mathcal{F}_t, t>0$ generated by $Y^t$ with $\EE[T]$ being finite. Suppose $(X_t, Y_t, Z_t)$ takes values from the finite set $\mathcal{X}\times \mathcal{Y}\times \mathcal{Z}$ for all $t>0$. 
 The entropy of $X^T$ conditioned on $Y^T$ is given by  
\begin{align}\label{eq:var cond entropy}
 H( X^T | Y^T) = \sum_{t\geq 0} \prob{T=t} H(X^{t} | Y^{t}, T=t ).
\end{align}
Using the entropy,  mutual information is defined as  
    $I(X^T; Y^T)  =  H(X^T) - H(X^T | Y^T)$.


We next define the casual conditioning and the variable-length directed information between sequences of random variables with random lengths.

\begin{definition} \label{def:VdI}
Given the pair $(X^T, Y^T)$, the variable-length  entropy of $X^T$ casually conditioned on $Y^T$ is defined as 
\begin{align*}
H(X^T \| Y^T) &\deq \EE\Big[\sum_{t=1}^T H(X_t |X^{t-1}, \CF_{Y, t}) \Big].
\end{align*}
Also, the variable-length directed information from $X^T$ to $Y^T$ is defined as 
\begin{align}\label{eq:directed MI T}
I({X}^T \rightarrow {Y}^T \| {Z}^T) \deq H(X^T) - H(X^T\| Y^T) = \EE\Big[ \sum_{t=1}^T I({X}^t; {Y}_t~ |\mathcal{F}_{Y, t-1})\Big].  
\end{align}
See an illustrative example in Appendix \ref{app:VL example}.

\end{definition}


%

\subsection{Problem formulation}
Consider a discrete memoryless MAC with input alphabets $\mathcal{X}_1,\mathcal{X}_2$, and output alphabet $\mathcal{Y}$. The channel conditional probability distribution is denoted by $Q(y|x_1, x_2)$ for all $(y, x_1, x_2)\in \mathcal{Y}\times \mathcal{X}_1\times \mathcal{X}_2$. Such setup is denoted by $(\mathcal{X}_1,\mathcal{X}_2, \mathcal{Y}, Q)$.  Let $y^t$ and $x_{i}^t$, $i=1,2,$ be the channel output and the input sequences after $t$ uses of the channel, respectively. Then, the following condition is satisfied:
\begin{align}\label{eq: chann probabilities}
P(y_t|y^{t-1}, x_1^{t-1},x_2^{t-1})=Q(y_t|x_{1t}, x_{2t}).
\end{align}
We assume the channel output is available at the encoders with one delay unit through noiseless feedback. 
\begin{definition}
An $(M_1, M_2, N)$-variable-length code (VLC) for a MAC $(\mathcal{X}_1,\mathcal{X}_2, \mathcal{Y}, Q)$ with feedback is defined  by 
\begin{itemize}
 \item A pair of messages $W_1, W_2$ selected randomly, independently and with uniform distribution from  $ \{1,2, \dots, M_i\}, i=1,2$.


\item Two sequences of encoding functions 
$e_{i,t}: \{1,2, \dots, M_i\} \times \mathcal{Y}^{t-1} \rightarrow \mathcal{X}_i$, 
$t\in \NN, ~ i=1,2,$
one for each transmitter. 
\item A sequence of decoding functions 
$d_t: \mathcal{Y}^t \rightarrow \{1,2, ... , M_1\} \times \{1,2, ... , M_2\}$, $t\in \NN$.
\item  A stopping time $T$ with respect to (w.r.t.) the filtration $\mathcal{F}_t$, where $\mathcal{F}_t$ is the sigma-algebra generated by the channel output sequence at time $Y^t$ for $t\in \NN$. Furthermore, it is assumed that $T$ is almost surely bounded as $T \leq N$.
\end{itemize}

\end{definition}
For each $i=1,2$, given a message $W_i$, the $t$th output of Transmitter $i$ is denoted by $X_{i,t}=e_{i,t}(W_i, Y^{t-1})$. Let $(\hat{W}_{1,t}, \hat{W}_{2,t})=d_t(Y^t)$ represent the ``provisional'' decoded messages at time $t$. Then, the decoded messages at the decoder are denoted by $\hat{W}_1=\hat{W}_{1,T}$, and $\hat{W}_2=\hat{W}_{2,T}$. In what follows, for any $(M_1, M_2, N)$ VLC, we define the average rate-pair, error probability, and error exponent. 
Average rates for an $(M_1, M_2, N)$ VLC are defined as  
$$R_i \triangleq \frac{\log_2 M_i}{\EE [T]},\quad  i=1,2.$$ 
The probability of error is defined as 
$P_e \triangleq P\left((\hat{W}_1,\hat{W_2} )\neq (W_1,W_2) \right).$  
The error exponent of a VLC with the probability of error $P_e$ and stopping time $T$ is defined as 
$$E\triangleq-\frac{\log_2 P_e}{\EE[T]}.$$

In this paper, we study a class of $(M_1, M_2, N)$-VLCs for which the parameter $N$ grows polynomially with $\log M_1M_2,$ that is $N\leq (\log M_1M_2)^m$, where $m$ is a fixed integer. For example, a sequence of $(M_1^{(n)}, M_2^{(n)}, N^{(n)})$-VLCs, $n\geq 1,$ where $$M_1^{(n)} = 2^{nr_1}, \quad M_2^{(n)}= 2^{nr_2}, \quad  N^{(n)} \leq n^m,$$
and  with fixed parameters $r_1,r_2, m>0$, satisfies this condition.

\addva{\begin{definition}\label{def:ErrExponent}
A triplet of reliability and rates $(E,R_1,R_2)$ is said to be achievable for a given MAC if for all $\varepsilon\in (0,1)$ and all $M\in \NN$  there exists an $(M_1^{(n)},M_2^{(n)}, N^{(n)})$-VLC for sufficiently large $n$,
such that
 \begin{align*}
 -\frac{\log P^{(n)}_e}{\EE[T^{(n)}]} \geq E-\varepsilon, \qquad \frac{\log M^{(n)}_i}{\EE[T^{(n)}]} \geq R_i-\varepsilon,\qquad  M_i^{(n)}\geq M, \qquad  ~i=1,2, 
 \end{align*}
 and $N^{(n)}\leq n^m$,  where $m$ is fixed. The reliability function 
 $E(R_1,R_2)$ of a MAC with feedback is defined as the supremum of all reliability $E$ such that  $(E,R_1,R_2)$ is achievable. 
\end{definition}}

Given a MAC, and any integer $N$, let $\PMAC^N$ be the set of all $N$-letter distributions on $\mathcal{X}_1^N\times \mathcal{X}_2^N\times \mathcal{Y}^N$ that factor as  
\begin{align}\label{eq: factorization of P}
P_{X_1^N, X_2^N, Y^N} = \prod_{\ell=1}^N P_{1,\ell}(x_{1,\ell}| x_1^{\ell-1} y^{\ell-1})P_{2,\ell}(x_{2,\ell}| x_2^{\ell-1} y^{\ell-1})Q(y_\ell|x_{1, \ell}x_{2,\ell}),
\end{align}
where $Q$ is the transition matrix of the channel. For a more compact representation, define $\PMAC$, without the superscript, as the union $\medcup_N\PMAC^N$. For any $P\in \PMAC$, denote by $N_P$ as the letter size of $P$.

%% file: CapOuter.tex
\section{Main Results}\label{sec:main results}

\subsection{ Outer bound on the capacity region of 
MAC}

In this section, we study the feedback capacity for communications over a MAC using VLCs, and provide an upper bound on the capacity region. We start with the definition of achievable rates.
\addva{\begin{definition}\label{def:capcity}
A rate pair $(R_1,R_2)$ is said to be achievable for a given MAC if for all $\varepsilon\in (0,1)$, there exists an $(M_1^{(n)},M_2^{(n)}, N^{(n)})$-VLC for sufficiently large $n$ such that
 \begin{align*}
P^{(n)}_e \leq \varepsilon, \qquad \frac{\log M^{(n)}_i}{\EE[T^{(n)}]} \geq R_i-\varepsilon,  \qquad ~i=1,2, 
 \end{align*}
 and $N^{(n)}\leq n^m$, where $m$ is fixed. The feedback capacity of VLCs, denoted by $\mathcal{C}_{FB}$, is the convex closure of all achievable rates.
\end{definition}}


Next, we define a set of rate-pairs denoted by $\CMAC$, and show that it contains the capacity region, thus forming an outer bound. For any $P^N\in \PMAC^N$, let $T$ be a stopping time  with respect to the filtration $\mathcal{F}_t$ with $T\leq N$ almost surely. 
Using this notation, we define $\CMAC$ as follows

\begin{align*}
\CMAC : = cl\left[\bigcup_{P \in \PMAC} \bigcup_{T: T\leq N_P} \left\{ (R_1,R_2)\in [0, \infty)^2 : \begin{array}{rl}
  R_1&\leq \frac{1}{\EE[T]} I(X_{1}^T \rightarrow Y^T \|~ X_2^T)\\
  R_2&\leq  \frac{1}{\EE[T]} I(X_{2}^T \rightarrow Y^T \|~ X_1^T)\\
   R_1+R_2&\leq   \frac{1}{\EE[T]} I(X_{1}^T, X_2^T \rightarrow Y^T )
  \end{array}\right\}\right],
\end{align*}
where $cl$ denotes the \textit{convex closure} operation. 

\begin{theorem}\label{thm:VLC Capacity}
$\CMAC$ is an outer bound on the capacity region of VLCs for communications over the discrete memoryless MAC with noiseless feedback, i.e., $\mathcal{C}_{FB} \subseteq \CMAC$.
\end{theorem}
A proof is provided in Appendix \ref{proof:thm: VLC Capacity}. This theorem generalizes the results on the feedback capacity of MAC with fixed-length codes \cite{Kramer-thesis} to that with VLCs.
%
%
%
In what follows, we present another characterization of the outer bound on the capacity region. This characterization is via supporting hyperplanes, as in \cite{Salehi1978}, and is useful to characterize our outer bound on the error exponent. 


\begin{theorem}\label{thm:capacity C lambda}
$\CMAC$ equals to the set of rate pairs $(R_1, R_2)$ such that for all $\lambda_1, \lambda_2, \lambda_3 \in [0,1]$ with $\lambda_1+\lambda_2+\lambda_3=1$,   the following inequality holds
\begin{align*}
\lambda_1 R_1 + \lambda_2 R_2+\lambda_3(R_1+R_2) \leq \mathcal{C}_{\underline{\lambda}},
\end{align*}
where $C_{\underline{\lambda}} $ is the supremum of
\begin{align*}
 \frac{1}{\EE[T]}\Big(\lambda_1 I(X^T_1 \rightarrow Y^T|| X^T_2~| Q)+\lambda_2 I(X^T_2 \rightarrow Y^T|| X^T_1~| Q)+\lambda_3 I(X^T_1 X^T_2 \rightarrow Y^T | Q)\Big),
\end{align*}
over all distributions $P_{QX_1^NX_2^NY^N}$ for some $N\in \NN$ such that $P_{X_1^NX_2^NY^N|Q=q} \in \PMAC^N, \forall q\in\mathcal{Q}$, where $|\mathcal{Q}|=3$, and $T$ is a stopping time with respect to the filtration $\mathcal{F}_t$ with $T\leq N$ almost surely.
\end{theorem}
A proof of the theorem and detailed discussion are provided in Appendix \ref{app:CFB supporting HYP}.

%% file: ErrExpResults.tex
\subsection{Upper bounds on the error exponents}\label{subsec:threephase up}
Before proceeding with the results we establish some additional notation. For  any $P \in \PMAC$, let $\Qbar{1}_r\deq P_{Y_r|X_{1, r},Y^{r-1}}$ be the effective channel from the first user's perspective at time $t=r$. Then, for  any $y^{r-1}\in \mathcal{Y}^{r-1}$, define
\begin{align*}
    \D_r^1(y^{r-1})  \deq \max_{x_1, x'_1}\sup_{P_{X'_2|Y^{r-1}}} D_{KL}\Big( \Qbar{1}_r(\cdot | x_1, y^{r-1})~\|~ \Qbarp{1}_r(\cdot | x'_1, y^{r-1}) \Big),
\end{align*}
where $\Qbarp{1}_r$ is an effective channel of the first user  with the distribution of the second user chosen as $P_{X'_2|Y^{r-1}}$. Similarly, let $\Qbar{2}_r\deq P_{Y_r|X_{2, r},Y^{r-1}}$ be the effective channel from the second user's perspective at time $t=r$, and define 
\begin{align*}
    \D_r^2(y^{r-1}) \deq \max_{x_2, x'_2}\sup_{P_{X'_1|Y^{r-1}}} D_{KL}\Big( \Qbar{2}_r(\cdot | x_2, y^{r-1})~\|~ \Qbarp{2}_r(\cdot | x'_2, y^{r-1}) \Big).
\end{align*}
Moreover, define 
\begin{align*}
    \D^3 \deq \max_{x_1, x'_1}\max_{x_2, x'_2} D_{KL}\Big( Q(\cdot | x_1, x_2)~\|~ Q(\cdot | x'_1, x'_2) \Big).
\end{align*}
With this notation, we are ready to present the main result of this section with the detailed proofs given in Section \ref{sec:analysis}. 

\begin{theorem}\label{thm:upper bound three phase}
The reliability function of a MAC channel is bounded from above as 
\begin{align*}
E(R_1,R_2) \leq 
E^{ub}(R_1,R_2) :=
\sup_{P \in \PMAC} \sup_{\substack{ T:\\  T\leq N_P}} \sup_{\substack{ T_1,T_2, T_3 :\\ T_i\leq T} }  \min_{i\in \{1,2,3\}} \left\{D^i(P) \Big(1 - \frac{R_i}{I^i(P)}\Big)\right\},
\end{align*}
where $R_3=R_1+R_2$, and $I^i(P)$ and $D^i(P)$ are defined as
\begin{align*}
I^1(P) &= \frac{1}{\EE[T_1]}    I(X_1^{T_1} \rightarrow Y^{T_1}|| X_2^{T_1}), \qquad    &     D^1(P) &=\frac{1}{\EE[T-T_1]} \EE\Big[ \sum_{r=T_1}^T   \D_r^1(Y^{r-1})   \Big] \\
I^2(P) &= \frac{1}{\EE[T_2]}    I(X_2^{T_2} \rightarrow Y^{T_2}|| X_1^{T_2} ), \qquad    &     D^2(P) &=\frac{1}{\EE[T-T_2]} \EE\Big[ \sum_{r=T_2}^T   \D_r^2(Y^{r-1})   \Big] \\
I^3(P) &= \frac{1}{\EE[T_3]}    I(X_1^{T_3}X_2^{T_3} \rightarrow Y^{T_3}),    \qquad    &     D^3(P) &=D^3.  
\end{align*}
\end{theorem}
This upper bound is proved by analyzing an arbitrary coding scheme in three stages. 
Roughly speaking, these stages correspond with the rate of decrease of appropriately defined entropy functions. Specifically, in the first stage, entropies decrease linearly, while in the second, exponentially. The boundary between these two behaviors can be different for the two users and is quantified by the stopping times $T_i$ for user $i$.
Next, we present a simplified version of the upper bound as a corollary.  
\begin{corollary}[Two-Phase Upper Bound]\label{cor:twophase up}
The reliability function of a MAC channel is bounded from above as
\begin{align}\label{eq: E_u}
E(R_1,R_2) \leq E^{ub}_{\text{two-phase}}(R_1,R_2) :=\sup_{P \in \PMAC} \sup_{\substack{ T:\\  T\leq N_P}} \sup_{\substack{ T_1,T_2, T_3 :\\ T_i\leq T} }  D_{ub} \min_{i\in \{1,2,3\}} \Big\{ 1 - \frac{R_i}{I^i(P)}\Big\},
\end{align}
where $D_{ub}=\max_{\substack{x_1, x'_1\in \mathcal{X}_1,\\ x_2,x'_2 \in \mathcal{X}_2}} D_{KL}\Big(Q(\cdot | x_1,x_2) \| Q(\cdot | x'_1, x'_2) \Big)$.
\end{corollary}
The proof follows from Theorem \ref{thm:upper bound three phase} and the fact that $D_{ub}\geq D^i(P)$ for any $i=1,2,3,$ and $P\in \PMAC$. 


\subsection{Lower bounds on the error exponent}\label{subsec:LB}
In this section, we present our lower bounds on the reliability function.  For any $L\in \NN$ and $P\in \PMAC$, and rates $(R_1, R_2)$, define
\begin{align}\label{eq:Eo}
\Eo(R_1, R_2):=1- \max\set{\frac{R_1}{\ILP{1}},\frac{R_1}{\ILP{2}},\frac{R_1+R_2}{\ILP{3}}}.
\end{align}

For a PMF $P_{X_2} \in \mathcal{P}(\mathcal{X}_2)$, define
 $\Qbar{1}(y|x_1) := \sum_{x_2}P_{X_2}(x_2) Q(y|x_1, x_2)$ as the single-letter effective channel from the first user's perspective. Similarly, define $\Qbar{2}$. 
 Define for any $x_j \in \mathcal{X}_j$,
$$\DQbar{j}(x_j, P_{X_{\bar{j}}}):=\sup_{z \in \mathcal{X}_j} D_{KL}(\Qbar{j}(x_j) \|\Qbar{j}(z)),$$
where $j=1,2$ and $\bar{j}=3-j$.
  For $x_1\in \mathcal{X}_1$, define $C_2(x_1)=I(X_2; Y|X_1=x_1)$. Similarly, define  $C_1(x_2)$. 
Let $Z= (Z_1(0), Z_2(0), Z_1(1), Z_2(1))$ be a vector of random variables on $\mathcal{X}_1\times \mathcal{X}_2\times\mathcal{X}_1\times\mathcal{X}_2$  with joint distribution $P_{{\bfZ}}$. 
Define 
for all $a, b \in \{0,1\}$, 
\begin{align}\label{eq:dpar ij}
\bar{D}_{P_{{\bfZ}}}(00\| ab):= \EE_{\bfZ\sim P_{{\bfZ}}}\Big[ D_{KL}\Big(Q\big(\cdot |Z_1(0), Z_2(0)\big) \| Q\big(\cdot|Z_1(a), Z_2(b)\big)\Big)\Big].
\end{align}

\begin{theorem}\label{thm: Error Exp lowerbound three phase}
The following is a lower-bound for the reliability function  of any discrete memoryless MAC:
\begin{align}\label{eq: E_l three}
E^{lb}(R_1,R_2) = \max\set{\Elb^1(R_1,R_2), \Elb^2(R_1,R_2)},
\end{align}
where, 
\begin{align*}
\Elb^1(R_1,R_2) =  &\sup_{\substack{x_1\in \mathcal{X}_1 \\0\leq \gamma_2<1}}\sup_{P_{{\bfZ}}} \sup_{\substack{P_{X_2}\\ P \in \PMAC}} \min\Big\{
 \big(\Eo(R_1, R_2-\gamma_2C_2(x_1))-\gamma_2\big)\bar{D}_{P_{Z}}(00\| 01),\\
& ~~ \big(\Eo(R_1, R_2-\gamma_2C_2(x_1)-\gamma_2\big) \bar{D}_{P_{Z}}(00\|10)+\gamma_2\DQbar{1}(x_1, P_{X_2}),\\\numberthis\label{eq:El 1}
& ~~ \big(\Eo(R_1, R_2-\gamma_2C_2(x_1)-\gamma_2\big) \bar{D}_{P_{Z}}(00\| 11)+\gamma_2\DQbar{1}(x_1, P_{X_2})\Big\},
\end{align*}
and
\begin{align*}
\Elb^2(R_1,R_2) =  &\sup_{\substack{x_2\in \mathcal{X}_2 \\0\leq \gamma_2<1}}\sup_{P_{{\bfZ}}} \sup_{\substack{P_{X_1}\\ P \in \PMAC}}  \min\Big\{
 \big(\Eo(R_1-\gamma_2C_1(x_2), R_2)-\gamma_2\big) \bar{D}_{P_{Z}}(00\| 01)\\
 &~~ +\gamma_2\DQbar{2}(x_2, P_{X_1}),  \big(\Eo(R_1-\gamma_2C_1(x_2), R_2)-\gamma_2\big) \bar{D}_{P_{Z}}(00\| 10),\\\numberthis\label{eq:El 2}
&~~ \big(\Eo(R_1-\gamma_2C_1(x_2), R_2)-\gamma_2\big) \bar{D}_{P_{Z}}(00\| 11)+\gamma_2\DQbar{2}(x_2, P_{X_1})\Big\}.
\end{align*}
\end{theorem}
A proof is given in Appendix \ref{proof: Error Exp lowerbound three phase}.
Note that we derive this theorem by proposing a coding scheme consisting of three phases: (1) data transmission, in which a capacity-achieving code is used, (2) hybrid data-confirmation stage, in which one user stops transmission and sends one bit of confirmation message while the other user continues the transmission, and (3) final confirmation, where both users send a one-bit confirmation message.   
As a special case, we present a simplified inner bound with a two-phase coding scheme. 
\begin{corollary}[Two-Phase Lower Bound]\label{cor:twophase lb}
The following is a lower-bound for the reliability function  of any discrete memoryless MAC:
\begin{align}\label{eq: E_l}
E^{lb}_{\text{two-phase}}(R_1,R_2) = \sup_{P\in \PMAC}  D_{lb}~\Eo(R_1, R_2),
\end{align}
where, 	
\begin{align}\label{eq: D_l}
D_{lb}\deq \sup_{P_{Z}}\min_{ab\neq 00} \bar{D}_{P_{Z}}(00\| ab).
\end{align}
\end{corollary}
The result is obtained by  setting $\gamma_2=0$ in Theorem~\ref{thm: Error Exp lowerbound three phase}.

%% file: shape_bounds.tex
\section{On The Shape and Tightness of The Bounds}\label{sec: the shape of the bounds}
\subsection{Geometric characterization}
In this Section, we study the tightness of the upper bound in Theorem \ref{thm:upper bound three phase} and the lower bound in Theorem \ref{thm: Error Exp lowerbound three phase}. Furthermore, we provide an alternative representation for the lower bound. 
We start with  the  two-phase lower bound in Corollary~\ref{cor:twophase lb}. 
For that, suppose $(R_1,R_2)$ is a point inside the capacity region $\mathcal{C}$. Denote by $(\|R\|, \theta_{R})$ the polar coordinates of $(R_1,R_2)$ in $\RR^2$. Let $\mathcal{C}(\theta_{R})$ be the point at the intersection of the capacity boundary and the line crossing the origin with angle $\theta_{R}$, as shown in Fig. \ref{fig: R' and R}. 
\begin{figure}[hbtp]
\centering
\includegraphics[scale=0.8]{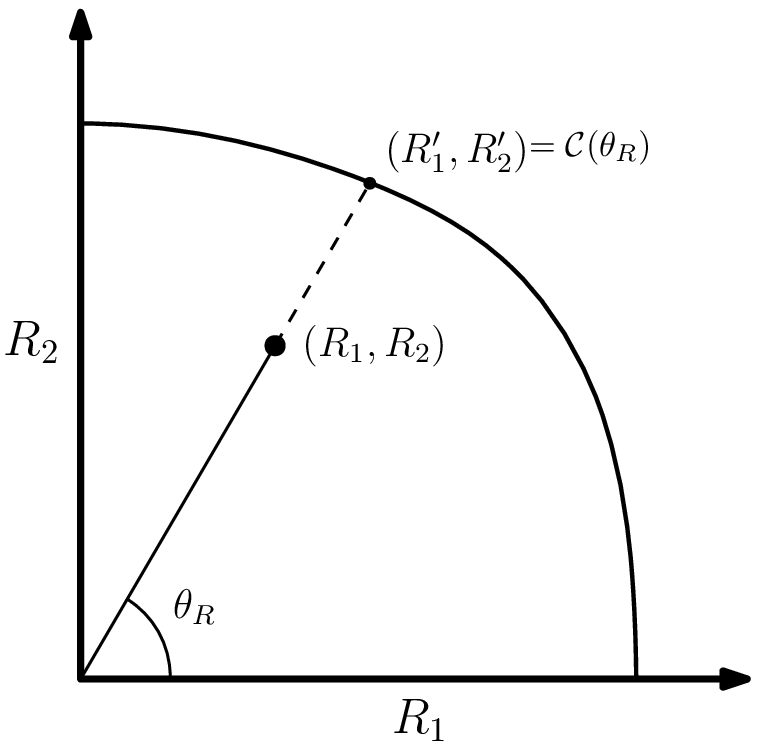}
\hspace{0.5in}
\includegraphics[scale=0.35]{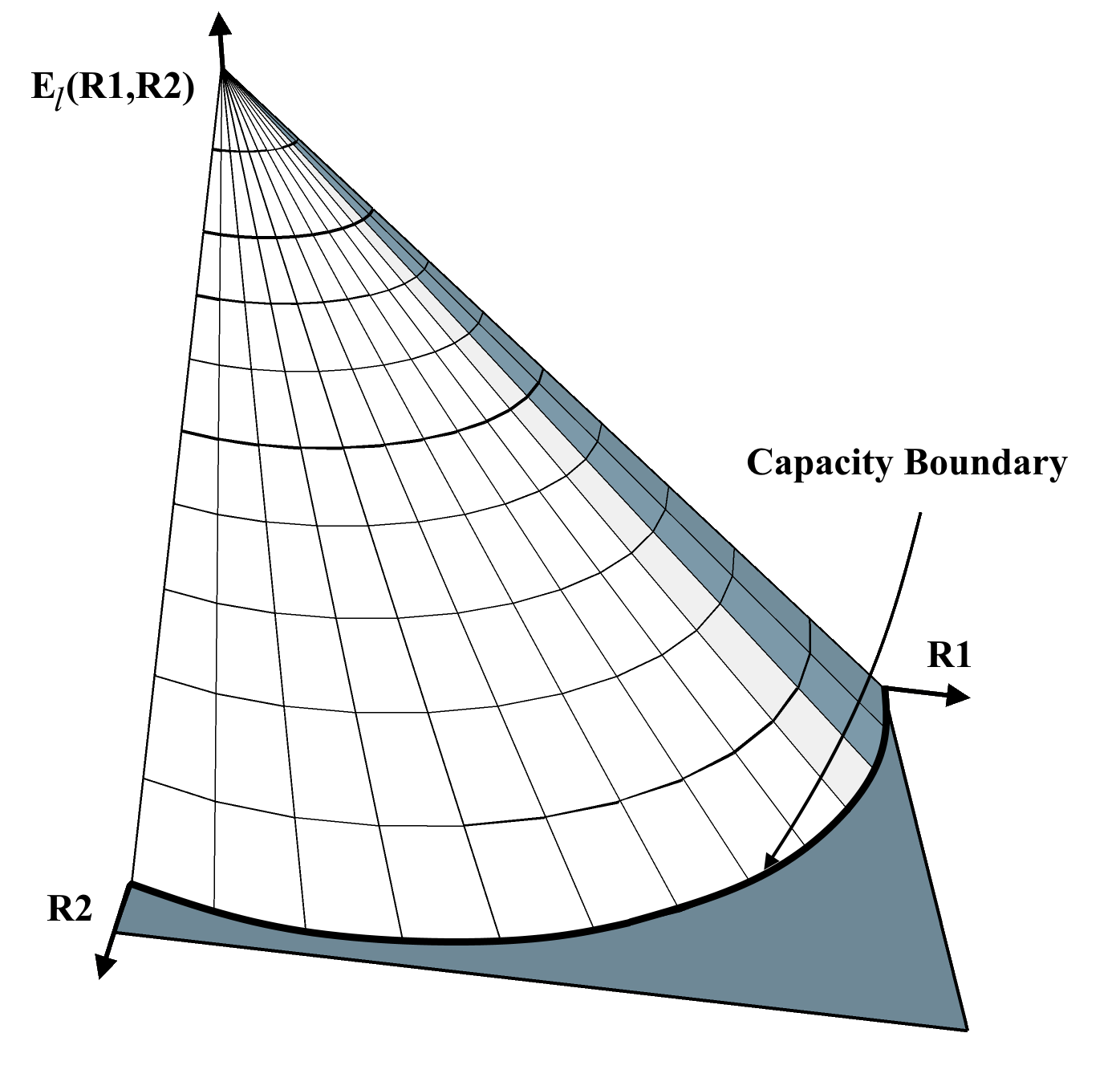}
\hspace{0.4in} \large{(a) \hspace{2.5in} (b)}
\caption{(a) Given a rate pair $(R_1, R_2)$ that is inside the capacity region, consider the line passing $(R_1,R_2)$ and the origin. Then, $(R'_1, R'_2)$ is the point of intersection of this line with the boundary of the capacity region.
(b) The conceptual shape of the lower bound on the error exponent of a given MAC with respect to the transmission rate pair $(R_1,R_2)$.}
\label{fig: R' and R}
\end{figure}
The following lemma provides a geometric characterization of the lower bound with the proof provided in Appendix \ref{proof:lem:lb_shape}.
\begin{lem}\label{lem:lb_shape}
    The reliability function of a MAC channel is bounded from below as 
\begin{align*}
   E(R_1, R_2)\geq  D_{lb} \Big(1-\frac{\|R\|}{\mathcal{C}(\theta_{R})}\Big).  
\end{align*}
\end{lem}
This lemma provides a lower bound that decreases linearly by $\frac{\|R\|}{\mathcal{C}(\theta_{R})}$ as $(R_1,R_2)$ gets closer to the capacity boundary (see Fig. \ref{fig: R' and R}(a)). Based on this observation, we present  Fig. \ref{fig: R' and R}(b) that shows the shape of a typical lower bound as a function of the transmission rate pairs.

\begin{figure}
    \centering
    \includegraphics[scale=0.8]{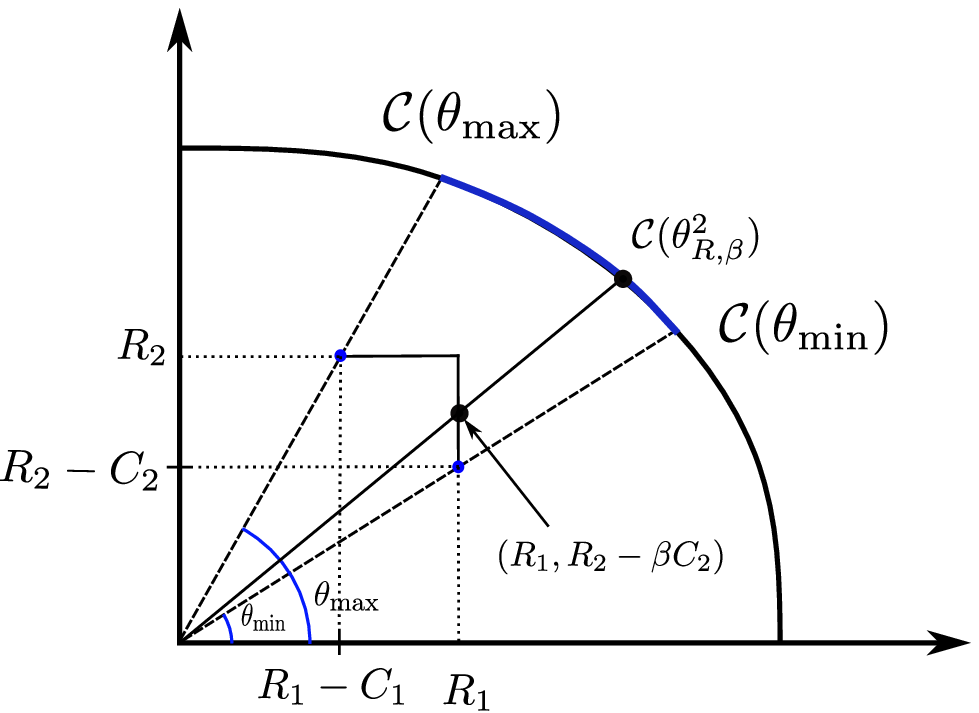}
    \caption{The three-phase scheme provides more flexibility to adjust the rates during the second phase. Recall that the two-phase lower bound in Lemma \ref{lem:lb_shape} depends  on the ratio of the first stage transmission rates $\norm{R}$ and $\mathcal{C}(\theta)$, the capacity boundary as in Fig. \ref{fig: R' and R}.  The three-phase lower bound is similar but with the flexibility that the first-stage transmission rate can be adjusted by allocating a fraction of the rate to the second stage. For example, suppose the second user transmits with rate $\beta C_2$ during the second phase with $\beta \in [0,1]$ and $C_2$ being the capacity during this stage.  Then the first stage rates become $(R_1, R_2-\beta C_2)$.  In that case, the three-phase lower bound scales with the ratio of the new effective rate and $\mathcal{C}(\theta^2_{R, \beta})$, the capacity boundary at the new angle.  With such flexibility, the three-phase scheme gives a higher error exponent than the two-phase scheme.}
    \label{fig:shape three phase}
\end{figure}
The above bound is a result of a two-phase scheme. Next, we study a geometric representation of a lower bound using the three-phase scheme. Recall that there are three stages in this scheme: (1) data communication for both users; (2) confirmation stage only for one of the users; and (3) confirmation stage for both users. As a result, one user continues data transmission during the second phase. Because of this approach, the user allocates part of its message to the second phase,  lowering its rate during the first phase. For a more intuitive argument, we consider a simplified version of the three-phase lower bound given in Theorem \ref{thm: Error Exp lowerbound three phase}. For that, we consider the following looser bound
\begin{align*}
E(R_1,R_2) \geq  D_{lb} \max_{0\leq \gamma_2 \leq 1} \max\set{\Eo(R_1, R_2-\gamma_2C_2)-\gamma_2, \Eo(R_1-\gamma_2C_1, R_2)-\gamma_2},
\end{align*}
where, $C_1=\max_{x_1}C_1(x_1)$ and  $C_2=\max_{x_2}C_2(x_2)$ are the capacity of the \ac{ptp} channel during the second phase. We obtain this bound by noting that $\bar{D}_{P_{{\bfZ}}}(00\| ab)$ in the theorem is greater than $D_{lb}$ as in \eqref{eq: D_l}. Moreover, we ignored the non-negative terms $\gamma_2\DQbar{1}$ and $\gamma_2\DQbar{2}$ in \eqref{eq:El 1} and \eqref{eq:El 2}.

Using a similar argument as in Lemma \ref{lem:lb_shape}, we obtain a geometric interpretation of the above bound as
\begin{align*}
E(R_1,R_2) \geq  D_{lb} \max_{0\leq \beta \leq 1} \max\set{(1-\beta-\frac{\|(R_1-\beta C_1, R_2)\|}{\mathcal{C}(\theta^1_{R, \beta})}), (1-\beta-\frac{\|(R_1, R_2 - \beta C_2)\|}{\mathcal{C}(\theta^2_{R, \beta})})},
\end{align*}
where $\theta^1_{R, \beta}$ and $\theta^2_{R, \beta}$ are the angle of the point $(R_1-\beta C_1, R_2)$ and $(R_1, R_2 - \beta C_2)$ in polar coordinates, respectively. Note that the second maximization determines whether the first or the second user must continue data transmission during the second stage. Moreover, the parameter $\beta$ determines the amount of the rate allocated to the second phase. Therefore, by tuning $\beta$ and choosing which user to transmit information  to during the second stage, we get   the effective rate during the first stage ranging from $(R_1-C_1, R_2)$, when the first user sends data with rate $C_1$ during the second stage ($\beta=1$),  to $(R_1, R_2)$ where there is no second stage ($\beta=0$), and   to $(R_1, R_2-C_2)$  where the second user sends data with rate $C_2$ during the second stage ($\beta=1$). These rates are demonstrated as the vertical and horizontal lines in Fig. \ref{fig:shape three phase}. Observe that by setting $\beta$, the corresponding angle ranges from $\theta_{\min}$ to $\theta_{\max}$ as in Fig. \ref{fig:shape three phase}. These angles are calculated as 
\begin{align*}
    \theta = \arctan(\frac{R_2}{R_1})\qquad 
    \theta_{\min}= \arctan(\frac{[R_2-C_2]^+}{R_1}) \qquad \theta_{\max}= \arctan(\frac{R_2}{[R_1-C_1]^+}).
\end{align*}
As a result, $\mathcal{C}(\theta_R)$ is the capacity boundary as in Fig.~\ref{fig: R' and R}   ranges between $\mathcal{C}(\theta_{\min})$  to $\mathcal{C}(\theta_{\max})$ as in Fig. \ref{fig:shape three phase}. 
Therefore, by appropriately tuning $\beta$, we can get a better achievable exponent than the two-phase scheme. In other words,  the three-phase scheme allows more flexibility in changing the angle $\theta_{R}$ by allocating parts of the rates to the second phase.  We should note that the three-phase bound in Theorem \ref{thm: Error Exp lowerbound three phase} is larger due to the additional terms such as $\gamma_2\DQbar{1}$ and $\gamma_2\DQbar{2}$. We ignored those terms here for the sake of a more intuitive argument.

\subsection{On the tightness of the error exponent bounds}
In what follows, we provide examples of classes of channels for which the lower and upper bound coincide. 
The following is an example of a MAC for which the two-phase bounds match. 
\begin{example}\label{ex:symmetric MAC}
Consider a MAC with input alphabets $\mathcal{X}_1=\mathcal{X}_2=\{0,1,2\}$, and output alphabet $\mathcal{Y}=\{0,1,2\}$. The following relation describes the transition probability of the channel:
$Y=X_1 \oplus _3 X_2 \oplus_3 N_p$,
where the additions are modulo-3, and $N_p$ is a random variable with $P(N_p=1)=P(N_p=2)=p$, and $P(N_p=0)=1-2p$, where $0\leq p\leq 1/2$. 
It can be shown that for this channel $D_l=D_u=(1-3p) \log \frac{1-2p}{p}.$
Hence, the upper-bound in Corollary \ref{cor:twophase up} matches the lower-bound in Corollary \ref{cor:twophase lb}.
\end{example}
The argument in the above example can be extended to $m$-ary additive MACs for $m>2$,
where all the random variables take values from $\ZZ_m$, and $N_p$ is a random variable with $P(N_p=i)=p$ for any $i\in \ZZ_m, ~i\neq 0$ and $P(N_p=0)=1-(m-1)p$. It can be shown that for this channel 
$$D_l=D_u=(1-mp) \log \frac{1-(m-1)p}{p}.$$
Next, MAC consists of two independent channels. Interestingly, the two-phase bounds do not match but the three-phase bounds match.  

\begin{example}\label{ex:two ptp}
Consider a MAC in which the output is $(y_1, y_2)\in \set{0,1}\times \set{0,1}$ with the transition probability matrix described by the product $Q_{Y_1, Y_2 | X_1, X_2}= Q_{Y_1|X_1}Q_{Y_2|X_2}$. This MAC consists of two parallel (independent) \ac{ptp} channels. Suppose $C_1$ and $C_2$ are the first and second parallel channel capacities. Next, we define the corresponding relative entropies of these channels. For $i=1,2$, let $$D_i=\max_{x_i, z_i} D_{KL}\Big(Q_{Y_i|X_i}(\cdot | x_i) \| Q_{Y_i|X_i}(\cdot | z_i)\Big).$$ 

We start with calculating the two-phase bounds in Corollaries \ref{cor:twophase up} and \ref{cor:twophase lb}. They simplify to the following 
\begin{align*}
   E^{ub}_{\text{two-phase}}(R_1,R_2) &= \max(D_1, D_2) \min\Big\{(1-\frac{R_1}{C_1}), (1-\frac{R_2}{C_2})\Big\}\\
    E^{lb}_{\text{two-phase}}(R_1,R_2) &= \min(D_1, D_2) \min\Big\{(1-\frac{R_1}{C_1}), (1-\frac{R_2}{C_2})\Big\}.
\end{align*}
These bounds do not match when $D_1\neq D_2$.

Next, we  study the upper bound in Theorem \ref{thm:upper bound three phase}. Note that the relative entropy terms in the theorem are equal to $D^1(P)=D_1, D^2(P)=D_2$, and $D^3=D_1+D_2$. Moreover, maximizing over $P$ gives $I^1(P)=C_1, I^2(P)=C_2,$ and $I^3(P)=C_1+C_2$. Hence, it is not difficult to see that the upper bound simplifies to the following
\begin{align*}
    E(R_1,R_2)\leq E^{ub}(R_1, R_2) = \min\Big\{D_1(1-\frac{R_1}{C_1}), D_2(1-\frac{R_2}{C_2})\Big\},
\end{align*}
where we used the inequality $\max\{\frac{a}{b}, \frac{a'}{b'}\}\geq \frac{a+a'}{b+b'}$ for any $a, a',b, b'>0$.
Next, we show that this upper bound matches the lower bound in Theorem \ref{thm: Error Exp lowerbound three phase}. The calculations that are more involved are given below. Without loss of generality, assume $\frac{R_1}{C_1}\leq \frac{R_2}{C_2}$.

We start with simplifying $\Elb^1(R_1,R_2)$ in \eqref{eq:El 1}. Since the channels are independent, then
 $C_2(x_1)=C_2$ and $C_1(x_2)=C_1$ for any $x_1\in \mathcal{X}_1$ and $x_2\in \mathcal{X}_2$. Hence, we can independently optimize over $x_1$ in  \eqref{eq:El 1}. Note that  $\bar{D}_{P^*_{Z}}(00\| 10)=\DQbar{1}(x^*_1, P_{X_2})=D_1$, where the superscript $^*$ implies that we optimized over the corresponding variable. Similarly, note that $\bar{D}_{P^*_{Z}}(00\| 01)=\DQbar{2}(x^*_2, P_{X_1})=D_2$ and $\bar{D}_{P^*_{Z}}(00\| 11)=D_1+D_2$.  As a result, $\Elb^1(R_1,R_2)$ simplifies to the following
\begin{align*}
\Elb^1(R_1,R_2) =  \sup_{0\leq \gamma_2\leq 1}\min\Big\{ &
 \big(\Eo(R_1, R_2-\gamma_2C_2)-\gamma_2\big)D_2,\\
&  \big(\Eo(R_1, R_2-\gamma_2C_2)-\gamma_2\big) D_1+\gamma_2D_1,\\
&  \big(\Eo(R_1, R_2-\gamma_2C_2)-\gamma_2\big) (D_1+D_2)+\gamma_2D_1\Big\}.
\end{align*}
Note that the third term can be ignored as it is always larger than the second term. We set $\gamma_2=\frac{R_2}{C_2}-\frac{R_1}{C_1}$. This is a valid assignment as $0\leq \gamma_2\leq 1$. With this assignment,  $\Eo(R_1, R_2-\gamma_2C_2-\gamma_2\big) = (1-\frac{R_1}{C_1})$. Therefore, due to the supremum over $\gamma_2$, with the above assignment, we get a lower bound for $\Elb^1(R_1,R_2)$:
\begin{align*}
\Elb^1(R_1,R_2) &\geq  \min\Big\{ \big(1-\frac{R_1}{C_1}-(\frac{R_2}{C_2}-\frac{R_1}{C_1})\big)D_2,  \big(1-\frac{R_1}{C_1}\big) D_1\Big\}\\
&=\min\Big\{D_1(1-\frac{R_1}{C_1}), D_2(1-\frac{R_2}{C_2})\Big\}.
\end{align*}
The lower bound in Theorem \ref{thm: Error Exp lowerbound three phase} is greater or equal to the above expression. Hence, we get that 
\[
E(R_1, R_2) \geq \min\Big\{D_1(1-\frac{R_1}{C_1}), D_2(1-\frac{R_2}{C_2})\Big\}. 
\]
This is a lower bound that matches the upper bound we derived above.  
\end{example}


%
%

%% file: DriftAnalysis.tex
\subsection{Prune-timing technique}
One of the major technical challenges in our proofs is the transient period from linear to logarithmic drifts. We address the transient period by proposing a new technique called ``time pruning". Particularly, we allow the time to increase randomly as a stopping time with respect to the underlying random process.  

Let $\set{H_t}_{t>0}$ be a non-negative random process. Based on this process we  define the pruned time random process $\set{t_n}_{n>0}$. First, for any $\epsilon \geq 0$ and $N\in \NN$  define the following random variables
\begin{align}\label{eq:Tleps}
\Tleps &\deq \inf\set{t>0: H_t\leq \epsilon} \wedge N,\\\label{eq:Tueps}
\Tueps &\deq \sup\set{t>0: H_{t-1}\geq \epsilon} \wedge N =\Big[ \sup\set{t>0: H_t\geq \epsilon}+1\Big]\wedge N.
\end{align}
Then the pruned time process is defined as 
\begin{align}\label{eq:tn}
t_n \deq \begin{cases}
               n & \text{if}~ n< \Tleps,\\
               n\vee \Tueps & \text{if}~ \Tleps \leq n\leq N,\\
                N  & \text{if}~ n> N.                 
            \end{cases}
\end{align} 
The diagram of $t_n$ as a function of $n$ is demonstrated in Fig. \ref{fig:t_n}.
\begin{figure}[hbtp]
\centering
\includegraphics[scale=0.9]{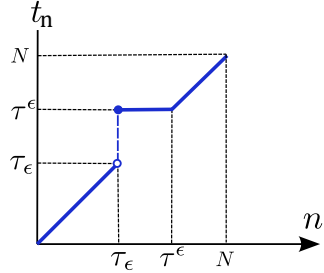}
\caption{The evolution of $t_n$ with $n$.}
\label{fig:t_n}
\end{figure}
Now we consider the random sampling of the random process $\set{Y_t}_{t>0}$ based on the pruned time process $\set{t_n}_{n>0}$. Particularly, we are interested in the resulted random process $\set{Y^{t_n}}_{n>0}$ which represents the receiver's observation in pruned time $t_n$.
\begin{example}
Fig. \ref{fig:Yt_n} shows an example of the process $\set{Y^{t_n}}_{n>0}$. Let $N=4$, and $Y_t$ take values from $\set{0, 1}$. Then, each path in the tree in Fig. \ref{fig:Yt_n} represents a particular realization of $Y^N$. For example, the left most path corresponds to $y^4=(0,0,0,0)$. In this case $\Tleps=2$ and $\Tueps=4$. Hence, $t_1=1, t_2=4, t_3=4, t_4=4$. Therefore, the pruned process evolves as
 $$y^{t_1}=0, \qquad y^{t_2}=(0,0,0,0),  \qquad y^{t_3}=(0,0,0,0),  \qquad y^{t_4}= (0,0,0,0).$$
It implies that in the second sample a jump in the future occurs and all samples of $y_t$ from $t=\Tleps=2$ to  $\Tueps=4$ become available.  Hence, in this example, the observed samples in pruned time are $0, (0,0,0),-,-$, implying that second sample contains $(y_2, y_3, y_4)$ and the third and fourth samples contain no new information.
\end{example}

\begin{figure}[hbtp]
\centering
\vspace{-20pt}
\includegraphics[scale=1.7]{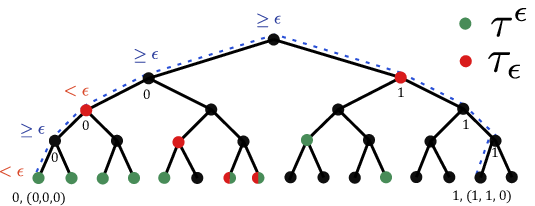}
\caption{An example of a random sampled random process $\set{Y_t}_{t>0}$ based on the pruned time $\set{t_n}_{n>0}$. Here, $N=4$, and $Y_t$ takes values from $\set{0, 1}$. The upper labels ($\geq \epsilon$, $<\epsilon$) indicate whether $H_t$ is greater or smaller than $\epsilon$. Each path in the tree represents a particular realization of $Y^N$. For example, the left most path corresponds to $y^4=(0,0,0,0)$. In this case $\Tleps=2$ and $\Tueps=4$. Hence, $t_1=1, t_2=4, t_3=4, t_4=4$. Therefore, $y^{t_1}=0, y^{t_2}=(0,0,0,0), y^{t_3}=(0,0,0,0), y^{t_4}= (0,0,0,0)$. Hence, the observed samples in pruned time are $0, (0,0,0),-,-$, implying that second sample contains $(y_2, y_3, y_4)$ and the third and fourth samples contain no new information.}
\label{fig:Yt_n}
\end{figure}

Let $\{\tilde{Y}_n\}_{n>0}$ be the sampled version of $\set{Y_t}_{t>0}$ with $\tilde{Y}_n=Y^{t_n}$. For instance, in Fig. \ref{fig:Yt_n}, 
$$\tilde{y}_1=0, \qquad \tilde{y}_2=(0,0,0,0),  \qquad \tilde{y}_3=(0,0,0,0),  \qquad \tilde{y}_4= (0,0,0,0).$$
Let $\mathcal{G}_n, n>0$ represent the $\sigma-$algebra generated by $\tilde{Y}_n$. Then, as $t_n$ is monotonically increasing with $n$, $\mathcal{G}_n$ is a filtration, i.e., $\mathcal{G}_1\subseteq \mathcal{G}_2 \subseteq \cdots \subseteq \mathcal{G}_N$.

Note that $\Tleps$ is a stopping time with respect to $\{H_t\}_{t>0}$, but this is not the case for $\Tueps$. This raises the question as to whether one can determine the occurrence of  $\Tleps$ and $\Tueps$ based on the past observation of the sampled process $\{\tilde{Y}_n\}_{n>0}$. We show that the answer is affirmative. More precisely, whether $\Tleps\geq t_n$ and/or $\Tueps\geq t_n$ are determined  having $\tilde{Y}^n$. For that, we prove the following statements. 

\begin{lem}\label{prop:measurable}
The processes $t_n$, $t_n\wedge \Tleps$, and $t_n\wedge \Tueps$ are measurable with respect to the filtration $\mathcal{G}_n, n>0$.
\end{lem}
\begin{proof}
We show that the random processes in the statement are functions of $\tilde{Y}^n, n>0$. Note that $t_n$ is measurable as it is a function of the length of the vector $\tilde{Y}^n$. For any $n\leq N$, the other random processes can be written as 
\begin{align*}
 t_n\wedge \Tleps = n \11\set{\Tleps >n}+\Tleps \11\set{\Tleps\leq n}, \qquad t_n\wedge \Tueps = n \11\set{\Tleps > n}+\Tueps \11\set{\Tleps\leq n}.
 \end{align*} 
 This statement follows from the definition of $t_n$ in \eqref{eq:tn} implying when $t_n<\Tleps$ then $t_n =n$, and that when $t_n \geq \Tleps$ then $t_n \geq \Tueps$.  Since, $\Tleps$ is a stopping time w.r.t $\set{Y_t}_{t>0}$, then $\11\set{\Tleps > n}$ is measurable w.r.t $\mathcal{F}_t$ and, hence, w.r.t $\mathcal{G}_n$. Moreover,  $\Tleps \11\set{t_n\geq \Tleps}$ is also measurable w.r.t $\mathcal{G}_n$. This is because $\Tleps$ is completely determined when $n\geq \Tleps$, as we have more than $\Tleps$ samples. Consequently, $t_n\wedge \Tleps$ is measurable. Similarly, $\Tueps \11\set{\Tleps\leq n}$ is measurable w.r.t $\mathcal{G}_n$; because when $\Tleps\leq n$, then $t_n\geq \Tueps$. In that case, $t_n$ already jumped up and $\Tueps$ is determined by counting the number of samples in the jump.  Consequently,  $t_n\wedge \Tueps$ is also measurable. 
\end{proof}

\subsection{Analysis of the entropy drift of the messages}
Next, we analyze the drift of massages' entropy. 
For that, we consider the marginal and conditional entropy of the messages. More precisely, define the following random processes: 
\begin{align*}
\Hbar_t^1 & = H(W_1 | W_2, \mathcal{F}_t ), &    \Hbar_t^2 &= H(W_2 | W_1, \mathcal{F}_t ),    & \Hbar_t^3  & = H(W_1, W_2 | \mathcal{F}_t ),\\\numberthis \label{eq:H processes}
\Htld_t^1 & = H(W_1 | \mathcal{F}_t ), &    \Htld_t^2 &= H(W_2 |\mathcal{F}_t ),    & \Htld_t^3 & = H(W_1, W_2 | \mathcal{F}_t ).
\end{align*}
We analyze the drift of each pair $\{(\Hbar_t^i, \Htld_t^i)\}_{t>0}$, where $i=1,2,3$. First, the following result is provided with the proof in Appendix  \ref{proof:lem:pruned sub martingale 2}.
\begin{lem}\label{lem:pruned sub martingale 2}
Let $\set{\Hbar_r}_{r>0}$ and $\{\Htld_r\}_{r>0}$  be a pair of random processes with the following properties with respect to a filtration $\mathcal{F}_r, r>0,$: 
\begin{subequations}\label{eq:Hbar Hrld conditions}
\begin{alignat}{5}\label{eq:Hbar less Htld}
\Hbar_r &\leq \Htld_r  & \qquad \text{if}~ r&> 0,\\\label{eq:Hbar linear K1}
 \EE[\Hbar_{r+1} - \Hbar_r |\mathcal{F}_r]&\geq -k_{1,r+1},  & \qquad  \\\label{eq:Htld log K2}
\EE[\log \Htld_{r+1} - \log \Htld_r |\mathcal{F}_r]&\geq -k_{2,r+1}      & \qquad \text{if}~   \Htld_r & < \epsilon\\\label{eq:Htld log K3}
\abs{\log \Htld_{r+1} - \log \Htld_r} & \leq k_3    & &\\\label{eq:Htld K4}
\abs{ \Htld_{r+1}-\Htld_r} &\leq k_4   & &
\end{alignat}
\end{subequations}
where $k_{1,r}, k_{2,r}$ are non-negative random variables measurable w.r.t $\mathcal{F}_{r}$, and $k_{1,r}\leq k_{2,r}$ for all $r>0$. Moreover, $k_3, k_4$ are non-negative constants. Let $\set{\tilde{t}_n}_{n>0}$ be the pruned time process defined in \eqref{eq:tn} with respect to $\{\Htld_r\}_{r>0}$. Further, let $\Tlepstld$ and $\Tuepstld$ be the random variables as in \eqref{eq:Tleps} and \eqref{eq:Tueps} but with $H_t=\Htld_t$, respectively. Given $\epsilon \in (0,1)$ and $I\geq D>0$ define the random process $\set{Z}_{t>0}$ as 
\begin{align*}
Z_t \deq (\frac{\Hbar_t-\epsilon}{I})\11\{\Htld_t\geq \epsilon\}+\Big(\frac{\log \Htld_t - \log \epsilon}{D}+f(\log \frac{\Htld_t}{\epsilon}) \Big) \11\{\Htld_t<\epsilon\},
\end{align*}
where $f(y)=\frac{1-e^{\lambda y}}{\lambda D}$ with $\lambda >0$. Further define the random process $\set{S_t}_{t>0}$ as 
\begin{align*}
S_t \deq \sum_{r=1}^{t\wedge \Tlepstld} \frac{k_{1,r}}{I} + \sum_{r=t\wedge \Tlepstld+1}^{t\wedge \Tuepstld} \frac{k_4}{I} \11\{\Htld_{r-1} \geq \addvb{\sqrt{\epsilon}}\}+\sum_{r=t\wedge \Tuepstld+1}^{t} \frac{k_{2,r}}{D} + \addvb{\sqrt{\epsilon}} \frac{N}{I}\11\{t\geq \Tuepstld\},
\end{align*}
Lastly define the random process $\set{\Zprune}_{n>0}$ as $\Zprune \deq Z_{t_n}+S_{t_n}.$ Then, for small enough $\lambda>0$ the process $\set{\Zprune}_{n>0}$ is a sub-martingale with respect to the time pruned filtration $\mathcal{G}_{n}, n>0$. 
\end{lem}

We want to show that for each $i=1,2,3$, the pair of random processes $\{(\Hbar^i_t, \Htld^i_r)\}_{r>0}$ satisfy the conditions in \eqref{eq:Hbar Hrld conditions} and hence Lemma \ref{lem:pruned sub martingale 2} can be applied. Clearly, \eqref{eq:Hbar less Htld} holds as conditioning reduces the entropy. Condition \eqref{eq:Hbar linear K1} holds because of the following lemma.  
\begin{lem}\label{lem: linear drift MAC FB}
Given an $(M_1, M_2, N)$-VLC,  the following inequalities hold almost surely for $1 \leq r \leq N$
\begin{subequations}\label{eq:linear drift J}
\begin{align}\label{eq:linear drift J 1}
\EE[\Hbar^1_{r}-\Hbar^1_{r-1}|\mathcal{F}_{r-1}] & \geq - \J_{r}^1, &  \text{where}~  \J_r^1 &\deq I(X_{1,r}; Y_r | X_2^r, \mathcal{F}_{r-1}),\\\label{eq:linear drift J 2}
\EE[\Hbar^2_{r}-\Hbar^2_{r-1}|\mathcal{F}_{r-1}] & \geq - \J_{r}^2, &   \text{where}~  \J_r^2 &\deq I(X_{2,r}; Y_r | X_1^r, \mathcal{F}_{r-1}),\\\label{eq:linear drift J 3}
\EE[\Hbar^3_{r}-\Hbar^3_{r-1}|\mathcal{F}_{r-1}] & \geq - \J_{r}^3, &   \text{where}~  \J_r^3 &\deq I(X_{1,r}, X_{2,r}; Y_r | \mathcal{F}_{r-1}),
\end{align}
\end{subequations}
where the mutual information quantities are calculated with the induced probability distribution $P_{X_1^NX_2^NY^N} \in \PMAC^N$.
\end{lem}
\begin{IEEEproof}
We start with proving the inequality in  \eqref{eq:linear drift J 1}. 
\begin{align*}
\EE[\Hbar^1_{r}-\Hbar^1_{r-1}|y^{r-1}] & = H(W_1|Y_r, W_2, y^{r-1}) -H(W_1|W_2,  y^{r-1}) = -I(W_1; Y_r | W_2, y^{r-1})\\ 
& \stackrel{(a)}{=} -I(W_1, X_{1,r}; Y_r | W_2, X_{2}^{r}, y^{r-1})\\
& \stackrel{(b)}{=} -H(Y_r | W_2, X_{2}^{r}, y^{r-1})+ H(Y_r | X_{2,r}, X_{1,r}, X_2^{r-1}, y^{r-1})\\
& \geq  -H(Y_r | X_{2}^{r}, y^{r-1})+ H(Y_r | X_{2}^r, X_{1,r}, y^{r-1}) = - J_r^1,
\end{align*}
where (a) holds as $X_{1,r} = e_{1,r}(W_1, y^{r-1})$ and $X_{2,\ell} = e_{1,\ell}(W_1, y^{\ell-1}), 1\leq \ell \leq r$. Equality (b) is due to \eqref{eq: chann probabilities} implying that conditioned on $(X_{1,r}, X_{2,r})$ the channel's output $Y_r$ is independent of $W_1, W_2$, and $X_2^{r-1}$. 

Similarly for the inequality in \eqref{eq:linear drift J 2}, the following lower-bound holds 
\begin{align*}
\EE[\Hbar^2_{r}-\Hbar^2_{r-1}|y^{r-1}] &= -I(W_2; Y_r | W_1, y^{r-1})=-I(W_2, X_{2,r}; Y_r | W_1, X_{1}^{r}, y^{r-1})\\
& \geq  -H(Y_r | X_{1}^{r}, y^{r-1})+ H(Y_r | X_{1}^{r}, X_{2,r}, y^{r-1}) = - J_r^2.
\end{align*}
Using a similar argument for \eqref{eq:linear drift J 3}, we can show that the following inequality holds 
\begin{align*}
\EE[\Hbar^3_{r}-\Hbar^3_{r-1}|y^{r-1}] &= -I(W_1, W_2; Y_r | y^{r-1})\\
&=-I(W_1, X_{1,r}, W_2, X_{2,r}; Y_r |y^{r-1}) 
= - J_r^3.
\end{align*}
\end{IEEEproof}
For condition \eqref{eq:Htld log K3}, we need the following result.
\begin{lem}\label{rem: H_t+1 - H_tis bounded}
For any $r\geq 1$ and $i=1,2,3$, the following inequality holds if the channel's transition probabilities $Q(\cdot |\cdot ,\cdot)$ are positive.  
\begin{align*}
\abs{\log \Htld_r^i-\log \Htld_{r-1}^i} \leq \eta  \deq \max_{x_1, z_1 \in\mathcal{X}_1}\max_{x_2, z_2 \in \mathcal{X}_2}\max_{y\in \mathcal{Y}} \log \frac{Q(y|x_1,x_2)}{Q(y|z_1,z_2)}.
\end{align*}
\end{lem}
\begin{IEEEproof}
The proof follows from the  argument given in that of Lemma 4 in \cite{Burnashev}.
\end{IEEEproof}

Condition \eqref{eq:Htld log K2} holds as a result of the following lemma which is proved in Appendix \ref{appx: proof of lem log drift}.
\begin{lem}\label{lem: log drift MAC FB}
Given an $(M_1, M_2, N)$-VLC and $\epsilon \in [0,\frac{1}{2}]$, if $\Htld^i_r<\epsilon$, then the following inequality holds almost surely
\begin{align}\label{eq:log drift D}
\EE[\log \Htld^i_{r}-\log \Htld^i_{r-1}|\mathcal{F}_{r-1}] & \geq - \D_{r}^i-O(\sqrt{h_b^{-1}(\epsilon)}), 
\end{align}
where and  $\D_r^i$ is a function of $y^{r-1}$ and is defined as 
\begin{align}\label{eq:D_r definition}
\D_r^i \deq \max_{x_i, x'_i}\sup_{P_{X'_j|Y^{r-1}}} D_{KL}\Big( \Qbar{i}_r(\cdot | x_i, y^{r-1})~\|~ \Qbarp{i}_r(\cdot | x'_i, y^{r-1}) \Big),
\end{align}
where $\Qbar{i}_r = P_{Y_r|X_{i,r}, Y^{r-1}}$ is the effective channel from the user $i$'s perspective. 
\end{lem}

Lastly, condition \eqref{eq:Htld K4} follows as $\Htld^i_r\leq \log M_i$, that is
\begin{align*}
\abs{\Htld^i_r-\Htld^i_{r-1}} \leq \max \{ \Htld^i_r, \Htld^i_{r-1} \} \leq \log M_i,
\end{align*}
where  $M_3 = M_1M_2$.
As a result,  we apply Lemma \ref{lem:pruned sub martingale 2} on each pair $\{(\Hbar^i_t, \Htld^i_r)\}_{r>0}$, where $i=1,2,3$ and with the following parameters
\begin{align*}
k_{1,r} = J^i_r, \qquad k_{2,r} = D^i_r, \qquad k_3 = \eta, \qquad k_4 = \log M_i, \qquad I=I^i, \qquad D=D^i.
\end{align*}
That said, let $Z^i_t, S^i_t, t^i_n,$ and $L^i_n$ be the corresponding random processes as in Lemma \ref{lem:pruned sub martingale 2} but for $(\Hbar_t, \Htld_t) = (\Hbar_t^i, \Htld^i_t), \forall t\geq 0$ and $D=D^i, I=I^i$, where $i=1,2,3$. Therefore, from the lemma, $\{L^i_n\}_{n >0}$ is a sub-martingale w.r.t $\mathcal{F}_{t^i_n}, n>0$.

\subsection{Connection to the error exponent}
Now we are ready to study the connection between the entropy drift and the error exponent. For that we have the following lemma with the proof given in Appendix \ref{proof:lem:Err Exp conenction}. 

\begin{lem}\label{lem:Err Exp conenction}
For any $(M_1,M_2, N)$-VLC with rates $(R_1, R_2)$, probability of error $P_e$ and stopping time $T$ the following holds
\begin{align}\label{eq:err exponent up eq 1}
\frac{-\log P_e}{\EE[T]} \leq \frac{1}{1-\Delta} D^i \Big( \frac{ \EE\big[S^i_{T\vee \Tuepstldi}\big]}{\EE[T]} - \frac{R_i}{I^i}\Big) + \frac{1}{1-\Delta} U_i(P_e, M_i, \epsilon), \qquad i=1,2,3.
\end{align}
where $\{S_t\}_{t>0}$ is the slope process as in Lemma \ref{lem:pruned sub martingale 2}, and
\begin{align*}
U_i(P_e, M_i, \epsilon) &\deq  R_i\Big( D^i \frac{h_b(P_e)+P_e \log(M_1M_2)+\epsilon}{I^i \log M_i} +\frac{- \log \epsilon}{\log M_i}+\frac{1}{\lambda  \log M_i}\Big),\\
\Delta & \deq \frac{\log \big(-\log P_e +2+ \log M_1M_2 \big)}{-\log P_e}.
\end{align*}
\end{lem}

%
%

Next, we find appropriate $I^i$ and $D^i$ so that $\EE\big[S^i_{T\vee \Tuepstldi}\big] \approx \EE[T]$. Further, we argue that $\Delta$ and $U_i(P_e, M_i, \epsilon)$ converge to zero for any sequence of VLCs satisfying the conditions in Definition \ref{def:ErrExponent}.
We have the following lemma toward this.

\begin{lem}\label{lem:bounding S}
Under the condition that $\epsilon > \alpha(P_e)$, for each $i=1,2,3$, if   
\begin{align*}
I^i = \frac{1}{\EE[\Tlepstldi]} \EE\Big[ \sum_{r=1}^{\Tlepstldi}\J^i_r \Big], \qquad D^i = \frac{1}{\EE[T -\Tlepstldi]} \EE\Big[ \sum_{r=\Tlepstldi+1}^{T}\D^i_r \Big],
\end{align*}
then $\EE\big[S^i_{T\vee \Tuepstldi}\big]\leq \EE[T] (1+ V(\epsilon, N)),$ where $V(\epsilon,N) =\frac{R_i}{I^i} \big( \addvb{\sqrt{\epsilon}} N \big)+ \addvb{\sqrt{\epsilon}} \frac{N}{\EE[T] I^i}$.
\end{lem}

Before presenting the proof of the lemma, let us discuss its consequences. From this lemma and \eqref{eq:err exponent up eq 1}, we get the desired upper bound by appropriately setting $I^i$ and $D^i$ as in the lemma. Hence, we get for $i=1,2,3$,
\begin{align}\label{eq:err exponent up eq 2}
\frac{-\log P_e}{\EE[T]} \leq  \frac{1}{1-\Delta}D^i \Big(1 - \frac{R_i}{I^i}\Big) +  \frac{1}{1-\Delta}D^i V(\epsilon,N) + \frac{1}{1-\Delta} U_i(P_e, M_i, \epsilon).
\end{align}
Note that all $D^i, t=1,2,3$ is bounded by a constant $d^{max}$ depending only on the channel's characteristics. Then, having $R_i\leq I^i$, the bound in \eqref{eq:err exponent up eq 2} is simplified to the following 
\begin{align}\label{eq:err exponent up eq 3}
\frac{-\log P_e}{\EE[T]} \leq  D^i \Big(1 - \frac{R_i}{I^i}\Big) +  \frac{d^{max}}{1-\Delta} \Big(\Delta+ V(\epsilon,N)\Big) + \frac{1}{1-\Delta} U_i(P_e, M_i, \epsilon),
\end{align}
for $i=1,2,3$.  Let the residual terms above denoted by
\begin{align*}
 \tilde{V}(P_e, M_1,M_2, \epsilon, N)  \deq d^{max} \frac{\Delta}{1-\Delta}+ \frac{1}{1-\Delta}d^{max} V(\epsilon,N) + \frac{1}{1-\Delta} \max_{i}\set{U_i(P_e, M_i, \epsilon)}.
 \end{align*}
We show that for any sequence of $(M_1^{(n)},M_2^{(n)}, N^{(n)})$-VLCs as in Definition \ref{def:ErrExponent} with vanishing $P_e^{(n)}$, the above term converges to zero as $n\rightarrow \infty$. It is easy to see that $\Delta$ as in \eqref{eq:Delta} converges to zero as $P_e^{(n)}\rightarrow 0$. Further, by setting $\epsilon^{(n)} = (\frac{1}{N^{(n)}})^2$, we can check that $\lim_{n\rightarrow \infty} V(\epsilon^{(n)}, N^{(n)})=0$. It remains to show the convergence of $U_i(\cdot)$ as in \eqref{eq:Ui}. The convergence of the third term in \eqref{eq:Ui} follows as $\lim_{n\rightarrow \infty}\frac{1}{M_i^{(n)}}= 0$. For the second term, as  $\epsilon^{(n)} = (\frac{1}{N^{(n)}})^2$ then we have that
\begin{align*}
\lim_{n\rightarrow \infty} \frac{-\log \epsilon^{(n)}}{\log M_i^{(n)}} = \lim_{n\rightarrow \infty} \frac{2 \log N^{(n)} }{\log M_i^{(n)}} =0,
 \end{align*} 
where the last equality holds as $N^{(n)}$ grows sub-exponentially with $n$. The convergence of the first term also follows from the fact that $\lim_{n\rightarrow \infty}\alpha(P_e^{(n)})=0$, as $P_e^{(n)}$ converges exponentially fast\footnote{The exponential convergence of $P_e^{(n)}$ holds because otherwise the error exponent is zero.}. To sum-up, with the above argument we showed that 
\begin{align*}
\lim_{n\rightarrow \infty}  \tilde{V}(P_e^{(n)}, M_1^{(n)},M_2^{(n)}, \epsilon^{(n)}, N^{(n)})=0.
\end{align*}

Hence, by maximizing over all distributions and from Definition \ref{def:ErrExponent}, we get the following upper bound on the error exponent
\begin{align}\label{eq: Err Exp up 1}
E(R_1,R_2) 
\leq \sup_{N\in \NN} \sup_{P^N \in \PMAC^N} \sup_{T:  T\leq N}  \min_{i\in \{1,2,3\}} \left\{D^i(P^N) \Big(1 - \frac{R_i}{I^i(P^N)}\Big)\right\}, 
\end{align}
where the maximizations are taken over all distributions $P^N\in \PMAC^N$. 
Further,  $I^i(P^N)$ and $D^i(P^N)$ are defined as in Lemma \ref{lem:bounding S} with the distribution $P^N$. Lastly, we complete our argument by presenting a proof of Lemma \ref{lem:bounding S} in the following.
%

\begin{proof}[Proof of Lemma \ref{lem:bounding S}]
From the definition of $\{S^i_t\}_{t>0},$ we have that
\begin{align*}
S^i_t = \sum_{r=1}^{t\wedge \Tlepstldi} \frac{\J^i_r}{I^i} + \sum_{r=t\wedge \Tlepstldi+1}^{t\wedge \Tuepstldi} \frac{\log M_i}{I^i} \11\{\Htld^i_{r-1} \geq \addvb{\sqrt{\epsilon}}\}+\sum_{r=t\wedge \Tuepstld+1}^{t} \frac{\D_r^i}{D^i} + \addvb{\sqrt{\epsilon}} \frac{N}{I^i}\11\{t\geq \Tuepstld\}.
\end{align*}
Therefore, as $(T\vee \Tuepstldi) \geq  \Tuepstldi \geq  \Tlepstldi$, then
\begin{align*}
\EE\big[S^i_{T\vee \Tuepstldi}\big] &= \EE\bigg[ \sum_{r=1}^{\Tlepstldi} \frac{\J^i_r}{I^i} + \sum_{r=\Tlepstldi+1}^{\Tuepstldi} \frac{\log M_i}{I^i} \11\{\Htld^i_{r-1} \geq \addvb{\sqrt{\epsilon}}\}+\sum_{r=\Tuepstld+1}^{T\vee \Tuepstldi} \frac{\D_r^i}{D^i}\bigg] + \addvb{\sqrt{\epsilon}} \frac{N}{I^i}.
\end{align*}
As for the first summation, after multiplying and dividing by $\EE[\Tlepstldi]$, we have that
\begin{align*}
\EE\Big[\sum_{r=1}^{\Tlepstldi} \frac{\J^i_r}{I^i}\Big] = \frac{\EE[\Tlepstldi]}{I^i}  \bigg(\frac{1}{\EE[\Tlepstldi]} \EE\Big[\sum_{r=1}^{\Tlepstldi} \J^i_r\Big] \bigg) = \EE[\Tlepstldi],
\end{align*}
where the last equality follows by setting $I^i$ as in the statement of the lemma. Similarly, the third summation is bounded as in the following 
\begin{align*}
\EE\Big[\sum_{r=\Tuepstld+1}^{T\vee \Tuepstldi} \frac{\D_r^i}{D^i}\Big] = \frac{\EE[T-\Tlepstldi]}{D^i}  \bigg(\frac{1}{\EE[T-\Tlepstldi]}\EE\Big[\sum_{r=\Tuepstld+1}^{T\vee \Tuepstldi} \D_r^i\Big] \bigg) = \EE[T-\Tlepstldi],
\end{align*}
where the first equality holds after multiplying and dividing by $\EE[T-\Tlepstldi]$, and the second equality follows by setting $D^i$ as in the statement of the lemma. As a result,
\begin{align*}
\EE\big[S^i_{T\vee \Tuepstldi}\big] & \leq \EE[T] + \frac{\log M_i}{I^i}\EE\bigg[ \sum_{r=\Tlepstldi+1}^{\Tuepstldi}  \11\{\Htld^i_{r-1} \geq \addvb{\sqrt{\epsilon}}\}\bigg] + \addvb{\sqrt{\epsilon}} \frac{N}{I^i}\\\numberthis \label{eq:bounding S eq 1}
&{\leq} \EE[T] + \frac{\log M_i}{I^i}\EE\bigg[ \sum_{r=\Tlepstldi+1}^{N}  \11\{\Htld^i_{r-1} \geq \addvb{\sqrt{\epsilon}}\}\bigg] + \addvb{\sqrt{\epsilon}} \frac{N}{I^i}
\end{align*}
where the inequality follows as $\Tuepstldi \leq N$.  Next, we bound the remaining summation. By iterative expectation we have that
\begin{align*}
 \EE\Big[\sum_{r=\addvb{\Tlepstldi}+1}^{N} \11\{\Htld^i_{r-1} &\geq \addvb{\sqrt{\epsilon}}\}\Big] = \EE_{\addvb{\Tlepstldi}}\Big[\sum_{r=\addvb{\Tlepstldi}+1}^{N} \EE\big[\11\{\Htld^i_{r-1} \geq \addvb{\sqrt{\epsilon}}\} ~\big| \addvb{\Tlepstldi} \big]\Big]\\
&\stackrel{(a)}{\leq}  \EE_{\addvb{\Tlepstldi}}\Big[\sum_{r=\addvb{\Tlepstldi}+1}^{N} \PP\Big( \sup_{\addvb{\Tlepstldi} \leq t \leq N-1 }\Htld^i_{t} \geq \addvb{\sqrt{\epsilon}} ~\big| \addvb{\Tlepstldi} \Big)\Big]\\
&=  \EE_{\addvb{\Tlepstldi}}\Big[\PP\Big( \sup_{ \addvb{\Tlepstldi} \leq t \leq N-1 }\Htld^i_{t} \geq \addvb{\sqrt{\epsilon}} ~\big|\addvb{\Tlepstldi}\Big) \sum_{r=\addvb{\Tlepstldi}+1}^{N} \11 \Big]\\\numberthis \label{eq:bounding S eq 2}
&\stackrel{(b)}{\leq}  \EE_{\addvb{\Tlepstldi}}\Big[N  \PP\Big( \sup_{ \addvb{\Tlepstldi} \leq t \leq N-1 }\Htld^i_{t} \geq \addvb{\sqrt{\epsilon}}~\big| \addvb{\Tlepstldi} \Big)\Big] \stackrel{(c)}{=}  N \PP\Big( \sup_{ \addvb{\Tlepstldi} \leq t \leq N-1 }\Htld^i_{t} \geq \addvb{\sqrt{\epsilon}} \Big),
\end{align*}
where (a) follows from taking the supremum over all  $\Htld^i_{r-1}$ appearing in the summation. Inequality (b) follows as the summation is less than $N-\addvb{\Tlepstldi}$ which is smaller than $N$. Lastly, (c) holds by taking the expectation of the conditional probability. We proceed with the following lemma which is a variant of Doob's maximal inequality for super-martingales, where a proof is provided in Appendix \ref{app:super martingale}.
\begin{lem}[Maximal Inequality for Supermartingales]\label{lem:super martingale}
Let $\{M_t\}_{t>0}$ be a non-negative supermartingale w.r.t a filtration $\{\mathcal{F}_t\}_{t>0}$ . If $\tau$ is a bounded stopping time w.r.t this filtration, then  the following inequality holds for any constant $c>0$
\begin{align*}
\prob{\sup_{t\geq \tau} M_t >c } \leq \frac{\EE[M_\tau]}{c}.
\end{align*}
\end{lem}

Note that $\{\Htld^i_t\}_{t>0}$ is a super martingale. Therefore, from Lemma \ref{lem:super martingale}, we have that 
\begin{align*}
\PP\Big( \sup_{ \addvb{\Tlepstldi}  \leq t \leq N-1 }\Htld^i_{t} \geq \addvb{\sqrt{\epsilon}} \Big) \leq \frac{\EE[\Htld^i_{\addvb{\Tlepstldi} }]}{\addvb{\sqrt{\epsilon}}}
\end{align*}
If $\addvb{\Tlepstldi} <N$, then by definition of this stopping time  $\Htld^i_{\addvb{\Tlepstldi}} \leq \epsilon$; otherwise $\addvb{\Tlepstldi} = N$ which implies that $\Htld^i_{\addvb{\Tlepstldi}} = \Htld^i_{N}$. However, as $T\leq N$, then 
\begin{align*}
\EE[\Htld^i_{N}]\leq \EE[\Htld^i_T] \leq h_b(P_e) + P_e\log M_1M_2 \leq \epsilon,
\end{align*}
where the second inequality follows from Fano's and the last inequality holds as $P_e \ll \epsilon$. Consequently,
\begin{align*}
\PP\Big( \sup_{ \addvb{\Tlepstldi} \leq t \leq N-1 }\Htld^i_{t} \geq \addvb{\sqrt{\epsilon}} \Big) \leq \frac{\epsilon}{\addvb{\sqrt{\epsilon}}}= \addvb{\sqrt{\epsilon}}.
\end{align*}
Therefore, using this inequality in \eqref{eq:bounding S eq 2} we obtain that 
\begin{align*}
\EE\Big[\sum_{r=\addvb{\Tlepstldi}+1}^{N} \11\{\Htld^i_{r-1} \geq \addvb{\sqrt{\epsilon}}\}\Big] \leq N \addvb{\sqrt{\epsilon}}.
\end{align*}
Thus, from \eqref{eq:bounding S eq 1}, we obtain that 
\begin{align*}
\EE[S^i_{T\vee \Tuepstldi}]  & \leq \EE[T] + \frac{\log M_i}{I^i} \big(\addvb{\sqrt{\epsilon} N} \big)+ \addvb{\sqrt{\epsilon}}\frac{N}{I^i}.
\end{align*}
Hence, factoring $\EE[T]$ gives the following inequality 
\begin{align*}
\EE[S^i_{T\vee \Tuepstldi}]  & \leq \EE[T]( 1 + V(\epsilon, N)),
\end{align*}
where $V(\epsilon,N) =\frac{R_i}{I^i} \big( \addvb{\sqrt{\epsilon}} N \big)+ \addvb{\sqrt{\epsilon}} \frac{N}{\EE[T] I^i}.$ Further, one can see that by setting $\epsilon=\epsilon_N = \frac{1}{N^{2+\zeta}}$ for some $\zeta\in (0,1)$ it holds that $\lim_{N \rightarrow \infty} V(\epsilon_N, N)=0$. Hence, the proof is complete. 
\end{proof}

%% file: TwoPhaseLB.tex
%


\subsection{Two-phase scheme}\label{subsec:two phase}

We build upon Yamamoto-Itoh transmission scheme for \ac{ptp} channel coding with feedback \cite{Yamamoto}. The scheme sends the messages $W_1, W_2$ through blocks of length $n$. The transmission process is performed in two stages: 1) The ``data transmission" stage taking up to $n(1-\gamma)$ channel uses, 2) The ``confirmation" stage taking up to $n\gamma$ channel uses, where $\gamma$ is a design parameter taking values from  $[0,1]$.

\paragraph*{\textbf{Stage 1}}
For the first stage, we use any coding scheme that achieves the feedback capacity of the MAC. The length of this coding scheme is at most $n(1-\gamma)$.  Let $\hat{W}_1, \hat{W}_2$ denote the decoder's estimation of the messages at the end of the first stage. Define the following random variables:
$\Theta_i=1\{\hat{W}_i \neq W_i\}, \quad i=1,2$.
%
Because of the feedback, $\hat{W}_1$ and $\hat{W}_2$ are known at each transmitter. 
Therefore, at the end of the first stage, transmitter $i$ has access to $W_i, \hat{W}_1, \hat{W}_2$, and $\Theta_i$, where $i=1,2$.  

\paragraph*{\textbf{Stage 2}}
The objective of the second stage is to inform the receiver whether the hypothesis $\mathcal{H}_0: (\hat{W}_1, \hat{W}_2)=(W_1,W_2)$ or  $\mathcal{H}_1: (\hat{W}_1, \hat{W}_2)\neq (W_1,W_2)$ is correct. 
Each transmitter employs a code of size two and length $\gamma n$. 
The codebook of user $i$ consists of two codewords $\{\underline{x_i}(0),\underline{x_i}(1)\}$. User $i$ transmits codeword $\underline{x_i}(\Theta_i)$. These two codebooks are not generated independently. Instead, the four codewords $(\underline{x_1}(0),\underline{x_2}(0), \underline{x_1}(1), \underline{x_2}(1))$ are selected randomly among all the sequences with joint-type $\mathsf{P}_{\gamma n}$ defined over the set $\mathcal{X}_1 \times \mathcal{X}_2\times \mathcal{X}_1 \times \mathcal{X}_2$ and for sequences of length $\gamma n$.

\paragraph*{\textbf{Decoding}}
Upon receiving the channel output, the receiver estimates $\Theta_1, \Theta_2$. Denote this estimation by $\hat{\Theta}_1, \hat{\Theta}_2$. 
If $(\hat{\Theta}_1, \hat{\Theta}_2) = (0,0)$, then the decoder declares that the hypothesis  $\mathcal{H}_0$ has occurred. Otherwise, $\mathcal{H}_1$ is declared. Because of the feedback, this decision is also known at each encoder. Therefore, under $\mathcal{H}_0$, and if $(\hat{\Theta}_1, \hat{\Theta}_2) = (0,0)$ transmission stops and a new data packet is transmitted at the next block. Otherwise, the message is transmitted again at the next block. The process continues until $(\hat{\Theta}_1, \hat{\Theta}_2) = (0,0)$ occurs. 

We use the log-likelihood decoder in this stage. Let $m=\gamma n$. The decision region for $\mathcal{H}_0$  in the log-likelihood decoder is define as
\begin{align*}
\mathcal{A}:=\set{y^m\in \mathcal{Y}^m: \log \frac{Q^m(y^m| \underline{x_1}(0), \underline{x_2}(0)) }{Q^m(y^m| \underline{x_1}(a_1), \underline{x_2}(a_2))}\geq \lambda, ~~(a_1,a_2)=(0,1), (1,0), (1,1)},
\end{align*}
where $\lambda$ is the decision threshold to be adjusted and $Q^m(y^m | x_1^m, x_2^m)=\prod_i Q(y_i|x_{1,i}, x_{2,i})$ for any $y^m \in \mathcal{Y}^m, x_1^m \in \mathcal{X}_1^m, x_2^m \in \mathcal{X}_2^m$. 

The confirmation stage in the proposed scheme can be viewed as a decentralized binary hypothesis problem in which a binary hypothesis $\{\mathcal{H}_0, \mathcal{H}_1\}$ is observed partially by two distributed agents and the objective is to convey the true hypothesis to a central receiver. This problem is qualitatively different from the sequential binary hypothesis testing problem as identified in \cite{Berlin} for \ac{ptp} channel. Note also that in the confirmation stage we use a different coding strategy than the one used in Yamamoto-Itoh scheme \cite{Yamamoto}. Here, all four codewords have a joint-type  $\mathsf{P}_{\gamma n}$. It can be shown that repetition codes, and more generally, constant composition codes are strictly suboptimal in this problem.  Based on this scheme, we can derive the lower bound in Corollary \ref{cor:twophase lb} in Section \ref{subsec:LB}.
Example \ref{ex:two ptp} shows the potential problem with the ``Two-Phase'' scheme. Indeed, to achieve the reliability function of the channel, each transmitter transitions from the data transmission to the data confirmation stage at a different time, while the proposed ``Two-Phase'' scheme enforces the same time for this transition. Motivated by this observation, we now propose a ``Three-Phase'' scheme.

%% file: ThreePhaseLB.tex
\subsection{Three-phase scheme}\label{subsec:threephase lb}

This section explains a more sophisticated scheme than the two-phase scheme described in Section \ref{subsec:two phase}. We propose a three-phase scheme that allows the encoders to finish the data transmission phase at different times. Let $n$ be the total number of channel uses and $\gamma_1, \gamma_2, \gamma_3 >0$ be the protocol's parameters corresponding to the length of each phase and  satisfying $\gamma_1+\gamma_2+\gamma_3=1$. Moreover, suppose user $i$'s message, $W_i$, takes values from $\set{1,2,\cdots, M_i}$, where $M_i=2^{nR_i}$ with $R_i\geq 0, i=1,2$.  For simplicity of the presentation, we assume that $n\gamma_i$ and $2^{nR_i}$ are integers for $i=1,2,3$.

The transmission is performed in three stages: (i) full data transmission taking up to $\gamma_1n$ channel uses; (ii) data-confirmation hybrid phase, where one encoder starts the confirmation while the other is still in the data transmission stage. This stage takes up to $\gamma_2n$ channel uses; (iii) full confirmation stage, where both users are in the confirmation taking $\gamma_3n$ channel uses.

Before the communication starts, the coding strategy determines $\gamma_i$ and selects which user should stop data transmission at the end of the first stage.  Without loss of generality, the first encoder starts the confirmation at the second stage while the second user continues data transmission until the end of the second stage. In this case, the second user splits its message into two parts, each to be communicated over stages one and two. 

\paragraph*{\textbf{Stage 1}}  The first encoder sends its entire message during this stage; while the second user splits its message into two parts with sizes $2^{nR_{2,1}}$ and $2^{nR_{2,2}}$.    Given that the length of this stage is $n\gamma_1$, the effective transmission rate during the first stage is $(\frac{R_1}{\gamma_1}, \frac{R_{2,1}}{\gamma_1})$. The effective rates must be inside the feedback capacity region for reliable communication during this stage.   We use a fixed-length capacity achieving code to send messages over $n\gamma_1$ uses of the MAC. At the end of this stage, the first user finishes data transmission. Let $\hat{W}_1^{1}$ be the decoder's estimate of $W_1$ at the end of this stage. Let $\Theta_1=\11\{\hat{W}_1^{1} \neq W_1\}$ be the indicator on whether $W_1$ is decoded correctly or not. 

\paragraph*{\textbf{Stage 2}} 
During the second phase, the second user transmits the second part of its message while the first user sends one bit of confirmation. Therefore, the effective transmission rate in this stage is $(0, \frac{R_{2,2}}{\gamma_2})$. As the rate of the first user is zero,  the second user can potentially send information at a higher rate than during the first stage. The first user employs a repetition code to inform the decoder whether its message is decoded correctly. It sends the decoder one bit of information described by $\Theta_1$.  For that, this user sends the symbol $x_1^1(\Theta_1)$ repeatedly for $n\gamma_2$ times. As a result, the second user faces a \ac{ptp} channel with a state that depends on $x_1^1(\Theta_1)$. This user assumes that $\Theta_1=0$ and  employs a Shannon \ac{ptp} capacity-achieving code for this channel. The case $\Theta_1=1$ is ignored by  the second user because when $\Theta_1=1$, the first user's message is decoded erroneously, and the communication must be restarted. In this case, it does not matter whether the second user's message is decoded correctly or not. Therefore, the channel's capacity from the second user's perspective is $$C_2(x_1^1(0)):= \max_{P_{X_2}} I(X_2; Y | X_1= x_1^1(0)).$$
Note that the feedback does not increase the capacity of this \ac{ptp} channel. The second user operates near the capacity by setting $R_{2,2}=C_2(x_1^1(0))-\epsilon$ for some $\epsilon>0$ sufficiently small. 
Let $\hat{W}^2_2$ be the overall decoded message of the second user at the end of this stage. Further, let $\Theta_2=\11\{\hat{W}_2^{2} \neq W_2\}$ as the indicator on whether $W_2$ is decoded correctly or not. 

\paragraph*{\textbf{Stage 3}} 
Both encoders are in the confirmation phase during the third stage. They inform the receiver whether the  messages are decoded correctly. We use the same strategy in this stage as in the two-phase scheme. The codebook of user $i$ consists of two codewords $\{\underline{x^2_i}(0),\underline{x^2_i}(1)\}$. User $i$ transmits codeword $\underline{x^2_i}(\Theta_i)$. The four codewords $(\underline{x^2_1}(0),\underline{x^2_2}(0), \underline{x^2_1}(1), \underline{x^2_2}(1))$ are selected randomly among all the sequences with joint-type $\mathsf{P}_{\gamma_3 n}$ defined over the set $\mathcal{X}_1 \times \mathcal{X}_2\times \mathcal{X}_1 \times \mathcal{X}_2$ and for sequences of length $\gamma_3 n$.

\paragraph*{\textbf{Decoding}}
Upon receiving the channel output, the receiver estimates $\Theta_1, \Theta_2$. Denote this estimation by $\hat{\Theta}_1, \hat{\Theta}_2$. 
If $(\hat{\Theta}_1, \hat{\Theta}_2) = (0,0)$, then the decoder declares that the hypothesis  $\mathcal{H}_0$ is occurred. Otherwise, $\mathcal{H}_1$ is declared. Because of the feedback, this decision is also known at each encoders. 
Therefore, under $\mathcal{H}_0$, and if $(\hat{\Theta}_1, \hat{\Theta}_2) = (0,0)$ transmission stops and a new data packet is transmitted at the next block. Otherwise, the message is transmitted again at the next block. The process continues until $(\hat{\Theta}_1, \hat{\Theta}_2) = (0,0)$ occurs.

We use the log-likelihood decoder for confirmation during the second and third stages. Let $m_2=\gamma_2 n$ and $m_3 = \gamma_3n$ be the length of the corresponding stages. 
Let $\Qbar{1}(y|x_1) \deq \sum_{x_2} P_{X_2}(x_2)Q(y|x_1, x_2)$ denote the effective channel from the first user's perspective during the second stage, where $P_{X_2}$ is the single-letter probability distribution based on which the codewords of the second user are generated. With this definition, the decision region for $\mathcal{H}_0$  in the log-likelihood decoder is given by 
\begin{align*}
\mathcal{A}:=\Big\{y^m\in \mathcal{Y}^m: \log \frac{\Qbar{1, m_2}(y^{m_2}| x^1_1(0))}{\Qbar{1, m_2}(y^{m_2}| x^1_1(1))} &+\log \frac{Q^{m_3}(y_{m_2+1}^{m}| \underline{x_1}(0), \underline{x_2}(0)) }{Q^{m_3}(y_{m_2+1}^{m}| \underline{x_1}(a_1), \underline{x_2}(a_2))}\geq \lambda,\\
&\hspace{40pt}(a_1,a_2)=(0,1), (1,0), (1,1)\Big\},
\end{align*}
where $m=m_2+m_3$, and $\lambda$ is the decision threshold to be adjusted. 

 The case where the second user is selected to stop transmission at the end of stage (i) is the same as the role of the users interchanged.  In that case,  we have $x_2^1(0)$ and $x_2^1(1)$ as the symbols for the second user's repetition code with $C_1(x_2^1(0))$ being the capacity of the first user's channel in stage (ii). To maximize the error exponent, we optimize over the parameters $(\gamma_1, \gamma_2, \gamma_3)$,  the type  $\mathsf{P}_n$, and the repeating symbols $x_1^1(0), x_1^1(1), x_2^1(0), x_2^1(1)$. Based on this optimization,  we decide which user must stop transmission at the end of stage (i). 

\subsection{Exponent-rate region}
The analysis of the three-phase scheme is different from the two phase scheme. This is because of the second phase of the communication, where one user sends data and the other is in the confirmation stage (see Fig. \ref{fig:phase2}). In what follows, we consider this problem.

\begin{figure}[hbtp]
\centering
\includegraphics[scale=1]{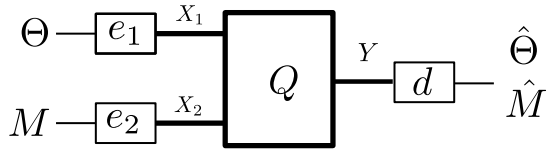}
\caption{Diagram of the second phase of the transmission. The first user performs a binary hypothesis for $\Theta$ testing while the second user transmits a message $M$}
\label{fig:phase2}
\end{figure}

Note that the hypothesis in this problem is highly asymmetric, as the probability of error in the decoded message converges to zero. Given this assumption, suppose the prior on the hypotheses are: $\PP(\mathcal{H}_1)=1-\PP(\mathcal{H}_0)\leq \zeta_n,$ where $\zeta_n$ is the probability of error for decoding the first message during the first stage of the transmission. By standard arguments, $\lim_{n\rightarrow \infty} \zeta_n=0$ as long as the transmission rates are inside the capacity. Therefore, we are interested only on bounding the exponent for the type-II error in decoding $\Theta$ and the transmission rate for the second user. 

\begin{lem}\label{lem:rate err exp}
With $n$ uses of the channel $Q$ as in Fig. \ref{fig:phase2}, let $R = \frac{\log M}{n},$ and $\beta =  \PP\big(\hat{\Theta}=1 ~|~ \Theta =0\big)$, then there exist $x_1(0), x_1(1) \in \mathcal{X}_1$ and a distribution $P_{X_2}$ on $\mathcal{X}_2$ such that 
\begin{align*}
R&\leq \min\set{I(X_2;Y|x_1(0)), I(X_2;Y|x_1(1))}, & \beta &\leq 2^{-nD(\bar{Q}(x_1(0)) \| \bar{Q}(x_1(1)))}.
\end{align*}
where, $\Qbar{1}(y| x_1)=\sum_{x_2} P_{X_2}(x_2)Q(y|x_1, x_2)$ for all $x_1 \in \mathcal{X}_1$ and $y\in \mathcal{Y}$.
\end{lem}
\begin{proof}
The first encoder is based on a repetition code: $e_1(\theta)=x_1(\theta)^n$. For the second encoder, we use standard random codes generated \ac{iid} based on $P_{X_2}$. 
Independent decoders are used for $M$ and $\Theta$. The loglikelihood decoder is used for $\Theta$ with the decision region as 
\begin{align*}
\mathcal{A}:=\set{y^n\in \mathcal{Y}^n: \frac{1}{n}\log \frac{\bar{Q}^{n}(y^{n}| x_1(0))}{\bar{Q}^{n}(y^{n}| x_1(1))} \geq \lambda },
 \end{align*} 
 where $\bar{Q}(y|x_1):=\sum_{x_2}P_{X_2}(x_2) Q(y|x_1, x_2)$ for all $y\in\mathcal{Y}$ and $x_2\in\mathcal{X}_2$. If $y\notin \mathcal{A}$, then $\hat{\Theta}=1$ is declared and the decoding is aborted (alternatively, we set $\hat{m}=1$). If $y\in \mathcal{A}$, then $\hat{\Theta}=0$ and we continue to decoding $M$.  For decoding $M$ we find $\hat{m}$ such that $(x_2^n(\hat{m}), y^n)$ are jointly typical with respect to $P_{X_2}Q(\cdot | x_1(0), \cdot )$.
Next, we analyze $\beta$:
\begin{align*}
\beta &= \prob{\mathcal{A}| x_1^n(1)} = \sum_{m} \frac{1}{m} \sum_{x^n_2} P_{X_2}^n(x_2^n) \sum_{y^n \in \mathcal{A}} Q^n(y^n | x^n_1(1), x_2^n)=\sum_{y^n \in \mathcal{A}} \bar{Q}^n(y^n | x^n_1(1)).
\end{align*}
  Let $\hat{P}_{y^n}$ be the empirical type (distribution) of $y^n$. Then, the condition in $\mathcal{A}$ is equivalent to $$[D(\hat{P}_{y^n}\| \bar{Q}(x_1(1))) - D(\hat{P}_{y^n}\| \bar{Q}(x_1(0)))]\geq \lambda.$$ Therefore, from the Chernoff-Stein lemma, it is not difficult to see that $\beta \approx 2^{-nD(\bar{Q}(x_1(0)) \| \bar{Q}(x_1(1))}$.
  
Since the first user employs a repetition code, the second user faces a memoryless \ac{ptp} channel with unknown state $\Theta$. Therefore, from standard arguments, $M$ is correctly decoded with probability converging to one as $n\rightarrow \infty$, if $R$ is inside the capacity of the channel seen from the second user. That capacity is expressed as $\min\set{I(X_2;Y|x_1(0)), I(X_2;Y|x_1(1))}$.
\end{proof}

\begin{remark}\label{rem:rate err exp}
Note that for the analysis of the three-phase scheme, we are only interested in the error for decoding $M$ when $\Theta=0$. The reason is that $\Theta=1$ implies that the message during the first stage of transmission is decoded incorrectly. Hence, decoding the message in the second phase is not of interest as a re-transmission must occur.  As a result, the data transmission rate in the Lemma \ref{lem:rate err exp} can be made slightly larger to $R\leq I(X_2;Y|x_1(0)).$
\end{remark}

With this scheme, we establish Theorem \ref{thm: Error Exp lowerbound three phase} in Section \ref{subsec:LB}, with the full proof given  in Appendix \ref{proof: Error Exp lowerbound three phase}.

%% file: VLExample.tex
\section{Illustrative Examples}\label{app:VL example}
\begin{example}\label{exp:tree VH}
Let us consider an example for illustration. Let $Y_t, t=1,2,3,$ be a sequence of \ac{iid} Bernoulli random variables with uniform distribution. Let $T$ be the first time $Y_t=1$, that is $T=\min\set{1\leq t\leq 3: Y_t=1}$. Then, $Y^T$ is characterized by the following table:
\begin{center}
\begin{tabular}{|c|cccc|}
 \toprule
$Y^T$ & $1$ & $01$ & $000$ & $001$\\
\midrule
$P(Y^T)$ & $\frac{1}{2}$ & $\frac{1}{4}$ & $\frac{1}{8}$ & $\frac{1}{8}$\\\bottomrule\hline
\end{tabular}
\end{center}
Fig. \ref{fig:TreeVH} shows the tree of the realizations of $Y^T$ and its entropy quantities. The total entropy of $Y^T$ is calculated by  
\begin{align*}
 H(Y^T) = &\prob{Y^T=1} H(Y_1) + \prob{Y^T=01} (H(Y_1)+H(Y_2|Y_1=0))\\
 &+\prob{Y^T=000} (H(Y_1)+H(Y_2|Y_1=0)+H(Y_3|Y_1=0, Y_2=0))\\
 &+\prob{Y^T=001} (H(Y_1)+H(Y_2|Y_1=0)+H(Y_3|Y_1=0, Y_2=0))
 \end{align*} 
 Therefore, we have that 
$ H(Y^T) =  \frac{7}{4}.$
  Note that $H(T) = H([\frac{1}{2}, \frac{1}{4}, \frac{1}{4}])= \frac{3}{2},$ and 
 \begin{align*}
 H(Y^T | T) &= \frac{1}{2} H(Y_1| T=1)+\frac{1}{4} H(Y_1, Y_2| T=2)+\frac{1}{4} H(Y_1, Y_2, Y_3| T=3) = 0+0+\frac{1}{4}=\frac{1}{4}.
 \end{align*}
 Hence, $H(Y^T) = \frac{3}{2}+\frac{1}{4} = \frac{7}{4}$.
\end{example}
\begin{figure}
    \centering
    \includegraphics[scale=2]{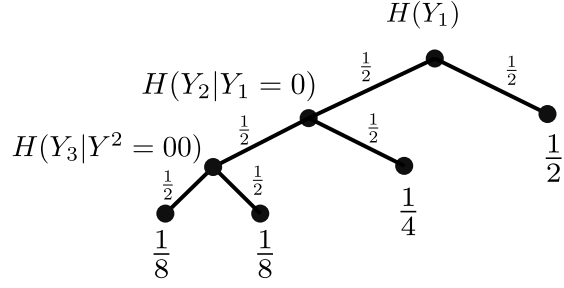}
    \caption{The tree of realizations of $Y^T$, in Example \ref{exp:tree MI}, with left branches denoting $y_t=0$  and the right branches denoting $y_t=1$. }
    \label{fig:TreeVH}
\end{figure}

\begin{example}\label{exp:tree MI}
Let us consider an example for illustration. Consider the sequence of Boolean random variables $(X_t, Y_t, Z_t), t=1,2,...$. Define the stopping time $T=\min\set{t\in \{1,2,3\}: y_t=1}$. Figure \ref{fig:Tree Example} shows the tree of realizations of $Y^T$, with left branches denoting $y_t=0$  and the right branches denoting $y_t=1$. Corresponding to each non-leaf node is a mutual information $I(X^t; Y_t | Z^t, y^{t-1})$. The weights represent the probability of each realization as 
\begin{align*}
    a=\prob{Y_1=0}, \qquad b=\prob{Y^2=00}, \qquad c=\prob{Y^3=000}.
\end{align*}
Therefore, 
\begin{align*}
I({X}^T \rightarrow {Y}^T \| {Z}^T) = I(X_1; Y_1 | Z_1) + a I(X_1 {X}_2; {Y}_2~ |~ {Z}^2, Y_1=0)+b I({X}^3; {Y}_3~ |~ {Z}^3, Y^2=00).  
\end{align*}
\end{example}
\begin{figure}
    \centering
    \includegraphics[scale=2]{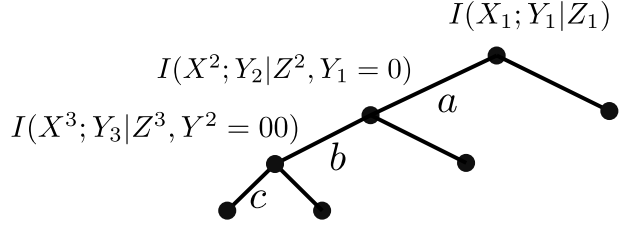}
    \caption{The tree of realizations of $Y^T$, in Example \ref{exp:tree MI}, with left branches denoting $y_t=0$  and the right branches denoting $y_t=1$. }
    \label{fig:Tree Example}
\end{figure}

%% file: proofs/Proof_thm_capacity.tex
\section{Proof of Theorem \ref{thm:VLC Capacity}}\label{proof:thm: VLC Capacity}
From Definition \ref{def:capcity}, consider an achievable rate pair $(R_1,R_2)$. Based on the definition of achievability, consider any $(M_1, M_2,N)$-VLC with the probability of error $P_e$ and a stopping time $T$ that is less than $N$ almost surely such that $R_i \leq \frac{\log M_i}{\EE[T]}, i=1,2$. 
Then, Fano's inequality implies that
$H(W_1, W_2 | Y^T)  \leq \alpha(P_e)$,
where $\alpha(P_e)= h_b(P_e) + P_e \log M_1 M_2.$ As conditioning reduces the entropy, we get a similar bound for the other  entropy terms: 
$$H(W_1|W_2, Y^T)\leq \alpha(P_e), \qquad H(W_2 | W, Y^T)\leq \alpha(P_e).$$

Next, we start with bounding the rate $R_1$. Since at time $t=0$, the messages $(W_1,W_2)$ are independent with uniform distribution, then we have that   
\begin{align*}
\log M_1 &= H(W_1) = H(W_1|W_2) = I(W_1; Y^T | W_2) + H(W_1|W_2, Y^T)\\\numberthis\label{eq: bound on M1}
&\leq I(W_1; Y^T | W_2) + \alpha(P_e).
\end{align*}
We proceed by showing that 
$I(W_1; Y^T | W_2) \leq I(X_1^T \rightarrow Y^T \| X_2^T)$.
We first pad $Y^T$ to make it a sequence of length $N$. Let $\xi$ be an auxiliary symbol and define 
$$Y^N =(Y_1, Y_2, \cdots, Y_T, \xi, \xi, \cdots, \xi).$$
Similarly, we extend the encoding functions and the channel's transition probability to include $\xi$. Specifically, after the stopping time $T$, the encoders send the constant symbol $\xi$ and the channel outputs $\xi$ to the receiver. More precisely,  
\begin{align*}
e_i(W_i, Y^n) = \xi, ~~~\forall n\geq T, i=1,2, \qquad Q(y|\xi, x_1) = Q(y|x_2, \xi) = \11\set{y=\xi}, ~~~ \forall x_1, x_2, y.
\end{align*}
This auxiliary adjustment is only for  tractability of the analysis as it does not affect the performance of the code. Specifically, the mutual information stays the same by replacing $Y^T$ with $Y^N$:
\begin{align*}
 I(W_1; Y^N|W_2) &= I(W_1; Y^T |W_2) + I(W_1; Y^N | W_2, Y^T) \\\numberthis \label{eq:Mu eq 1}
 &=  I(W_1; Y^T |W_2) + I(W_1; \xi_{T+1}^N | W_2, Y^T)  = I(W_1; Y^T |W_2).
 \end{align*} 
From the chain rule, we have that
\begin{align*}
I(W_1; Y^N|W_2)& = \sum_{r=1}^N I(W_1; Y_r|W_2, Y^{r-1})\\
& = \sum_{r=1}^N H(Y_r | W_2, X_2^r, Y^{r-1}) - H(Y_r |W_1, X_{1,r}, W_2, X_2^r, Y^{r-1})\\
& \stackrel{(a)}{\leq} \sum_{r=1}^N H(Y_r | X_2^r, Y^{r-1}) - H(Y_r |W_1, X_{1,r}, W_2, X_2^r, Y^{r-1})\\
& \stackrel{(b)}{=} \sum_{r=1}^N H(Y_r |X_2^r, Y^{r-1}) - H(Y_r |X_{1,r}, X_2^r, Y^{r-1})\\
&\stackrel{(c)}{=} I(X_1^N \rightarrow Y^N \| X_2^N)
\end{align*}
where (a) holds as conditioning reduces the entropy, and (b) holds because conditioned on $X_{1,r}, X_{2,r}$ the channel's output $Y_r$ is independent of the messages $(W_1,W_2)$. Lastly, (c) is due to the definition of directed mutual information as in \eqref{eq:directed MI}.
Next, we show that the directed mutual information above equals to the following: 
\begin{align*}
I(X_1^N \rightarrow Y^N \| X_2^N) =  I(X_{1}^T \rightarrow Y^T \|~ X_2^T) = \EE\Big[\sum_{r=1}^T I(X_{1,r}; Y_r | X_2^r, \mathcal{F}_{r-1})\Big],
\end{align*}
where the second equality is due to the definition given in \eqref{eq:directed MI T}.
Note that $I(X_{1,r}; Y_r | X_2^r, \mathcal{F}_{r-1}) = 0$ almost surely for any $r>T$ as $Y_r=\xi$. Therefore, we have that
\begin{align*}
 I(X_{1}^T \rightarrow Y^T \|~ X_2^T) & =  \EE\Big[\sum_{r=1}^N I(X_{1,r}; Y_r | X_2^r, \mathcal{F}_{r-1})\Big]\\\numberthis\label{eq:Mu eq 3}
 & =\sum_{r=1}^N  \EE\big[I(X_{1,r}; Y_r | X_2^r, \mathcal{F}_{r-1})\big]  
  = I(X_1^N \rightarrow Y^N \| X_2^N).
\end{align*}
Therefore, combining \eqref{eq: bound on M1}-\eqref{eq:Mu eq 3} gives an upper bound on $\log M_1$. Dividing both sides by $\EE[T]$ gives the following upper bound on $R_1$
\begin{align}\label{eq:R1 up}
R_1 \leq \frac{\log M_1}{\EE[T]}\leq  \frac{1}{\EE[T]}I(X_1^T \rightarrow Y^T \| X_2^T)+\frac{\alpha(P_e)}{\EE[T]}.
\end{align}
The first term above is the desired expression. The second term is vanishing as $P_e \rightarrow 0$. 
Based on a similar argument, we can bound  $R_2$ and $R_1+R_2$ as in the following 
\begin{align*}
R_2 &\leq \frac{1}{\EE[T]} I(X_{2}^T \rightarrow Y^T \|~ X_1^T) +\frac{\alpha(P_e)}{\EE[T]},\\\numberthis\label{eq:R2 R1+R2 up}
R_1+R_3 &\leq \frac{1}{\EE[T]} I(X_{1}^T, X_2^T \rightarrow Y^T) +\frac{\alpha(P_e)}{\EE[T]}.
\end{align*}
Note that the upper bounds in \eqref{eq:R1 up} and \eqref{eq:R2 R1+R2 up} hold for any VLC as long as $R_i \leq \frac{\log M_i}{\EE[T]}$. In view of Definition \ref{def:capcity}, we consider a sequence of $(M_1^{(n)},M_2^{(n)}, N^{(n)})$-VLCs with $P_e^{(n)}\rightarrow 0$ as $n\rightarrow \infty$. 
Moreover, the normalized mutual information quantities in \eqref{eq:R1 up} and \eqref{eq:R2 R1+R2 up} belong to $\CMAC$.
Hence, the proof is completed as the above argument implies that $(R_1,R_2)\in \CMAC$.

%% file: SupportingHP.tex
\section{Characterization of the Feedback Capacity Region Via Supporting Hyperplanes} \label{app:CFB supporting HYP}



From Theorem \ref{thm:VLC Capacity}, $\CMAC$ is the outer bound on the feedback-capacity region. In what follows, we provide an alternative characterization for our achievability scheme. 
Let $\mathcal{Q}$ be a finite set. It suffices to assume $|\mathcal{Q}|= 3$. For any positive integer $L$, define  $\mathcal{PQ}_L$ as the set of 
all joint distributions $P_{QX_1^LX_2^LY^L}$ such that the conditioned distribution $P_{X_1^LX_2^LY^L|Q=q}\in \mathcal{P}_L$ for any $q\in \mathcal{Q}$.

\begin{align*}
\mathcal{C}'_{FB}=\mbox{cl}\bigg[&\bigcup_{N\in \NN}\bigcup_{P\in \mathcal{PQ}_N}\bigcup_{T\leq N} \bigg\{(R_1, R_2)\in [0,\infty)^2: R_1 \leq \frac{1}{\EE[T]} I(X^T_1 \rightarrow Y^T|| X^T_2~|Q),\\\numberthis \label{eq:C_FB prime}
 R_2 &\leq \frac{1}{\EE[T]} I(X^T_2 \rightarrow Y^T|| X^T_1~|Q), \quad  R_1+R_2 \leq \frac{1}{\EE[T]} I(X^T_1 X^T_2 \rightarrow Y^T|Q) \bigg\}\bigg],
\end{align*}
where $\mbox{cl}$ means closure, and $T$ is a stopping time with respect to the filtration of $y^t, t>0$. 

\begin{lem}\label{lem:C_FB prime}
$\mathcal{C}'_{FB}=\CMAC$.
\end{lem}
\begin{proof}
Clearly, $\CMAC\subseteq \mathcal{C}'_{FB}$. Hence, it suffices to show that $\mathcal{C}'_{FB}\subseteq \CMAC$.  Let $(R_1, R_2)$ be a rate-pair inside $\mathcal{C}'_{FB}$. Then, there exists $N>0, P\in \mathcal{PQ}_N$ and stopping time $T\leq N$ such that the inequalities in \eqref{eq:C_FB prime} are satisfied. Note that, 
\begin{align*}
I(X^T_1 \rightarrow Y^T|| X^T_2~|Q) &= \sum_{q\in \mathcal{Q}} P_Q(q) I(X^T_1 \rightarrow Y^T|| X^T_2~|Q=q)\\
&=\sum_{q\in \mathcal{Q}} P_Q(q) I(X^T_1 \rightarrow Y^T|| X^T_2)_{P_{X_1^NX_2^NY^N|Q=q}}.
\end{align*}
Similarly, other mutual information quantities are decomposed. By assumption, $P_{X_1^NX_2^NY^N|Q=q}\in \PMAC^N$. Then, we can write 
\begin{align*}
(R_1, R_2)=\sum_{q\in \mathcal{Q}}P_Q(q) (R_{1,q}, R_{2,q}),
\end{align*}
where $(R_{1,q}, R_{2,q})$ are rate-pair satisfying 
\begin{align*}
   R_{1,q} &\leq \frac{1}{\EE[T]} I(X^T_1 \rightarrow Y^T|| X^T_2~|Q=q),\\
 R_{2,q} &\leq \frac{1}{\EE[T]} I(X^T_2 \rightarrow Y^T|| X^T_1~|Q=q),\\
  R_{1,q}+ R_{2,q} &\leq \frac{1}{\EE[T]} I(X^T_1 X^T_2 \rightarrow Y^T|Q=q).
 \end{align*}
This implies that $(R_{1,q}, R_{2,q})\in \CMAC$. Hence, their convex combination is also in $\CMAC$ implying that $\mathcal{C}'_{FB}\subseteq \CMAC$.
\end{proof}

We are ready to prove Theorem \ref{thm:capacity C lambda}.
For any $P\in \mathcal{PQ}_N$ and the corresponding stopping time $T\leq N$, define
\begin{align*}
E(P, T)\deq \bigg\{&(R_1, R_2, R_3)\in [0,\infty)^3: R_1 \leq  \frac{1}{\EE[T]} I(X^T_1 \rightarrow Y^T|| X^T_2~| Q),\\
 R_2 &\leq  \frac{1}{\EE[T]} I(X^T_2 \rightarrow Y^T|| X^T_1~| Q),\quad   R_3 \leq  \frac{1}{\EE[T]} I(X^T_1 X^T_2 \rightarrow Y^T | Q) \bigg\}.
\end{align*}
and let 
\begin{align*}
E\deq \mbox{cl}\bigg[\bigcup_{N\in \NN}\bigcup_{P\in \mathcal{PQ}_N}\bigcup_{T\leq N} E(P, T)\bigg].
\end{align*}
Note that $E$ is a convex set in $\RR^3$. This follows from the convexifying random variable $Q$  in  the mutual information quantities and from Caratheodory's theorem ( implying that it is sufficient to have $|\mathcal{Q}|=3$).
Let $\underline{\lambda} \in \Delta^3$, where $\Delta^3$ is the simplex of vectors $v\in [0,1]^3$ satisfying $\sum_iv_i=1$. Define, 
\begin{align*}
C_{\underline{\lambda}} \deq \sup_{N\in \NN}\sup_{P\in \mathcal{PQ}_N} \sup_{T\leq N} \frac{1}{\EE[T]}\Big(&\lambda_1 I(X^T_1 \rightarrow Y^T|| X^T_2~| Q)+\lambda_2 I(X^T_2 \rightarrow Y^T|| X^T_1~| Q)\\
&+\lambda_3 I(X^T_1 X^T_2 \rightarrow Y^T | Q)\Big)
\end{align*}
With that, let
\begin{align*}
E' \deq \bigg\{(R_1, R_2, R_3)\in [0,\infty)^3: \forall \underline{\lambda} \in \Delta^3, \sum_{i}\lambda_i R_i \leq C_{\underline{\lambda}}\bigg\}.
\end{align*}
Using the same techniques as in \cite{Salehi1978}, one can show that $E'=E.$ Next, define 
\begin{align*}
\mathcal{D}\deq \Big\{(R_1, R_2)\in [0,\infty)^2: (R_1, R_2, R_1+R_2)\in E\Big\}. 
\end{align*}
Note that $\mathcal{D}=\CMAC$ as follows from  Lemma \ref{lem:C_FB prime}, indicating that $\CMAC=\mathcal{C}'_{FB}=\mathcal{D}$. 
Consequently, we obtain another characterization of the outer bound region via supporting hyperplanes. That is the set of rate pairs $(R_1, R_2)$ such that for all $\underline{\lambda}\in \Delta^3$ the inequality holds:
$\lambda_1 R_1 + \lambda_2 R_2+\lambda_3(R_1+R_2) \leq \mathcal{C}_{\underline{\lambda}}$.
With that the proof of Theorem \ref{thm:capacity C lambda} is complete.

%% file: proofs/proof_lem_pruned_submartingale.tex
\begin{lem}\label{lem:pruned sub martingale 1}
Suppose a non-negative random process $\set{H_r}_{r>0}$ has the following properties with respect to a filtration $\mathcal{F}_r, r>0,$
\begin{subequations}
\begin{alignat}{4}\label{eq:H linear K1}
 \EE[H_{r+1} - H_r |\mathcal{F}_r]&\geq -k_{1,r+1},  & \qquad \text{if}~  H_r & \geq \epsilon, \\\label{eq:H log K2}
\EE[\log H_{r+1} - \log H_r |\mathcal{F}_r]&\geq -k_{2,r+1}      & \text{if}~   H_r & < \epsilon\\\label{eq:H log K3}
\abs{\log H_{r+1} - \log H_r} & \leq k_3    & &\\\label{eq:H K4}
\abs{ H_{r+1}-H_r} &\leq k_4   & &
\end{alignat}
\end{subequations}
where $k_{1,r}, k_{2,r}$ are non-negative random variables measurable w.r.t $\mathcal{F}_{r}$, and $k_{1,r}\leq k_{2,r}$ for all $r>0$. Moreover, $k_3, k_4$ are non-negative constants. Given $\epsilon \in (0,1), N\in \NN$, and $I\geq D>0$ define the random process $\set{Z}_{t>0}$ as 
\begin{align*}
Z_t \deq (\frac{H_t-\epsilon}{I})\11\{H_t\geq \epsilon\}+\big(\frac{\log H_t - \log \epsilon}{D}+f(\log \frac{H_t}{\epsilon}) \big) \11\{H_t<\epsilon\},
\end{align*}
where $f(y)=\frac{1-e^{\lambda y}}{\lambda D}$ with $\lambda >0$. Further define the random process $\set{S_t}_{t>0}$ as 
\begin{align*}
S_t \deq \sum_{r=1}^{t\wedge \Tleps} \frac{k_{1,r}}{I} + \sum_{r=t\wedge \Tleps+1}^{t\wedge \Tueps} \frac{k_4}{I} \11\{H_{r-1} \geq \addvb{\sqrt{\epsilon}}\}+\sum_{r=t\wedge \Tueps+1}^{t} \frac{k_{2,r}}{D} + \addvb{\sqrt{\epsilon}} \frac{N}{I}\11\{t\geq \Tueps\}
\end{align*}
Let $\set{t_n}_{n>0}$ be the prune time process defined in \eqref{eq:tn} with respect to $\set{H_r}_{r>0}$. Lastly define the random process $\set{\Zprune}_{n>0}$ as $\Zprune \deq Z_{t_n}+S_{t_n}.$ Then, for small enough $\lambda>0$ the process $\set{\Zprune}_{n>0}$ is a sub-martingale with respect to the time pruned filtration $\mathcal{G}_{n}, n>0$. 
\end{lem}

\begin{IEEEproof}
First, we point out that $L_n$ is measurable w.r.t $\mathcal{G}_n, n>0$. This follows from Proposition \ref{prop:measurable} and that $H_t$ is measurable w.r.t $\mathcal{F}_t$. Next, the objective is to prove $\EE[L_{n+1}- L_n|y^{t_n}] \geq 0$ almost surely for all $n\geq 1$ and $y^{t_n}$. We prove the lemma by considering three cases depending on $n$. 

\noindent\textbf{Case (a). $\{y^{t_n}: n< \Tleps\}$:}
This corresponds to the case where the channel output yields  $H_n \geq \epsilon.$
From the definition of $t_n$ in \eqref{eq:tn}, in this case $t_n=n$ and $t_{n+1}$ is either  $n+1$ if $n< \Tleps -1$ or $\Tueps$ if $n = \Tleps -1$. We use indicator functions to separate these two possibilities.  We first consider 
$(L_{n+1}-L_n) \11\set{n< \Tleps -1}$.
The random process of interest equals to 
\begin{align*}
L_n \11\set{n< \Tleps -1}&=  (Z_n+S_n) \11\set{n< \Tleps -1}= \Big(\frac{H_n-\epsilon}{I} + \sum_{r=1}^{n} \frac{k_{1,r}}{I}\Big) \11\set{n< \Tleps -1}.
\end{align*}
Similarly,
\begin{align*}
 L_{n+1} \11\set{n< \Tleps -1} 
&= \Big(\frac{H_{n+1}-\epsilon}{I} + \sum_{r=1}^{n+1} \frac{k_{1,r}}{I}\Big)\11\set{n< \Tleps -1}. 
\end{align*}
As a result, the difference between $L_n$ and $L_{n+1}$ satisfies the following
\begin{align*}
\EE[(L_{n+1}- L_n)\11\{n< \Tleps -1\}|y^{t_n}] &= \EE[(L_{n+1}- L_n)\11\{n< \Tleps -1\}|y^{n}]\\\numberthis\label{eq:first subcase a}
&= \EE[(\frac{H_{n+1}-H_n}{I} + \frac{k_{1,n+1}}{I})\11\{n< \Tleps -1\} | y^n ].
\end{align*}

Next, we consider the  $(L_{n+1}-L_n) \11\set{n = \Tleps -1}$.  In this case, we have $t_{n+1}=\Tueps$.
Consequently, the random processes equal to 
\begin{align*}
L_n \11\set{n = \Tleps -1} &= \Big(\frac{H_n-\epsilon}{I} + \sum_{r=1}^{n} \frac{k_{1,r}}{I}\Big) \11\set{n = \Tleps -1}, \\
L_{n+1}\11\set{n = \Tleps -1} &= (Z_{\Tueps} + S_{\Tueps})\11\set{n = \Tleps -1}\\ &=\Big((\frac{H_\Tueps-\epsilon}{I})\11\{H_\Tueps\geq \epsilon\}+\big(\frac{\log H_\Tueps - \log \epsilon}{D}+f(\log \frac{H_\Tueps}{\epsilon}) \big) \11\{H_\Tueps<\epsilon\}\\
&\hspace{40pt} + \sum_{r=1}^{\Tleps} \frac{k_{1,r}}{I} + \sum_{r=\Tleps+1}^{\Tueps} \frac{k_4}{I} \11\{H_{r-1} \geq \addvb{\sqrt{\epsilon}}\} +\addvb{\sqrt{\epsilon}} \frac{N}{I}\Big) \11\set{n = \Tleps -1}. 
 \end{align*} 
 Note that one cannot be sure as to whether $H_{\Tueps}$ is less than $\epsilon$ or not. The reason is that $\Tueps$ is pruned by $N$ as in \eqref{eq:Tueps}. Thus, $H_{\Tueps}$ can be greater than $\epsilon$ when $\Tueps = N$. Therefore, We proceed by bounding $Z_{\Tueps}$. Note that, for small enough $\lambda$ the following inequality holds 
\begin{equation}\label{eq:fy 1}
\frac{\epsilon}{I}(e^{y}-1)-\frac{y}{D} < f(y), \qquad   \forall y <0 .
\end{equation}
The reason is that  $\lim_{\lambda \rightarrow 0} f(y) = \lim_{\lambda \rightarrow 0} \frac{1-e^{\lambda y}}{\lambda D} = \frac{-y}{D}$. This is greater than the left-hand side of \eqref{eq:fy 1} as $y<0$. 

Now, using inequality \eqref{eq:fy 1} with $y = \log \frac{H_\Tueps}{\epsilon}$, we obtain the following bound:
 \begin{align}\label{eq:Zt bound}
 \big(\frac{\log H_\Tueps - \log \epsilon}{D}+f(\log \frac{H_\Tueps}{\epsilon}) \big) \11\{H_\Tueps<\epsilon\}\geq \big(\frac{H_\Tueps-\epsilon}{I} \big)\11\{H_\Tueps<\epsilon\}.
   \end{align}

Consequently, the difference $L_{n+1}-L_n$ satisfies the following 
\begin{align*}
\EE[(L_{n+1}- L_n)\11\{n= \Tleps -1\}|y^{n}] &\geq \EE\bigg[ \Big(\frac{H_{\Tueps}-H_n }{I} +  \frac{k_{1,\Tleps}}{I} + \sum_{r=\Tleps+1}^{\Tueps} \frac{k_4}{I} \11\{H_{r-1} \geq \sqrt{\epsilon}\}\\\numberthis \label{eq:case c 1}
&\qquad +\sqrt{\epsilon}\frac{N}{I} \Big) \11\{n= \Tleps -1\}\Big| y^n \bigg].
\end{align*}
Next, we bound the first term above as  
\begin{align*}
H_{\Tueps}-H_n  &= H_{n+1}-H_n  + \sum_{r=n+2}^{\Tueps} (H_{r}-H_{r-1}), 
\end{align*}
where in the first equality, we add and subtract the intermediate terms $H_r, n+1\leq r \leq \Tueps-1$. Next,we substitute the above terms in the right-hand side of \eqref{eq:case c 1}. As $n+2 = \Tleps+1$, then the RHS of \eqref{eq:case c 1} equals the following:
\begin{align*}
&\frac{1}{I} \EE\Big[\Big( H_{n+1}-H_n  +  k_{1,\Tleps} \Big)\11\{n= \Tleps -1\} \Big| y^n \Big]\\\numberthis \label{eq:case c 2}
&\quad + \frac{1}{I}\EE\bigg[\Big(\sum_{r=\Tleps+1}^{\Tueps} \Big( {H_{r}-H_{r-1}}+k_4 \11\{H_{r-1} \geq \addvb{\sqrt{\epsilon}}\}\Big)+\addvb{\sqrt{\epsilon}}N\Big) \11\{n= \Tleps -1\}\Big| y^n \bigg].
\end{align*}
Now we combine the two sub-cases (i.e., \eqref{eq:first subcase a} and \eqref{eq:case c 2}) to obtain:
\begin{align*}
 \EE[L_{n+1}&- L_n|y^{n}] \geq  \frac{1}{I} \EE\Big[ H_{n+1}-H_n  +  k_{1,n+1}  \Big| y^n \Big]\\\numberthis \label{eq:case a residual}
& + \frac{1}{I}\EE\bigg[\Big(\sum_{r=\Tleps+1}^{\Tueps} \Big( {H_{r}-H_{r-1}}+k_4 \11\{H_{r-1} \geq \addvb{\sqrt{\epsilon}}\}\Big)+\addvb{\sqrt{\epsilon}}N\Big) \11\{n= \Tleps -1\}\Big| y^n \bigg].
\end{align*}

Note that using condition \eqref{eq:H linear K1}, we infer that the first term is non-negative. We work on the second term using the following chain of inequalities: 
\begin{align*}
& \frac{1}{I}\EE\bigg[\Big[\sum_{r=\Tleps+1}^{\Tueps} \Big((H_{r}-H_{r-1})+{k_4}\Big) \11\{H_{r-1} \geq \addvb{\sqrt{\epsilon}}\}\\
&\qquad + \Big( ({H_{r}-H_{r-1}}) \11\{H_{r-1} < \addvb{\sqrt{\epsilon}}\}\Big) +\addvb{\sqrt{\epsilon}} N \Big] \11\{n= \Tleps -1\}\Big| y^n \bigg]\\
&\stackrel{(a)}{\geq}  \frac{1}{I}\EE\left[\Big[\sum_{r=\Tleps+1}^{\Tueps} \Big( ({H_{r}-H_{r-1}}) \11\{H_{r-1} <\addvb{\sqrt{\epsilon}}\}\Big) +\addvb{\sqrt{\epsilon}} N \Big] \11\{n= \Tleps -1\}\Big| y^n \right]\\
&\stackrel{(b)}{\geq}  \frac{1}{I}\EE\left[\Big[ \Big(\sum_{r=\Tleps+1}^{\Tueps} -H_{r-1} \11\{H_{r-1} < \addvb{\sqrt{\epsilon}}\}\Big) +\addvb{\sqrt{\epsilon}} N\Big] \11\{n= \Tleps -1\}\Big| y^n \right]\\
&\stackrel{(c)}{>}  \frac{1}{I}\EE\left[\Big[ \Big(\sum_{r=\Tleps+1}^{\Tueps} -\addvb{\sqrt{\epsilon}} \Big) +\addvb{\sqrt{\epsilon}} N\Big] \11\{n= \Tleps -1\} \Big| y^n \right]\\
&\stackrel{(d)}{\geq}  \frac{1}{I}\EE\left[\Big[ \Big(\sum_{r=1}^{N} -\addvb{\sqrt{\epsilon}} \Big) +\addvb{\sqrt{\epsilon}} N\Big]\11\{n= \Tleps -1\} \Big| y^n \right] = 0,
\end{align*}
where (a) is due to \eqref{eq:H K4}, inequality (b) holds as $H_r\geq 0$, inequality (c) holds as  $H_{r-1} \11\{H_{r-1} < \sqrt{\epsilon}\}<\sqrt{\epsilon}$, and lastly (d) holds as $\Tueps \leq N$. To sum up, we proved that 
\begin{align*}
\EE[(L_{n+1}- L_n)|y^{n}] \geq 0.
\end{align*}

\noindent\textbf{Case (b). $\{y^{t_n}: n\geq  \Tleps\}$:} Note that if
$n<\Tueps$, then $t_n =t_{n+1} = \Tueps$. Thus, immediately, $L_{n+1} -L_n =0$ almost surely.    Otherwise, if $n\geq \Tueps$ and $\{\Tueps =N\}\medcup \{n\geq N\}$, then $t_n=t_{n+1}=N$ and hence   $L_{n+1} -L_n =0$. Therefore, it remains to consider the case that $\Tueps<N$ and $\Tueps\leq n<N$. Therefore, $t_n = n$ and $t_{n+1} = n+1$. Furthermore, as $n+1>n\geq \Tueps$ and $\Tueps<N$, then $H_n<\epsilon$ and $H_{n+1}<\epsilon$, implying that we are in the logarithmic drift. Therefore, we have
\begin{align*}
 L_n = Z_n+S_n & = \frac{\log H_n -\epsilon}{D} +f(\log \frac{H_n}{\epsilon})+S_n,\\
 L_{n+1} = Z_{n+1} +S_{n+1} & = \frac{\log H_{n+1} -\epsilon}{D} +f(\log \frac{H_{n+1}}{\epsilon})+S_{n+1}. 
 \end{align*} 
Hence, to sum up, in all of the above sub-cases in (b), the following holds:
\begin{align*}
L_{n+1}-L_n= \frac{\log H_{t_{n+1}} - \log H_{t_{n}}}{D} +f(\log \frac{H_{t_{n+1}}}{\epsilon})-f(\log \frac{H_{t_{n}}}{\epsilon})+S_{t_{n+1}}-S_{t_n}.
\end{align*}
Note that from \eqref{eq:H log K2}, one can check that the following inequality holds
in all sub-cases:
\begin{align*}
\EE\left[\frac{\log H_{t_{n+1}} - \log H_{t_{n}}}{D}+S_{t_{n+1}}-S_{t_n} \Big| y^{t_n}\right] \geq 0.
\end{align*}
Therefore, the difference $L_{n+1}-L_n$ satisfies the following 
\begin{align*}
\EE[(L_{n+1}- L_n)|y^{t_n}] & \geq  \EE\left[f(\log \frac{H_{t_{n+1}}}{\epsilon})-f(\log \frac{H_{t_{n}}}{\epsilon})\Big| y^{t_n}\right] .
\end{align*}
Next, we use the Taylor's theorem for $f$. We only need to consider the case that $\Tueps<N$ and $\Tueps\leq n<N$ implying that $t_n = n$ and $t_{n+1} = n+1$.   Using the Taylor's theorem we can write
\begin{align*}
f(\log \frac{H_{{n+1}}}{\epsilon}) = f(\log \frac{H_{{n}}}{\epsilon}) + \frac{\partial f}{\partial y}\Big|_{y = \log \frac{H_{{n}}}{\epsilon}} \big(\log \frac{H_{n+1}}{H_n} \big)+\frac{\partial^2 f}{\partial y^2}\Big|_{y = \zeta}(\log \frac{H_{n+1}}{H_n})^2,
\end{align*}
where $\zeta$ is between $\log \frac{H_{n+1}}{\epsilon}$ and $\log \frac{H_n}{\epsilon}$ and 
\begin{align*}
\frac{\partial f}{\partial y}\Big|_{y = \log \frac{H_{{n}}}{\epsilon}}  = -\frac{e^{\lambda \log \frac{H_n}{\epsilon}}}{D}, \qquad \frac{\partial^2 f}{\partial y^2}\Big|_{y = \zeta} = -\frac{\lambda }{D}e^{\lambda \zeta}.
\end{align*}
As a result, we have
\begin{align*}
\EE\bigg[f(\log \frac{H_{{n+1}}}{\epsilon})-f(\log & \frac{H_{{n}}}{\epsilon})\Big| y^{n}\bigg]  = \EE\left[-\frac{e^{\lambda \log \frac{H_n}{\epsilon}}}{D} \big(\log \frac{H_{n+1}}{H_n} \big)-\frac{\lambda }{D}e^{\lambda \zeta}(\log \frac{H_{n+1}}{H_n})^2\Big| y^n \right]\\
&= \EE\left[-\frac{e^{\lambda \log \frac{H_n}{\epsilon}}}{D} \big(\log \frac{H_{n+1}}{H_n} \big)-\frac{\lambda }{D}e^{\lambda ( \zeta \pm \log \frac{H_n}{\epsilon})}(\log \frac{H_{n+1}}{H_n})^2\Big| y^n \right]\\
&\stackrel{(a)}{\geq}  \EE\left[-\frac{e^{\lambda \log \frac{H_n}{\epsilon}}}{D} \big(\log \frac{H_{n+1}}{H_n} \big)-\frac{\lambda }{D}e^{\lambda ( k_3 + \log \frac{H_n}{\epsilon})}(\log \frac{H_{n+1}}{H_n})^2\Big| y^n \right]\\
&\stackrel{(b)}{\geq}  \EE\left[-\frac{e^{\lambda \log \frac{H_n}{\epsilon}}}{D} k_3 -\frac{\lambda k_3^2}{D}e^{\lambda ( k_3 + \log \frac{H_n}{\epsilon})}\Big| y^n \right]\\
&=   \big( \frac{k_3}{D}  -\frac{\lambda k_3^2e^{\lambda k_3} }{D}\big) e^{\lambda \log \frac{H_n}{\epsilon}} \geq 0,
\end{align*}
where inequality (a) holds as $\abs{\zeta - \log \frac{H_n}{\epsilon}} \leq \abs{\log \frac{H_{n+1}}{\epsilon} - \log \frac{H_n}{\epsilon}}\leq k_3$ using \eqref{eq:H log K3}. Again, (b) follows from \eqref{eq:H log K3}. The last inequality holds for sufficiently small $\lambda >0$. Consequently, for case (b), we have 
\begin{align*}
    \EE[(L_{n+1}- L_n)|y^{t_n}] \geq 0.
\end{align*}
Lastly, combining cases (a) and (b), we get the desired result.
\end{IEEEproof}

%% file: proofs/proof_lem_pruned_submartingale2.tex
We follow the same argument with the same  cases as in the proof of Lemma \ref{lem:pruned sub martingale 1} but with $\Tleps=\Tlepstld$ and $\Tueps=\Tuepstld$. Measurability of $L_n$ w.r.t $\mathcal{G}_{n}, n>0$ follows from the same argument as in the proof of Lemma \ref{lem:pruned sub martingale 1}. Below, we consider the corresponding cases. 

\noindent\textbf{Case (a). $\{y^{t_n}: n< \Tlepstld\}$:} This case consists of two sub-cases depending on whether $n< \Tlepstld -1$ or not. 

For the first event, $n< \Tlepstld -1$, the argument in this case is  the same as the first sub-case of (a) as in the proof of Lemma \ref{lem:pruned sub martingale 1} but with $H_n= \Hbar_n, n>0$. Particularly, we can show that
\begin{align*}\numberthis\label{eq:first subcase}
\EE[(L_{n+1}- L_n)\11\{n< \Tlepstld -1\}|y^{t_n}] = \EE[(\frac{\Hbar_{n+1}-\Hbar_n}{I} + \frac{k_{1,n+1}}{I})\11\{n< \Tlepstld -1\} | y^n ].
\end{align*}

Next, we consider the other sub-case with $(L_{n+1}-L_n)\11\set{n = \Tlepstld -1}$. 
Following the same argument as Case (b) in the proof of Lemma \ref{lem:pruned sub martingale 1}, we have:
\begin{align*}
L_n \11\set{n = \Tlepstld -1} &= \Big(\frac{\Hbar_n-\epsilon}{I} + \sum_{r=1}^{n} \frac{k_{1,r}}{I}\Big) \11\set{n = \Tlepstld -1}, \\
L_{n+1}&\11\set{n = \Tlepstld -1}= (Z_{\Tuepstld} + S_{\Tuepstld})\11\set{n = \Tlepstld -1}\\ &=\Big((\frac{\Hbar_\Tuepstld-\epsilon}{I})\11\{H_\Tuepstld\geq \epsilon\}+\big(\frac{\log \Htld_\Tuepstld - \log \epsilon}{D}+f(\log \frac{H_\Tuepstld}{\epsilon}) \big) \11\{\Htld_\Tuepstld<\epsilon\}\\
&\hspace{30pt} + \sum_{r=1}^{\Tlepstld} \frac{k_{1,r}}{I} + \sum_{r=\Tleps+1}^{\Tuepstld} \frac{k_4}{I} \11\{\Htld_{r-1} \geq \addvb{\sqrt{\epsilon}}\} + \addvb{\sqrt{\epsilon}} \frac{N}{I}\Big) \11\set{n = \Tlepstld -1}. 
 \end{align*} 
Next, using the inequality \eqref{eq:fy 1} with $y = \log \frac{\Htld_\Tuepstld}{\epsilon}$ we obtain  \begin{align*}
\big(\frac{\log \Htld_\Tuepstld - \log \epsilon}{D}+f(\log \frac{\Htld_\Tuepstld}{\epsilon}) \big) \11\{\Htld_\Tuepstld<\epsilon\} &\geq \big(\frac{\Htld_\Tuepstld-\epsilon}{I} \big)\11\{\Htld_\Tuepstld<\epsilon\}  \\
&\geq \big(\frac{\Hbar_\Tuepstld-\epsilon}{I} \big)\11\{\Htld_\Tuepstld<\epsilon\}.
   \end{align*} 
   where the last inequality is due to \eqref{eq:Hbar less Htld} implying that $\Htld_\Tuepstld \geq \Hbar_\Tuepstld$. With the above inequality, we have
   \begin{align*}
\EE[(L_{n+1}- L_n)\11\{n= \Tlepstld -1\}|y^{n}] &\geq \EE\bigg[ \Big(\frac{\Hbar_\Tuepstld - \Hbar_n }{I} +  \frac{k_{1,\Tlepstld}}{I} + \sum_{r=\Tlepstld+1}^{\Tuepstld} \frac{k_4}{I} \11\{\Htld_{r-1} \geq \sqrt{\epsilon}\}\\
&\qquad +\sqrt{\epsilon}\frac{N}{I} \Big) \11\{n= \Tlepstld -1\}\Big| y^n \bigg].
\end{align*}
Combining the two sub-cases and the same argument as in the proof of Lemma \ref{lem:pruned sub martingale 1}, we can show that $\EE[L_{n+1}- L_n|y^{t_n}] \geq 0$.
   
   \noindent\textbf{Case (b). $\{y^{t_n}: n\geq  \Tlepstld\}$:} The argument is the same as in Case (b) in the proof of Lemma \ref{lem:pruned sub martingale 1} but with $H_n = \Htld_n$. Hence, in this case, $\EE[L_{n+1}- L_n|y^{t_n}] \geq 0$.
   
Lastly, combining these cases we get the desired result.

%% file: proofs/proof_lem_logdrift.tex
\section{Proof of Lemma \ref{lem: log drift MAC FB} }\label{appx: proof of lem log drift}
\begin{IEEEproof}
We start with the case $i=1$. The proof for the other two cases follows from the same argument.  Define the following quantities
\begin{align*}
\mu_{1}(w)&=\prob{W_1=w| Y^{r-1}=y^{r-1}}\\
\mu_{1}(w, y_{r})&=\prob{W_1=w| Y^{r-1}=y^{r-1},Y_{r}=y_{r}}\\
Q^1_{w}(y_{r})&=\prob{Y_{r}=y_{r}|W_1=w, Y^{r-1}=y^{r-1}},
\end{align*}
where $w\in [1:M_1], y_{t+1}\in \mathcal{Y}$. 
 
Let $\Wstar \in [1: M_1]$ be the most likely message condition on $y^{r-1}$. That is $\wstar = \argmax_{w\in [1: M_1]}\mu_{1}(w)$.
First, we show that having $\Htld^1_{r-1} < \epsilon, \epsilon\in [0,1)$ implies that $\mu_1(\wstar)= 1-\eta_1(\epsilon)$ with $\eta_1$ being a function satisfying $\lim_{\epsilon \rightarrow 0} \eta_1(\epsilon) =0.$ The argument is as follows:
 
Using the grouping axiom we have
\begin{align}\label{eq:Htld eq 1}
\Htld^1_{r-1}=H(W_1|y^{r-1})=h_b( \mu_1(w^*_{1,r}))+(1-\mu_1(w^*_{1,r}))H(\hat{W}_1), 
\end{align}
where $\hat{W}_1$ is a random variable with probability distribution $P(\hat{W}_1=w)=\frac{\mu_1(w)}{1-\mu_1(w^*_{1,r})}, ~w\in [1:M_1], w\neq w^*_{1,r}$. Hence, having $\Htld^1_{r-1}\leq \epsilon$ implies that $h_b(\mu_1(\wstar))\leq \epsilon$.
Taking the inverse image of $h_b$ implies that either $\mu_1(\wstar)\geq 1-h_b^{-1}(\epsilon)$ or $\mu_1(\wstar)\leq h_b^{-1}(\epsilon)$, where $h_b^{-1}:[0,1]\rightarrow [0,\frac{1}{2}]$ is the lower-half inverse function of $h_b$.  We show that the second  case is not feasible. For this purpose, we show that the inequality $\mu_1(\wstar)\leq h_b^{-1}(\epsilon)$ implies that $\Htld^1_{r-1}\geq 1$ which is a contradiction  with the original assumption $\Htld^1_{r-1} \leq \epsilon < 1$. This statement is proved in the following proposition. With this argument, we conclude that $\Htld^1_{r-1}\leq \epsilon$ implies that $\mu_1(\wstar)\geq 1-\eta_1(\epsilon)$, where $\eta_1(\cdot) = h_b^{-1}(\cdot)$.

\begin{proposition}
Let $W$ be a random variable taking values from a finite set $\mathcal{W}$. Suppose that $P_W(w)\leq \frac{1}{2}$ for all $w\in \mathcal{W}$. Then $H(W)\geq 1$. 
\end{proposition}
\begin{proof}
The proof follows from an induction on $|\mathcal{W}|$. For $|\mathcal{W}|=2$ the condition in the statement implies that $W$ has uniform distribution and hence $H(W)=1$ trivially. Suppose the statement holds for $|\mathcal{W}|=n-1$. Then, for $|\mathcal{W}|=n$.  Sort elements of $\mathcal{W}$ in an descending order according to $P_W(w)$, from the most likely (denoted by $w_1$) to the least likely ($w_n$). If $P_W(w_n)=0$, then the statement holds trivially from the induction's hypothesis.  Suppose $P_W(w_n)>0$. In this case,  by the grouping axiom, we can reduce  $H(W)$ by increasing $P_W(w_1)$ and decreasing $P_W(w_n)$ so that $P_W(w_1)+P_W(w_n)$ remains constant. In that case, either $P_W(w_n)$ becomes zero or $P_W(w_1)$ reaches the limit $\frac{1}{2}$. The first case happens if $P_W(w_1)+P_W(w_n)\leq \frac{1}{2}$. For that, the statement $H(W)\geq 1$ follows from the induction's hypothesis, as there are only $(n-1)$ elements with non-zero probability.  It remains to consider the second case in which $P_W(w_1)=\frac{1}{2}$ and $P_W(w_n)>0$. Again, we can further reduce the entropy by increasing $P_W(w_2)$ and decreasing $P_W(w_n)$ while $P_W(w_2)+P_W(w_n)$ remains constant. Observe that $P_W(w_2)+P_W(w_n)\leq \frac{1}{2}$ as $\sum_{i=2}^n P_W(w_i) = 1-P_W(w_1)=\frac{1}{2}$ and $p_W(w_i)>0$.  Hence, after this redistribution process $P_W(w_n)$ becomes zero. Then, the statement $H(W)\geq 1$ follows from the induction's hypothesis, as there are only $(n-1)$ elements with non-zero probability.
\end{proof}


Next, we analyze the log-drift. Fix $\delta\in (0,1/2)$, and let $\mathcal{Y}_\delta :=\{y_r\in \mathcal{Y}: P(y_r|y^{r-1})>\delta\}.$ Then,

\begin{align*}
\EE[\log \Htld^1_{r-1}-\log \Htld^1_{r}|\mathcal{F}_{r-1}] &\leq \sum_{y_r\in \mathcal{Y}_\delta} P(y_r|y^{r-1}) \log \frac{\Htld^1_{r-1}}{\Htld^1_{r}(y_r)}+\delta  \sum_{y_r\notin \mathcal{Y}_\delta} \Big|\log \frac{\Htld^1_{r-1}}{\Htld^1_{r}(y_r)}\Big| \\
&\leq \sum_{y_r\in \mathcal{Y}_\delta} P(y_r|y^{r-1})\log \frac{\Htld^1_{r-1}}{\Htld^1_{r}(y_r)}+\delta  \eta |\mathcal{Y}|,
\end{align*}
where $\eta$ is given in Lemma \ref{rem: H_t+1 - H_tis bounded}.
We proceed  by applying Lemma 7 in \cite{Burnashev} stated below.
For any non-negative sequence of numbers $p_\ell,\mu_i$ and $\beta_{i,l}, \ell \in [1:L], i\in [1:N]$ the following inequality holds
\begin{align*}
\sum_{\ell=1}^L p_\ell \log\Big( \frac{\sum_{i=1}^N \mu_i}{\sum_{i=1}^N \beta_{i,\ell}}  \Big) \leq \max_{i} \sum_{\ell=1}^L p_{\ell} \log \frac{\mu_i}{\beta_{i,\ell}}.
\end{align*}

 By the definitions of $\Htld^1_{r-1}$ and  $\Htld^1_{r}$, we have that 
\begin{align*}
\EE[\log \Htld^1_{r-1}-\log \Htld^1_{r}|\mathcal{F}_{r-1}] &\leq \sum_{y_r\in \mathcal{Y}_\delta} P(y_r|y^{r-1}) \log\Big(\frac{-\sum_{w} \mu_1(w)\log \mu_1(w) }{-\sum_{w}\mu_1(w,y_r)\log \mu_1(w,y_r)}\Big) +\delta \eta |\mathcal{Y}|\\\numberthis\label{eq:log drift gamma}
& \leq  \max_{w} \Gamma(w) +\delta \eta |\mathcal{Y}|,
\end{align*}
where 
\begin{align*}
\Gamma(w) = \sum_{y_r\in \mathcal{Y}_\delta} P(y_r|y^{r-1}) \log \frac{-\mu_1(w)\log \mu_1(w)}{-\mu_1(w,y_r)\log \mu_1(w,y_r)}.
\end{align*}

Note that 
\begin{align}\label{eq:mu wstar and y}
\mu_{1}(w, y_{r})&=\frac{\mu_1(w) Q_w(y_r)}{P(y_r|y^{r-1})}.
\end{align}
Therefore, for a fixed $w\neq \wstar$ we have that
\begin{align}\nonumber
-\Gamma(w) &= \sum_{y_r\in \mathcal{Y}_\delta} P(y_r|y^{r-1}) \log\bigg[ \frac{Q_w(y_r)}{ P(y_r|y^{r-1})}\Big(1+\frac{\log\frac{P(y_r|y^{r-1})}{Q_w(y_r)}}{-\log \mu_1(w)}\Big)\bigg]\\\label{eq:gamma w 1}
&=\sum_{y_r\in \mathcal{Y}_\delta} P(y_r|y^{r-1}) \bigg( \log \frac{Q_w(y_r)}{ P(y_r|y^{r-1})}+ \log\Big(1+\frac{\log\frac{P(y_r|y^{r-1})}{Q_w(y_r)}}{-\log \mu_1(w)}\Big)\bigg)
\end{align}
We bound the first term above:
\begin{align*}
&\sum_{y_r\in\mathcal{Y}} P(y_r|y^{r-1}) \log \frac{Q_w(y_r)}{ P(y_r|y^{r-1})}-\sum_{y_r\notin \mathcal{Y}_\delta} P(y_r|y^{r-1}) \log \frac{Q_w(y_r)}{ P(y_r|y^{r-1})}\\
&= D_{KL}\Big( P(\cdot|y^{r-1})~\|~ Q_w \Big) +\sum_{y_r\notin \mathcal{Y}_\delta} P(y_r|y^{r-1}) \log \frac{ P(y_r|y^{r-1})}{Q_w(y_r)}\\
&\geq D_{KL}\Big( P(\cdot|y^{r-1})~\|~ Q_w \Big) +\sum_{y_r\notin \mathcal{Y}_\delta} P(y_r|y^{r-1}) \log  P(y_r|y^{r-1})\\
&\geq D_{KL}\Big( P(\cdot|y^{r-1})~\|~ Q_w \Big) - |\mathcal{Y}| h_b(\delta),
\end{align*}
where for the last inequality we used the facts that  $p\log p \geq -h_b(p),$ and $h_b(p)\leq h_b(\delta)$ for $p\leq\delta$, and $|\mathcal{Y}_\delta|\leq |\mathcal{Y}|$.

Next, we bound the second term in \eqref{eq:gamma w 1}. The term in the logarithm is bounded as
\begin{align*}
\frac{\log\frac{P(y_r|y^{r-1})}{Q_w(y_r)}}{-\log \mu_1(w)} \geq \frac{\log P(y_r|y^{r-1}) }{-\log \mu_1(w)} > \frac{\log \delta }{-\log \mu_1(w)},
\end{align*}
where we used the fact  that $P(y_r|y^{r-1})>\delta$ as $y_r\in \mathcal{Y}_\delta$.  Therefore,  the second term in \eqref{eq:gamma w 1} is bounded as
\begin{align*}
\sum_{y_r\in \mathcal{Y}_\delta} P(y_r|y^{r-1}) \log (1+ \frac{\log \delta }{-\log \mu_1(w)})\geq \delta |\mathcal{Y}_\delta| \log (1+ \frac{\log \delta }{-\log \mu_1(w)}).
\end{align*}
We select $\delta = \sqrt{h_b^{-1}(\epsilon)}$. Since,  $\mu_1(w)\leq h_b^{-1}(\epsilon)$ for all $w\neq \wstar$, then the RHS is bounded from below as   $-|\mathcal{Y}_\delta|\sqrt{h_b^{-1}(\epsilon)}\geq -|\mathcal{Y}|\sqrt{h_b^{-1}(\epsilon)}$.

Combining the bounds for each term in \eqref{eq:gamma w 1} gives, for all $w\neq \wstar$,
\begin{align}\label{eq:gamma w}
\Gamma(w) \leq D_{KL}\Big( P(\cdot|y^{r-1})~\|~ Q_w \Big) +|\mathcal{Y}|\Big( h_b\big(\sqrt{h_b^{-1}(\epsilon)}\big) +  \sqrt{h_b^{-1}(\epsilon)}\Big). 
\end{align}
Clearly the residual terms converge to zero as $\epsilon\rightarrow 0$.

Next, we consider the case $w=\wstar$. We use the Taylor's theorem for the function $f(x) = -x\log x$ around $x=1$. Hence, $f(x) = {(\log e)}\big(  (1-x) - \zeta (1-x)^2\big)$ for some $\zeta$ between $x$ and $1$. Hence, with $x= \mu_1(\Wstar)$, we have that 
\begin{align*}
-\mu_1(\Wstar) \log \mu_1(\Wstar) 
\leq (\log e) (1-\mu_1(\wstar)). 
\end{align*}
Next, from the inequality $\log x \leq (\log e)(x-1), \forall x>0$, we have that 
\begin{align*}
-\mu_1(\Wstar,y_r) \log \mu_1(\Wstar,y_r) \geq (\log e)\mu_1(\Wstar,y_r) (1- \mu_1(\Wstar,y_r)).
\end{align*}
As a result of these inequalities, we have that 
\begin{align*}
\Gamma(\wstar) &= \sum_{y_r} P(y_r|y^{r-1}) \log \frac{-\mu_1(\wstar)\log \mu_1(\wstar)}{-\mu_1(\wstar,y_r)\log \mu_1(\wstar,y_r)}\\
&\leq \sum_{y_r} P(y_r|y^{r-1}) \log \frac{(1-\mu_1(\wstar))}{\mu_1(\Wstar,y_r) (1- \mu_1(\Wstar,y_r))}\\\numberthis \label{eq:wstar eq1}
&=  \sum_{y_r} P(y_r|y^{r-1}) \log \frac{1-\mu_1(\wstar)}{1- \mu_1(\Wstar,y_r)}  - \sum_{y_r} P(y_r|y^{r-1}) \log \mu_1(\Wstar,y_r). 
\end{align*}

We proceed with simplifying the first summation above. Using \eqref{eq:mu wstar and y}, we have that 
\begin{align}\label{eq:1-mu1 approximate}
(1-\mu_{1}(\Wstar, y_{r}))&=(1-\mu_1(\Wstar)) \frac{\sum_{w\notin \Wstar} \frac{\mu_1(w)}{1-\mu_1(\Wstar)}   Q^1_{w}(y_{r})}{P(y_r|y^{r-1})}.
\end{align}
Thus, the first summation on the right-hand side of \eqref{eq:wstar eq1} is simplified as 
\begin{align*}
\sum_{y_r} P(y_r|y^{r-1}) \log \frac{P(y_r|y^{r-1})}{\sum_{w\notin \Wstar} \frac{\mu_1(w)}{1-\mu_1(\Wstar)}   Q^1_{w}(y_{r})} = D_{KL}\Big( P(\cdot|y^{r-1})~\|~ Q^1_{\sim \wstar} \Big),
\end{align*}
where $Q^1_{\sim \wstar}(y_r) \deq \sum_{w\notin \Wstar} \frac{\mu_1(w)}{1-\mu_1(\Wstar)}   Q^1_{w}(y_{r})$ for all $y_r\in \mathcal{Y}$.
Next, we bound the second summation in \eqref{eq:wstar eq1}. Using \eqref{eq:mu wstar and y}, we have that 

\begin{align*}
- \sum_{y_r} P(y_r|y^{r-1}) \log \mu_1(\Wstar,y_r)  &= -\log \mu_1(\wstar) - \sum_{y_r} P(y_r|y^{r-1}) \log \frac{Q^1_{\Wstar}(y_{r})}{P(y_r|y^{r-1}) }\\
 &= -\log \mu_1(\wstar) + D_{KL}\Big( P(\cdot|y^{r-1})~\|~ Q^1_{\wstar} \Big)\\\numberthis\label{eq:wstar eq2}
 &\leq 2 (\log e) \eta_1(\epsilon)+ D_{KL}\Big( P(\cdot|y^{r-1})~\|~ Q^1_{\wstar} \Big),
\end{align*}
where the last inequality holds from the fact that $\log x \geq (\log e) (1- \frac{1}{x})$ and that $\mu_1(\wstar)\geq 1-\eta_1(\epsilon)$, implying $$-\log \mu_1(\wstar)\leq -\log (1-\eta_1(\epsilon)) \leq (\log e) \frac{\eta_1(\epsilon)}{1-\eta_1(\epsilon)}\leq 2(\log e)\eta_1(\epsilon),$$ which holds as $\eta_1(\epsilon)\leq \frac{1}{2}$. Note that $P(\cdot|y^{r-1}) = \mu_1(\wstar)Q^1_{\wstar} + (1-\mu_1(\wstar))Q^1_{\sim \wstar}$. Therefore, from the convexity of the relative entropy, the right-hand side of \eqref{eq:wstar eq2} is bounded above by 
\begin{align*}
&~ 2(\log e)\eta_1(\epsilon) + \mu_1(\wstar) D_{KL}\Big( Q^1_{\wstar} ~\|~ Q^1_{\wstar} \Big)+(1-\mu_1(\wstar)) D_{KL}\Big( Q^1_{\sim \wstar} ~\|~ Q^1_{\wstar} \Big)\\
&\leq 2(\log e)\eta_1(\epsilon)+ \eta_1(\epsilon) D_{KL}\Big( Q^1_{\sim \wstar} ~\|~ Q^1_{\wstar} \Big)\\
& \leq (2\log e+d^{\max}) \eta_1(\epsilon),
\end{align*}
where $d^{\max}$ is the maximum of the above relative entropy for all $r\geq 1$. As a result of the above argument, we have that 
\begin{align*}
\Gamma(\wstar) &\leq D_{KL}\Big( P(\cdot|y^{r-1})~\|~ Q^1_{\sim \wstar} \Big) + (2\log e +d^{\max}) \eta_1(\epsilon)\\
&\stackrel{(a)}{\leq} \sum_{w\notin \Wstar} \frac{\mu_1(w)}{1-\mu_1(\Wstar)}  D_{KL}\Big( P(\cdot|y^{r-1})~\|~ Q^1_{w} \Big) + (2\log e +d^{\max}) \eta_1(\epsilon)\\\numberthis\label{eq:gamma wstar}
&\stackrel{(b)}{\leq} \max_{w\neq \wstar}D_{KL}\Big( P(\cdot|y^{r-1})~\|~ Q^1_{w} \Big) + (2\log e +d^{\max}) \eta_1(\epsilon),
\end{align*}
where (a) is due to the convexity of the relative entropy on the second argument and the definition of $Q_{\sim \wstar}$. Inequality (b) follows as $\frac{\mu_1(w)}{1-\mu_1(\Wstar)}$ form a probability distribution on $w\neq \wstar$.

Combining \eqref{eq:gamma w} and \eqref{eq:gamma wstar} and \eqref{eq:log drift gamma} gives the following bound 
\begin{align}\label{eq:log drift KL 1}
\EE[\log \Htld^1_{r-1}-\log \Htld^1_{r}|\mathcal{F}_{r-1}] &\leq   \max_{w'_1} D_{KL}\Big( P(\cdot|y^{r-1})~\|~ Q^1_{w'_1} \Big) + \kappa(\epsilon),
\end{align}
where 
\begin{align}\label{eq:kappa}
   \kappa(\epsilon) \deq (2+d^{\max}) h_b^{-1}(\epsilon) + |\mathcal{Y}|\Big( h_b\big(\sqrt{h_b^{-1}(\epsilon)}\big) +  (1+\eta) \sqrt{h_b^{-1}(\epsilon)}\Big).
\end{align}
Observe that $\lim_{\epsilon \rightarrow 0}\kappa(\epsilon)=0$. Since $h_b(p)\leq 2\sqrt{p(1-p)}\leq 2\sqrt{p}$, then 
\begin{align*}
   \kappa(\epsilon) &\leq  (2\log e+d^{\max}) h_b^{-1}(\epsilon) + |\mathcal{Y}|(3+\eta)\sqrt{h_b^{-1}(\epsilon)}\\
   &\leq \Big(2\log e+d^{\max} + |\mathcal{Y}|(3+\eta)\Big) \sqrt{h_b^{-1}(\epsilon)}=O(\sqrt{h_b^{-1}(\epsilon)}).
\end{align*}
Note that the right-hand side of \eqref{eq:log drift KL 1} depends on the coding scheme. In what follows, we remove this dependency. 

In what follows, we bound each relative entropy from above. Note that for each $w'_1$
\begin{align*}
    Q^1_{w'_1}(y_r) &= \prob{Y_r = y_r \big|~Y^{r-1}= y^{r-1}, W_1= w'_1 }\\
    &=\sum_{x'_1\in \mathcal{X}_1} \sum_{x'_2\in\mathcal{X}_2} \prob{X_{1,r} = x'_1, X_{2,r} = x'_2, Y_r = y_r \big|~y^{r-1}, w'_1 }\\
      &= \sum_{x'_2\in\mathcal{X}_2} \prob{X_{2,r} = x'_2\big|~Y^{r-1}= y^{r-1}, w_1}Q\Big(y_r\big|x'_1, x'_2\Big).
\end{align*}
Hence, the convexity of relative entropy implies that 
\begin{align*}
    D_{KL}\Big( P(\cdot|y^{r-1})~\|~ Q^1_{w'_1} \Big) &\leq \sum_{x'_2\in\mathcal{X}_2} \prob{X_{2,r} = x'_2\big|~Y^{r-1}= y^{r-1}, w_1}\\
    &\qquad \times D_{KL}\Big( P(\cdot|y^{r-1})~\|~ Q(\cdot|x'_1, x'_2) \Big)\\\numberthis \label{eq:DKL bound 1}
    &\leq \max_{x'_1, x'_2} D_{KL}\Big( P(\cdot|y^{r-1})~\|~ Q(\cdot|x'_1, x'_2) \Big).
\end{align*}
Let $\Qbar{1}_r = P_{Y_r|X_{1,r}, Y^{r-1}}$, which is the effective channel from the first user's perspective at time $r$.  Note that 
\begin{align*}
    P(y_r|y^{r-1}) &= \sum_{x_1\in \mathcal{X}_1} \prob{X_{1,r} = x_1, Y_r = y_r \big|~y^{r-1}}\\
    &=\sum_{x_1\in \mathcal{X}_1} \prob{X_{1,r} = x_1\big|~y^{r-1}} \Qbar{1}_r(y_r| x_1, y^{r-1}).
\end{align*}
Denote $\nu_1(x_1)\deq \prob{X_{1,r} = x_1\big|~y^{r-1}}$. Hence the convexity of relative entropy gives
\begin{align*}
     D_{KL}\Big( P(\cdot|y^{r-1})~\|~ Q(\cdot|x'_1, x'_2) \Big) &\leq \sum_{x_1\in \mathcal{X}_1} \nu_1(x_1)  D_{KL}\Big( \Qbar{1}_r(\cdot | x_1, y^{r-1})~\|~ Q(\cdot|x'_1, x'_2) \Big).   
\end{align*}
Let $\Xstar  = e_1(\Wstar, y^{r-1}).$ 
Note that $\nu_1(\xstar)\geq \mu_1(\wstar)\geq 1-\eta_1(\epsilon)$. Hence, the above term is bounded from above as 
\begin{align} \label{eq:DKL bound 2}
     D_{KL}\Big( P(\cdot|y^{r-1})~\|~ Q(\cdot|x'_1, x'_2) \Big) &\leq D_{KL}\Big( \Qbar{1}_r(\cdot | \xstar, y^{r-1})~\|~ Q(\cdot|x'_1, x'_2) \Big) +\eta_1(\epsilon) d^{\max}.  
\end{align}
Lastly, combining \eqref{eq:log drift KL 1}, \eqref{eq:DKL bound 1}, and \eqref{eq:DKL bound 2} gives the following bound on the log-drift
\begin{align*}
    \EE[\log \Htld^1_{r-1}&-\log \Htld^1_{r}|\mathcal{F}_{r-1}] \leq   \max_{x'_1, x'_2} D_{KL}\Big( \Qbar{1}_r(\cdot | \xstar, y^{r-1})~\|~ Q(\cdot|x'_1, x'_2) \Big) +\eta_1(\epsilon) d^{\max} + \kappa(\epsilon)\\
   & \leq   \max_{x'_1}\sup_{P_{X'_2|Y^{r-1}}} D_{KL}\Big( \Qbar{1}_r(\cdot | \xstar, y^{r-1})~\|~ \Qbarp{1}_r(\cdot | x'_1, y^{r-1}) \Big) +\eta_1(\epsilon) d^{\max} + \kappa(\epsilon),
\end{align*}
where $\Qbarp{1}_r$ is the same as $\Qbar{1}_r$ but with $P_{X'_2|Y^{r-1}}$ at user 2.

\end{IEEEproof}

%% file: proofs/proof_lem_ErrExp_connection.tex
\begin{IEEEproof}
Since $\{L^i_n\}_{n >0}$ is a sub-martingale w.r.t $\mathcal{F}_{t^i_n}, n>0$ then, $L^i_0\leq \EE[L^i_{T\vee \Tuepstldi}]$, where $T$ is the stopping time used in the VLC and $\Tuepstldi$ is as in \eqref{eq:Tueps} but with $H_t = \Htld^i_t$. 
Note that $L^i_0 = \frac{\log M_i}{I}$. Next, we analyze $\EE[L^i_{T\vee \Tuepstldi}]$. By definition $L^i_n = Z^i_{t^i_n}+S^i_{t^i_n}$. Since, $T\leq N$, then from \eqref{eq:tn} we have that $t^i_{T\vee \Tuepstldi} = (T\vee \Tuepstldi) \vee \Tuepstldi=T\vee\Tuepstldi$. Therefore,  we have that
\begin{align*}
\EE[L^i_{T\vee \Tuepstldi}] & = \EE[Z^i_{T\vee \Tuepstldi}+S^i_{T\vee \Tuepstldi}]\\
&\stackrel{(a)}{\leq} \EE\left[\frac{\Hbar^i_{T\vee \Tuepstldi}+\epsilon}{I^i}\right]+\EE\left[\frac{\log \Htld^i_{T\vee \Tuepstldi} - \log \epsilon}{D^i}+f(\log \frac{\Htld^i_{T\vee \Tuepstldi}}{\epsilon})\right] + \EE\big[S^i_{T\vee \Tuepstldi}\big]\\
&\stackrel{(b)}{\leq} \EE\left[\frac{\Hbar^i_{T\vee \Tuepstldi}+\epsilon}{I^i}\right]+\EE\left[\frac{\log \Htld^i_{T\vee \Tuepstldi} - \log \epsilon}{D^i}\right]+\frac{1}{\lambda D^i} + \EE\big[S^i_{T\vee \Tuepstldi}\big]\\
& \stackrel{(c)}{\leq}\frac{\EE\big[\Hbar^i_{T\vee \Tuepstldi}\big]+\epsilon}{I^i} +\frac{\log \EE\big[\Htld^i_{T\vee \Tuepstldi}\big] - \log \epsilon}{D^i}+\frac{1}{\lambda D^i} + \EE\big[S^i_{T\vee \Tuepstldi}\big]\\\numberthis \label{eq:error exp up 1 }
& \stackrel{(d)}{\leq}\frac{\EE\big[\Hbar^i_{T}\big]+\epsilon}{I^i} +\frac{\log \EE\big[\Htld^i_{T}\big] - \log \epsilon}{D^i}+\frac{1}{\lambda D^i} + \EE\big[S^i_{T\vee \Tuepstldi}\big],
\end{align*}
where $(a)$ follows by changing $-\epsilon$ to $+\epsilon$ for the linear part and from the following inequality for the logarithmic part
\begin{align*}
(\log x - \epsilon) \11\{x < \epsilon\} \leq (\log x - \epsilon) \11\{x < \epsilon\} + (\log x - \epsilon) \11\{x \geq  \epsilon\} = \log x - \epsilon,
\end{align*}
inequality (b) holds as $f_i(y)\leq \frac{1}{\lambda D^i}$, inequality (c) follows from Jensen's inequality and concavity of $\log(x)$. Lastly, inequality (d) holds because of the following argument
\begin{align*}
\Hbar^i_{T\vee \Tuepstldi} = H(W_i | Y^{T\vee \Tuepstldi}) \leq  H(W_i | Y^{T}),
\end{align*}
where the last inequality holds because conditioning reduces the entropy and that $Y^T$ is a function of $Y^{T\vee \Tuepstldi}$.  Next, we bound  $\EE\big[\Hbar^i_{T}\big]$ and $\EE\big[\Htld^i_{T}\big]$. For that we use Fano's inequality to get the following inequality 
\begin{align}\label{eq:Fanos}
\EE\big[\Hbar^i_{T}\big] \leq \EE\big[\Htld^i_{T}\big]  = \EE[H(W_i | Y^T)] \leq \EE[H(W_1, W_2 | Y^T)] \leq h_b(P_e) + P_e \log M_1 M_2,
\end{align}
where $P_e$ is the probability of error. Define $\alpha(P_e)=h_b(P_e) + P_e \log(M_1M_2)$.  Therefore, from \eqref{eq:error exp up 1 } and \eqref{eq:Fanos}, we obtain that
\begin{align*}
\frac{\log M_i}{I} &\leq \frac{\alpha(P_e)+\epsilon}{I^i} +\frac{\log \alpha(P_e) - \log \epsilon}{D^i}+\frac{1}{\lambda D^i} + \EE\big[S^i_{T\vee \Tuepstldi}\big].
\end{align*}
Rearranging the terms gives the following inequality
\begin{align*}
\frac{-\log \alpha(P_e)}{D^i} \leq \frac{\alpha(P_e)+\epsilon}{I^i} +\frac{- \log \epsilon}{D^i}+\frac{1}{\lambda D^i} + \EE\big[S^i_{T\vee \Tuepstldi}\big] - \frac{\log M_i}{I}.
\end{align*}
Therefore, multiplying by $D^i$ and dividing by $\EE[T]$ gives the following
\begin{align*}
\frac{-\log \alpha(P_e)}{\EE[T]} \leq D^i \Big( \frac{ \EE\big[S^i_{T\vee \Tuepstldi}\big]}{\EE[T]} - \frac{R_i}{I^i}\Big)+U_i(P_e, M_i, \epsilon)
\end{align*}
where we used the fact that $\frac{\log M_i}{\EE[T]} \geq R_i$, and 
\begin{align}\label{eq:Ui}
U_i(P_e, M_i, \epsilon) =  R_i\Big( D^i \frac{\alpha(P_e)+\epsilon}{I^i \log M_i} +\frac{- \log \epsilon}{\log M_i}+\frac{1}{\lambda  \log M_i}\Big).
\end{align}
For the left hand side, we can write that 
\begin{align*}
-\log \alpha(P_e) &= 
-\log P_e -\log \Big(-\log P_e -\frac{(1-P_e)}{P_e} \log (1-P_e) + \log M_1M_2   \Big)\\
& \geq -\log P_e -\log \Big(-\log P_e -\frac{1}{P_e} \log (1-P_e) + \log M_1M_2 \Big)\\
& \geq -\log P_e -\log \Big(-\log P_e +2 + \log M_1M_2 \Big),
\end{align*}
where the last inequality follows because  $\log(x)\geq 1-\frac{1}{x}$ for $x>0$ implying that  $\log (1-P_e) \geq 1-\frac{1}{1-P_e} = \frac{-P_e}{1-P_e}$; and hence, $-\frac{1}{P_e} \log (1-P_e) \leq \frac{1}{1-P_e}\leq 2$ as $P_e \leq \frac{1}{2}$. Therefore, by factoring $-\log P_e$ we have that 
\begin{align*}
-\log \alpha(P_e) \geq (-\log P_e) ( 1 -\Delta ),
\end{align*}
where 
\begin{align}\label{eq:Delta}
\Delta = \frac{\log \big(-\log P_e +2+ \log M_1M_2 \big)}{-\log P_e}.
\end{align}
Using the above inequality, we get the following bound on the error exponent
\begin{align*}
\frac{-\log P_e}{\EE[T]} \leq \frac{1}{1-\Delta} D^i \Big( \frac{ \EE\big[S^i_{T\vee \Tuepstldi}\big]}{\EE[T]} - \frac{R_i}{I^i}\Big) + \frac{1}{1-\Delta} U_i(P_e, M_i, \epsilon), \qquad i=1,2,3.
\end{align*}
\end{IEEEproof}

%% file: proofs/Proof_lem_maxmartingale.tex
\section{Proof of Lemma \ref{lem:super martingale}}\label{app:super martingale}
\begin{proof}
Define $S\deq \inf\set{t>0: t\geq \tau, M_t >c}$. Note that $S$ is a stopping time. Since $\{M_t\}_{t>0}$ is non-negative, then for any fixed $n\in \NN$, we have that 
\begin{align*}
M_{n \wedge S} \geq c \11\Big\{\sup_{\tau \leq t \leq n} M_t >c\Big\}. 
\end{align*}
Therefore, taking the expectation of both sides and rearranging the terms gives the following inequality 
\begin{align}\label{eq:lem super martingale 1}
\PP\Big\{\sup_{\tau \leq t \leq n} M_t >c\Big\} \leq \frac{\EE[M_{n \wedge S}]}{c}.
\end{align}
Since $\{M_t\}_{t>0}$ is a super-martingale and that $\tau \leq S$, then $\EE[M_{n \wedge \tau}] \geq \EE[M_{n \wedge S}].$ Therefore, we can write 
\begin{align*}
\PP\Big\{\sup_{\tau \leq t \leq n} M_t >c\Big\} \leq \frac{\EE[M_{\tau}]}{c}.
 \end{align*} 
This is because if $n<\tau$ then the left-hand side is zero and the inequality holds trivially. When $n \geq \tau$, using the above argument, the right-hand side of \eqref{eq:lem super martingale 1} is less than $\frac{\EE[M_{n \wedge \tau}]}{c} = \frac{\EE[M_{\tau}]}{c}$.

Next, taking the limit $n\rightarrow \infty$ and from monotone convergence theorem we get that 
\begin{align*}
\PP\Big\{\sup_{\tau \leq t } M_t >c\Big\} &= \PP\Big\{\medcup_{n>0} \big\{\sup_{\tau \leq t \leq n} M_t >c\big\} \Big\}\\
&= \lim_{n\rightarrow \infty} \PP\Big\{\sup_{\tau \leq t \leq n} M_t >c\Big\}\leq \frac{\EE[M_{\tau}]}{c},
\end{align*}
where the second equality follows from the continuity of the probability measure. 
With that the proof is complete. 
\end{proof}

%% file: proofs/proof_three_phase_lb.tex
\section{Proof of Theorem \ref{thm: Error Exp lowerbound three phase}}\label{proof: Error Exp lowerbound three phase}
\begin{proof}
We first consider the variant of the coding scheme in which the first user finishes the data transmission sooner than the second user. 
At each block a re-transmission occurs  with probability $q$, an error occurs with probability $P_{eb}$ and a correct decoding process happens with probability $1-q-P_{eb}$. Let $\hat{\mathcal{H}}$ denote the decoders declared hypothesis. Note that $\mathcal{H}_1$ and $\mathcal{H}_0$ are the hypothesis that $\mathcal{H}_0: (\hat{W}_1, \hat{W}_2)=(W_1,W_2)$ and  $\mathcal{H}_1: (\hat{W}_1, \hat{W}_2)\neq (W_1,W_2)$, respectively.  The probability of a re-transmission at each block is 
\begin{align*}
q=P(\hat{\mathcal{H}} = \mathcal{H}_1).
\end{align*}
The probability of error at each block is
\begin{align*}
P_{eb}=P(\mathcal{H}_1)P(\hat{\mathcal{H}} = \mathcal{H}_0|\mathcal{H}_1),
\end{align*}
meaning that the decoder wrongfully declares $\mathcal{H}_0$ --- the no error hypothesis.
 Therefore, with this setting the total probability of error for the transmission of a message is 
\begin{align}\label{eq: P_e in terms of P_eb 2}
P_e=\sum_{k=0}^{\infty} q^k P_{eb}=\frac{P_{eb}}{1-q}.
\end{align}
 
The number of blocks required to complete the transmission of one message is a geometric random variable with probability of success $1-q$. Thus, the expected number of blocks for transmission of a message is $\frac{1}{1-q}$.


In what follows, we analyze $q$ and $P_{eb}$ for the proposed three-phase scheme. For shorthand,  denote $\Theta_{12}=(\Theta_1, \Theta_2), \hat{\Theta}_{12}=(\hat{\Theta}_1, \hat{\Theta}_2)$. Then
\begin{align}\label{eq:P_eb}
P_{eb}&=P\left(\hat{\Theta}_{12}=00, \Theta_{12}\neq 00\right)=\sum_{a\in \{01,10,11\}}P(\Theta_{12}=a)P(\hat{\Theta}_{12}=00|\Theta_{12}=a).
\end{align}
Note that, from a standard argument in channel coding, if the effective transmission rates are inside the capacity region, then $P(\Theta_{12}\neq 00)$ can be made sufficiently small. Observe that the effective rates of the three-phase scheme is $(\frac{R_1}{\gamma_1}, \frac{R_2}{\gamma_1+\gamma_2})$. We need to choose $\gamma_1, \gamma_2$ such that the rates are inside the feedback-capacity region of the channel. 
Let the rate-splitting of the second user be $R_2=R_{2,1}+R_{2,2}$, where $R_{21}$ and $R_{2,2}$ are the transmission rate during the first and second phases, respectively. During the first stage, we face a MAC channel with feedback. Hence, from Theorem \ref{thm:VLC Capacity} 
the effective rates during this phase must satisfy the following inequalities for some $L\in \NN$, and $P^L\in \PMAC^L$
\begin{subequations}\label{eq:rates in phase 1}
\begin{align}
\frac{R_1}{\gamma_1} &\leq \frac{1}{L} I(X_1^L \rightarrow Y^L \|~ X^L_2)+\varepsilon\\
\frac{R_{2,1}}{\gamma_1} &\leq \frac{1}{L} I(X^L_2 \rightarrow Y^L \|~ X^L_1)+\varepsilon\\
\frac{R_1+R_{2,1}}{\gamma_1} &\leq \frac{1}{L}I(X^L_1, X^L_2 \rightarrow Y^L)+\varepsilon.
\end{align}
\end{subequations}
For shorthand, let $\ILP{1}, \ILP{2}$ and $\ILP{3}$ denote, respectively, the normalized directed mutual information terms above, where $\Pdata$ denotes the multi-letter probability distribution for the first phase.

During the second phase, only the second user transmits the remaining of its messages. Hence, it faces a \ac{ptp} channel during the second phase. Therefore, from Lemma \ref{lem:rate err exp} and Remark \ref{rem:rate err exp}, the transmission rate of the second user during this stage must satisfy
\begin{align*}
\frac{R_{2,2}}{\gamma_2}\leq C_2 := I(X_2;Y|\xDHT{1}(0)),
\end{align*}
where $\xDHT{1}(0)$ and $\xDHT{1}(1)$ are the symbols that the first user sends for confirmation during the second phase. These symbols together with $P_{X_2}$ are chosen to maximize the rate-exponent region during the second phase. 
We set $R_{2,2} = \gamma_2 C_2$. Hence, $R_{2,1} = R_2 - \gamma_2C_2$. Replacing this quantity in \eqref{eq:rates in phase 1} implies that   
\begin{align*}
\frac{R_1}{\gamma_1} &\leq \ILP{1}+\varepsilon\\
\frac{R_{2}}{\gamma_1} &\leq \ILP{2}+\frac{\gamma_2}{\gamma_1}C_2+\varepsilon\\
\frac{R_1+R_{2}}{\gamma_1} &\leq \ILP{3}+\frac{\gamma_2}{\gamma_1}C_2+\varepsilon.
\end{align*}
Therefore, by rearranging the terms, the following condition must hold:
\begin{align}\label{eq:gamma1 and 2}
\gamma_1 \geq  \max\set{\frac{R_1}{\ILP{1}},\frac{R_2-\gamma_2 C_2}{\ILP{2}}, \frac{R_1+R_2-\gamma_2C_2}{\ILP{3}} }.
\end{align}
 As a result, with $\gamma_1$ and $\gamma_2$ satisfying the above inequality,  there exists a sequence $\zeta_n , n\geq 1$ with $\zeta_n \rightarrow 0$ such that after the error probability $\PP\set{\Theta_{12}=a}$ is  bounded by $\zeta_n$. Hence, we obtain the following upper bound
\begin{align} \label{eq: bound ofr P_e for confirmation stage}
P_{eb}&\leq  \sum_{a\in \{01,10,11\}} \zeta_n \prob{\hat{\Theta}_{12}=00|\Theta_{12}=a}.
\end{align}
In what follows, we bound each term inside \eqref{eq: bound ofr P_e for confirmation stage}. For  any $a_1, a_2 \in \set{0,1}$, let $\beta_{a_1,a_2}: = \prob{\hat{\Theta}_{12}=00|\Theta_{12}=(a_1, a_2)}$ which is the corresponding error term in \eqref{eq: bound ofr P_e for confirmation stage}. Then, we have that
\begin{align}\label{eq:betta a1a2}
\beta_{a_1a_2} &= \prob{\mathcal{A}|\Theta_{12}=a_1a_2}=\hspace{-15pt} \sum_{(y^{m_2},y^{m_3})\in \mathcal{A}} \bar{Q}(y^{m_2}| x_{12}^{m_2}(a_1)) Q^{m_3}(y^{m_3}| x^{m_3}_{13}(a_1), x^{m_3}_{23}(a_2)). 
\end{align}
\noindent\textbf{Decision Region:} Recall the definition of $\mathcal{A}$ in the main text. We proceed by simplifying the log-ratio terms in  $\mathcal{A}$. First, consider $a=01$. Then, the second log-ratio term in $\mathcal{A}$ 
is equivalent to the following
\begin{align*}
\frac{1}{m_3}&\log \frac{Q^{m_3}(y_{m_2+1}^{m}| \xHT{1}(0), \xHT{2}(0)) }{Q^{m_3}(y_{m_2+1}^{m} \xHT{1}(0), \xHT{2}(1))} = \frac{1}{m_3}\sum_{j=m_2+1}^m \log \frac{Q(y_j| x_{1,3,j}(0), x_{2,3, j}(0)) }{Q(y_j| x_{1,3, j}(0), x_{2,3,j}(1))}\\
& =  \sum_{u_1\in \mathcal{X}_1}\sum_{u_2,v_2\in \mathcal{X}_2} \mathsf{P}_m(u_1,u_2, *,v_2) \Big[ \sum_{s\in \mathcal{Y}} \hat{Q}_{\bfy}(s| u_1, u_2, v_2)   \log \frac{Q(s| u_1, u_2) }{Q(s| u_1, v_2)} \Big]\\
& =  \sum_{u_1\in \mathcal{X}_1}\sum_{u_2,v_2\in \mathcal{X}_2} \mathsf{P}_m(u_1,u_2, *,v_2) \Big[D\big( \hat{Q}_{\bfy}(\cdot | u_1, u_2, v_2)~\big\|~ Q(\cdot | u_1, v_2)\big)\\\numberthis \label{eq:decision reg1}
&\qquad  - D\big( \hat{Q}_{\bfy}(\cdot | u_1, u_2, v_2)~\big\|~ Q(\cdot | u_1, u_2)\big)\Big].
\end{align*}
Suppose ${\bfZ}= (Z_1(0), Z_2(0), Z_1(1), Z_2(1))$ are random variables with joint distribution $\mathsf{P}_n$. 
Then, the right-hand side of \eqref{eq:decision reg1} equals to 
\begin{align}\label{eq:decision KLD}
\frac{1}{m_3}\log \frac{Q^{m_3}(y_{m_2+1}^{m}| \xHT{1}(0), \xHT{2}(0)) }{Q^{m_3}(y_{m_2+1}^{m} \xHT{1}(0), \xHT{2}(1))} &= D\Big(\hat{Q}_{\bfy}(\cdot | {\bfZ})\big\| Q(\cdot | Z_1(0), Z_2(1)) \big| {\bfZ}\Big) \\
&\quad - D\Big(\hat{Q}_{\bfy}(\cdot | {\bfZ})\big\| Q(\cdot | Z_1(0), Z_2(0))\big| {\bfZ}\Big). \nonumber
\end{align}
We can derive similar expressions for other values of $a=10,11$.
Using a similar argument, the first log-ratio is simplified as
\begin{align*}
\log \frac{\bar{Q}^{m_2}(y^{m_2}|\xDHT{1}(0))}{\bar{Q}^{m_2}(y^{m_2}|\xDHT{1}(1))} &= m_2 \Big[ D(\hat{P}_{y^{m_2}}\| \bar{Q}(\xDHT{1}(1))) - D(\hat{P}_{y^{m_2}}\| \bar{Q}(\xDHT{1}(0)))\Big],
\end{align*}
where $\hat{P}_{y^{m_2}}$ is the empirical distribution of $y^{m_2}$ and $\bar{Q}$ is the averaged channel from the first user's perspective.
For any $a_1, a_2\in \set{0,1}$, define
\begin{align}\label{eq:Q lambda}
\mathcal{Q}_{\lambda}(a_1,a_2)&:= \Big\{\hat{Q}(y|\underline{\bfz}): D\Big(\hat{Q}(\cdot | \underline{\bfZ})~\big\|~ Q(\cdot | Z_1(a_1), Z_2(a_2)) \big| \underline{\bfZ}\Big)\\
&\qquad - D\Big(\hat{Q}(\cdot | \underline{\bfZ})~\big\|~ Q(\cdot | Z_1(0), Z_2(0))\big| \underline{\bfZ}\Big) \geq \frac{\lambda}{m_3}\Big\}. \nonumber
\end{align}

Also define
\begin{align*}
\mathcal{P}_{\lambda}(a_1):=\set{\hat{P}:  D(\hat{P}\| \bar{Q}(\xDHT{1}(a_1))) - D(\hat{P}\| \bar{Q}(\xDHT{1}(0))) \geq \frac{\lambda}{m}}.
\end{align*}
 Therefore, whether $y^m\in \mathcal{A}$ or not depends entirely on its corresponding empirical distributions ($\hat{P}_{y^{m_2}}, \hat{Q}_{\bfy}$). Such distributions must be such that the sum of the two log-ratio quantities in $\mathcal{A}$ are greater than $\lambda$. Then, given that $\frac{m_2}{m_2+m_3}=\frac{\gamma_2}{\gamma_2+\gamma_3}$, define:
\begin{align*}
\mathcal{D}_{\lambda}(a_1, a_2) :\set{(\hat{P}, \hat{Q}) : \hat{P}\in \mathcal{P}_{\alpha}(a_1), \hat{Q}\in \mathcal{Q}_{\beta}(a_1,a_2), \frac{\gamma_2}{\gamma_2+\gamma_3}\alpha + \frac{\gamma_3}{\gamma_2+\gamma_3} \beta \geq \lambda }.
\end{align*}

Then the decision region $\mathcal{A}$ is the set of all $\bfy$ such that $(\hat{P}_{\bfy},\hat{Q}_{\bfy})$ belongs to $\medcap_{\mathbf{a}\neq 00} \mathcal{D}_{\lambda}(\mathbf{a})$. For shorthand we drop the subscript $\lambda$ in $\mathcal{D}_{\lambda}$. Therefore, using a type-analysis argument, we can write
\begin{align*}
\beta_{a_1a_2} &=  
\hspace{-15pt}
\sum_{\substack{(\hat{P}, \hat{Q})\in \mathcal{D}(\bfa),\\ \bfa\neq 00}} \sum_{y^{m_2}\in T_{\delta}^{m_2}(\hat{P})}\hspace{-15pt}\bar{Q}(y^{m_2}| x_{12}^{m_2}(a_1)) \hspace{-10pt}\sum_{\underline{x}\in T[\mathsf{P}_{m_3}]}\hspace{-10pt}\mathsf{P}_{m_3}(\underline{x})\hspace{-20pt}\sum_{y^{m_3}\in T_{\delta}^{m_3}(\hat{Q}(\cdot|\underline{x}))}\hspace{-20pt}  Q^{m_3}(y^{m_3}| x^{m_3}_{13}(a_1), x^{m_3}_{23}(a_2))\\
&\leq \sup_{(\hat{P}, \hat{Q})\in \mathcal{D}(a_1, a_2)} \exp\Big\{-m_2D(\hat{P}\| \bar{Q}(x_{12}(a_1)))\\
&\hspace{90pt}-m_3\EE_{P_{m_3}} \big[D(\hat{Q}(\cdot | \underline{\bfX})~\|~ Q(\cdot | X_1(a_1), X_2(a_2)))\big]+\delta'\Big\}.
\end{align*}
Therefore, from Chernoff-Stein lemma, it is not difficult to see that
\begin{align*}
\beta_{a_1a_2} & \leq 2^{-n\gamma_2 D(\bar{Q}(x_{12}(0)) \| \bar{Q}(x_{12}(a_1)) }2^{-n\gamma_2 \bar{D}_{\mathsf{P}_{m_3}}(00||a_1, a_2)+\delta''},
\end{align*}
where we used the definition given in \eqref{eq:dpar ij}:
\begin{align*}
\bar{D}_{\mathsf{P}_{m_3}}(00||ij)=\EE_{\mathsf{P}_{m_3}}\Big[D_Q\big(X_1(0), X_2(0)|| X_1(i), X_2(j)\big)\Big], \quad \forall i,j \in \{0,1\} .
\end{align*}
Let $\DQbar{1}=D(\bar{Q}(x_{12}(0)) \| \bar{Q}(x_{12}(1))$. Then, the bound on $\beta_{a_1, a_2}$ simplifies to the following inequalities:
 \begin{align*}
 \beta_{01}&\leq 2^{-n\big(\gamma_3 \bar{D}_{\mathsf{P}_{m_3}}(00||01)+\delta''\big)},\\
 \beta_{10} &\leq 2^{-n \big(\gamma_2 \DQbar{1} + \gamma_3 \bar{D}_{\mathsf{P}_{m_3}}(00||10)+\delta''\big)},\\
 \beta_{11} &\leq 2^{-n \big(\gamma_2 \DQbar{1} + \gamma_3 \bar{D}_{\mathsf{P}_{m_3}}(00||11)+\delta''\big)}.
 \end{align*}

Therefore, combining such bounds and \eqref{eq: bound ofr P_e for confirmation stage}, we obtain the following 
\begin{equation}\label{eq: err exp lower bound gamma 3}
\frac{- \log P_e }{\EE[T]} \geq \min\set{\gamma_3 \DHT{01}, \gamma_2 \DQbar{1} \! +\! \gamma_3 \DHT{10}, \gamma_2\DQbar{1} \! + \! \gamma_3 \DHT{11}}.
\end{equation}
It remains to optimize the bound over the choices of all parameters: $\gamma_2, \gamma_3, \mathsf{P}_{m_3}, P_1$ and $\xDHT{1}(0), \xDHT{1}(1)$. Note that condition \eqref{eq:gamma1 and 2} must be satisfied. Given that $\gamma_1+\gamma_2+\gamma_3=1$, the condition \eqref{eq:gamma1 and 2} is equivalent to the following:
\begin{align*}
 \gamma_2+\gamma_3+ \max\set{\frac{R_1}{I_L^1(P^L)},\frac{R_2-\gamma_2 C_2}{I_L^2(P^L)}, \frac{R_1+R_2-\gamma_2C_2}{I_L^3(P^L)} }\leq 1.
\end{align*}
Clearly $\gamma_3^*$, the optimal $\gamma_3$, is a function of $\gamma_2$:
\begin{align*}
\gamma_3^* = 1-\gamma_2 - \max\set{\frac{R_1}{I_L^1(P^L)},\frac{R_2-\gamma_2 C_2}{I_L^2(P^L)}, \frac{R_1+R_2-\gamma_2C_2}{I_L^3(P^L)} }.
\end{align*}
Define
\begin{align*}
\Eo(r_1, r_2):=1- \min\set{\frac{r_1}{\ILP{1}},\frac{r_1}{\ILP{2}},\frac{r_1+r_2}{\ILP{3}}}.
\end{align*}
Then, we get the following lower bound:
\begin{align*}
\frac{- \log P_e }{\EE[T]} \geq \Elb^1(R_1,R_2),
\end{align*}
where
\begin{align*}
\Elb^1(R_1,R_2) =  \sup_{0\leq \gamma_2<1}&\sup_{\mathsf{P}_3} \sup_{\substack{P_1 \in \PMAC}} \min\Big\{
 \big(\Eo(R_1, R_2-\gamma_2C_2)-\gamma_2\big)\DHT{01},\\
&  \big(\Eo(R_1, R_2-\gamma_2C_2)-\gamma_2\big)\DHT{10}+\gamma_2\DQbar{1},\\
&  \big(\Eo(R_1, R_2-\gamma_2C_2)-\gamma_2\big)\DHT{10}+\gamma_2\DQbar{1}\Big\}
\end{align*}

Recall  that the proof was for the variant of the coding scheme where the first user finishes first. Next, we analyze the other variant with the second user finishing data transmission sooner. By symmetry, we obtain the following lower bound
\begin{align*}
\Elb^2(R_1,R_2) = \!\! \sup_{0\leq \gamma_2<1}&\!\!\sup_{\mathsf{P}_3} \!\!\sup_{\substack{P_1 \in \PMAC}} \! \! \!\!\! \min\Big\{
 \big(\Eo(R_1-\gamma_2C_1, R_2)-\gamma_2\big)\DHT{01}+\gamma_2\DQbar{2},\\
& \big(\Eo(R_1-\gamma_2C_1, R_2)-\gamma_2\big)\DHT{10},\\
& \big(\Eo(R_1-\gamma_2C_1, R_2)-\gamma_2\big)\DHT{10}+\gamma_2\DQbar{2}\Big\},
\end{align*}
where $C_1$ is the channel capacity seen at the first user, during the second phase of transmission.  Combining the two bound we obtain the desired result.

%
\end{proof}

%% file: proofs/proof_lem_lb_shape.tex
\section{Proof of Lemma \ref{lem:lb_shape}}\label{proof:lem:lb_shape}
\begin{IEEEproof}
The argument is similar to the proof of Theorem \ref{thm: Error Exp lowerbound three phase}. Here, the scheme is only two-phase: data transmission and a full  confirmation stage; implying that $R_{2,2}=0, \gamma_2=0$. Note that the probability of error is given by \eqref{eq: P_e in terms of P_eb 2}, where $P_{eb}$ is as in \eqref{eq:P_eb}:
\begin{align*}
P_{eb}&=P\left(\hat{\Theta}_{12}=00, \Theta_{12}\neq 00\right)=\sum_{a\in \{01,10,11\}}P(\Theta_{12}=a)P(\hat{\Theta}_{12}=00|\Theta_{12}=a).
\end{align*}
As argued, $P(\Theta_{12}=a$ can be made sufficiently small as long as the effective rates during the data transmission  $(R_1, R_2)$ are inside the feedback-capacity region. To ensure that, instead of the conditions in \eqref{eq:rates in phase 1}, we use the alternative expression of the capacity region given in Theorem \ref{thm:capacity C lambda}. Particularly, we require that 
\begin{align}\label{eq:gamma 1 condition}
\frac{1}{\gamma_1}( \lambda_1 R_1 + \lambda_2 R_2+\lambda_3(R_1+R_2)) \leq \mathcal{C}_{\underline{\lambda}},
\end{align}
for any $\lambda_1, \lambda_2, \lambda_3\geq 0$. Given that the rest of the argument is the same as in the proof of Theorem \ref{thm: Error Exp lowerbound three phase}. Particularly, we get the following bound which is the same as \eqref{eq: err exp lower bound gamma 3} but with $\gamma_2=0$ as stated in the beginning of this proof.
\begin{align*}
\frac{- \log P_e }{\EE[T]} \geq \gamma_3 \min\set{\DHT{01},\DHT{10},  \DHT{11}},
\end{align*}
If we optimize over the choice of $P_3$ in the scheme, the right-hand side of the above expression becomes $\gamma_3 D_{lb}$. 
Given that $\gamma_1+\gamma_3=1$, we optimize over all $\gamma_2$ that satisfy \eqref{eq:gamma 1 condition}. With that, the right-hand side becomes  $D_{lb} \Big( 1- \frac{\sum_{i=1}^3\lambda_i R_i}{\mathcal{C}_{\underline{\lambda}}}\Big)$.
Sine this inequality holds for any $\lambda_1, \lambda_2, \lambda_3\geq 0$, then we get the bound 
\begin{align}\label{eq: err exp lb lambda}
\frac{- \log P_e }{\EE[T]} \geq \sup_{\substack{\lambda_1, \lambda_2, \lambda_3\geq 0\\ \lambda_1+\lambda_2+\lambda_3=1}} D_{lb} \Big( 1- \frac{\sum_{i=1}^3\lambda_i R_i}{\mathcal{C}_{\underline{\lambda}}}\Big).
\end{align}
Next, we present this expression in polar coordination. By $(\| {R}\|, \theta_{ {R}})$ denote the polar coordinates of $(R_1,R_2)$ in $\RR^2$. 
First, note that the right-hand side of \eqref{eq: err exp lb lambda} equals to
\begin{align*}
 D_l\left( 1- \frac{\sum_{i=1}^3  \lambda^*_i R_i}{\mathcal{C}_{\underline{\lambda}^*}}\right),
\end{align*}
where $\underline{\lambda}^* = (\lambda_1^*,\lambda_2^*, \lambda_3^*)$ is the optimum choice. We show that given an arbitrary $\alpha>0$ and a rate pair $(R_1,R_2)$ in the capacity region, $\underline{\lambda}^* $  is the same as the one for $(\alpha R_1, \alpha R_2)$. To see this,  lets replace $(R_1, R_2)$ with $(\alpha R_1, \alpha R_2)$ for some constant $\alpha>0$. Then, the right-hand side of \eqref{eq: err exp lb lambda} for the new rates equals to
\begin{align*}
D_l \left( 1- \alpha \max_{\substack{\lambda_1, \lambda_2, \lambda_3 \geq 0\\ \lambda_1+\lambda_2+\lambda_3=1}} \frac{\sum_{i=1}^3  \lambda_i R_i}{C_{\underline{\lambda}}}\right) = D_l\left( 1- \alpha \frac{\sum_{i=1}^3  \lambda^*_i R_i}{C_{\underline{\lambda}^*}}\right),
\end{align*}
where the equality follows as $\alpha$ is a constant moving before the maximization over lambda. This means that the objective function for the maximization is the same as the previous one. This implies that there is an identical $\underline{\lambda}^*$ which optimizes the  expression for $(R_1,R_2)$ and $(\alpha R_1, \alpha R_2)$. Now, consider the line passing $(R_1,R_2)$ and the origin. Let $(R'_1, R'_2)$ denote the point of intersection of this line with the boundary of the capacity region. Fig. \ref{fig: R' and R} shows how $(R'_1, R'_2)$ is determined.  Since, $R'_i=\alpha R_i,i=1,2$ for some $\alpha>0$, then the optimum $\underline{\lambda}$ for $(R'_1, R'_2)$ is the same as the one for $(R_1, R_2)$. Therefore, from this argument and the fact that $R_i=\frac{R'_i}{\alpha}, i=1,2$, we can rewrite \eqref{eq: err exp lb lambda} as 
\begin{align*}
\frac{- \log P_e }{\EE[T]} \geq& \min_{\substack{\lambda_1, \lambda_2, \lambda_3 \geq 0\\ \lambda_1+\lambda_2+\lambda_3=1}}D_{lb}\left(1-\frac{1}{\alpha} \frac{\sum_{i=1}^3  \lambda_i R'_i}{C_{\underline{\lambda}}}\right)
\stackrel{(a)}{=}D_{lb}\left(1- \frac{1}{\alpha}\right) \stackrel{(b)}{=}D_{lb}\left(1- \frac{\|{R}\|}{\|{R'}\|}\right),
\end{align*}
where $(a)$ follows, since $(R'_1, R'_2)$ is on the capacity boundary, and (b) holds as $\alpha=\frac{\| {R'}\|}{\| {R}\|}$. Note that $\| {R}'\|$ depends on $(R_1,R_2)$ only through $\theta_{ {R}}$; in particular, it equals to $\mathcal{C}(\theta_{ {R}})$ which is a function of $\theta_{ {R}}$. With this notation, we get our desired lower bound as
$$\frac{- \log P_e }{\EE[T]} \geq D_{lb}\left(1- \frac{\|\underline{R}\|}{C(\theta_R)}\right).$$
\end{IEEEproof}